\documentclass[10pt]{article}

\usepackage{amssymb} 
\usepackage{amsmath, amsthm}
\usepackage{amsfonts}
\usepackage[utf8]{inputenc}
\usepackage{graphicx} 
\usepackage{color}
\usepackage{verbatim}
\usepackage{mathrsfs}
\usepackage[dvipsnames]{xcolor}

\usepackage[top=1in, bottom=0.6in, left=1.2in, right=1.2in,twoside]{geometry}
\usepackage[font=small,labelfont=bf]{caption}

\usepackage{setspace}
\usepackage{fancyhdr}
\usepackage{authblk}

\usepackage{hyperref}
\hypersetup{
  colorlinks=true,
  linkcolor=black,
  citecolor=black
}

\theoremstyle{plain}
\newtheorem{theorem}{Theorem}[section]
\newtheorem{proposition}[theorem]{Proposition}
\newtheorem{lemma}[theorem]{Lemma}
\newtheorem{corollary}[theorem]{Corollary}

\newenvironment{proofof}[1]{\noindent\emph{Proof of #1.\/}}{\hfill$\square$}
\newenvironment{sketchofproofof}[1]{\noindent\emph{Sketch of the proof of #1.\/}}{\hfill$\square$}

\begin{document}

\date{}
\title{\bf \Large A Variational Structure for
%Dynamical Fluctuations of 
 Interacting Particle Systems and their Hydrodynamic Scaling Limits}

\author[1]{\normalsize Marcus Kaiser}
\author[2,3,4]{Robert L. Jack}
\author[1]{Johannes Zimmer}
\affil[1]{Department of Mathematical Sciences, University of Bath, Bath BA2 7AY, UK}
\affil[2]{Department of Applied Mathematics and Theoretical Physics, University of Cambridge, Wilberforce Road, Cambridge, Cambridge CB3 0WA, UK}
\affil[3]{Department of Chemistry, University of Cambridge, Lensfield Road, Cambridge CB2 1EW, UK}
\affil[4]{Department of Physics, University of Bath, Bath BA2 7AY, UK}
\maketitle

\thispagestyle{empty}

\begin{abstract}
  We consider hydrodynamic scaling limits for a class of reversible interacting particle systems, which includes the symmetric
  simple exclusion process and certain zero-range processes.  We study a (non-quadratic) microscopic action functional for these
  systems. We analyse the behaviour of this functional in the hydrodynamic limit and we establish conditions under which it
  converges to the (quadratic) action functional of Macroscopic Fluctuation Theory.  We discuss the implications of these results
  for rigorous analysis of hydrodynamic limits.
\end{abstract}

\section{Introduction}

Recently, a \emph{canonical structure} has been introduced~\cite{Maes2008a,Maes2008b} to describe dynamical fluctuations in
stochastic systems.  The resulting theory has several attractive features: Firstly, it applies to a wide range of systems,
including finite-state Markov chains and Macroscopic Fluctuation Theory (MFT)~\cite{Bertini2015a}, see~\cite{Kaiser2018a}.
Secondly, it is based on an \emph{action functional} which is a relative entropy between probability measures on path spaces ---
this means that it provides a variational description of the systems under consideration, and the action can be related to large
deviation rate functionals. Thirdly, it extends the classical Onsager-Machlup theory~\cite{Onsager1953a} in a natural way, by
replacing the quadratic functionals that appear in that theory with a pair of convex but \emph{non-quadratic} Legendre duals
$\Psi$ and $\Psi^\star$. (This is sometimes called a $\Psi$-$\Psi^\star$ representation~\cite{Mielke2014a}.)  In Onsager-Machlup
theory and in MFT, the minimiser of the action describes the most probable evolution of a macroscopic system, either in terms of
thermodynamic forces and fluxes (in Onsager-Machlup theory) or densities and fluxes (in MFT): this feature is maintained in the
canonical structure.

This structure can be applied to any finite-state Markov chain and provides a unifying formulation of a wide range of
systems~\cite{Kaiser2018a}.  In particular, lattice systems of interacting particles can be described by canonical structures in
two ways: either on the microscopic (Markov chain) level via non-quadratic Legendre duals, or as a coarse-grained version through
the hydrodynamic limit, where the action reduces to a quadratic MFT functional. One therefore expects that in the hydrodynamic
scaling limit, the microscopic (non-quadratic) structure should converge (in some suitable sense) to the macroscopic one.  Such a
convergence would offer a new way to understand and derive hydrodynamic limits. The main question of this article is whether this
natural conjecture holds.

We give a partial (positive) answer, by proving several theorems that relate the microscopic and macroscopic action functionals
for interacting particle systems. Specifically, we consider a class of systems on periodic lattices with \emph{gradient dynamics}
and a conserved number of particles, which includes as special cases the symmetric simple exclusion process and a large class of
reversible zero-range processes.  In the hydrodynamic limit, the number of lattice sites and the number of particles go to
infinity together, at fixed density, and the microscopic transition rates have a parabolic scaling.  (These are among the simplest
models for which one can rigorously establish a hydrodynamic limit~\cite{Kipnis1999a}.)

Our analysis is based on the microscopic action, which is a relative entropy between two probability measures: one measure encodes
the dynamics of the particle system itself (the \emph{reference process}) and the other represents some other \emph{observed
  process}, which is to be compared with the reference process.  We consider observed processes that concentrate (in the
hydrodynamic limit) on deterministic paths. By comparing different processes, we can extract information about the hydrodynamic
limit of the reference process (if this limit exists).  That is, the reference process and the observed process have different
hydrodynamic limits in general, and the macroscopic action functional measures the difference between them. It is minimised in the
case where the observed process and the reference process coincide, in which case the action is zero --- under suitable
assumptions, this means that the hydrodynamic limit of the reference process can be characterised as the minimiser of the
macroscopic action. Moreover, the macroscopic action can be represented as a sum of three terms --- we show that these individual
contributions are asymptotically dominated by corresponding contributions to the microscopic action, see
Theorem~\ref{thm:lower-bound-thm}.  Then, for a specific choice of the observed process (which is related to the hydrodynamic
limit of the reference process), we show that the microscopic action converges to the macroscopic one, see
Theorems~\ref{thm:Q_H_limit} and~\ref{thm:final_thm}.

The inspiration for this study comes from~\cite{Fathi2016a} and~\cite{Erbar2016a}, which derive hydrodynamic (or mean-field)
limits as minimisers of macroscopic action functionals, for the simple exclusion process~\cite{Fathi2016a} and for a McKean-Vlasov
equation on a finite graph~\cite{Erbar2016a}.  In common with these works, our approach is (loosely) based on the Sandier-Serfaty
approach~\cite{Serfaty2011a} to study sequences of gradient flows via $\varGamma$-convergence.  However, our approach is different
from~\cite{Fathi2016a,Erbar2016a} because it starts from the (non-quadratic) canonical structure, instead of the \emph{quadratic}
structure for time-reversal symmetric Markov chains, that was independently derived by Maas~\cite{Maas2011a} and
Mielke~\cite{Mielke2011a}. A similar structure to the canonical one exploited here was recently used in~\cite{Basile2017a} to
derive a diffusive limit for the linear Boltzmann equation. All of these approaches have in common that they consider
time-reversal symmetric systems for which the dynamics can be identified with gradient flows of a free energy functional, so that
the limiting probability measure concentrates on curves of maximal slope, which can be identified as minimisers of the macroscopic
action. Further, our approach is also closely related to EDP-convergence, where EDP stands for Energy-Dissipation-Principle, see
e.g.~\cite{Liero2017a,Bonaschi2016a,Duong2017a,Mielke2016b}.

Compared with previous studies, our work has two novel features. First, we do not restrict to curves of maximal slope (which
follow the gradient of the free energy): instead we consider a class of paths for which the microscopic action functional stays
controlled, in the hydrodynamic limit.  In principle, this means that our methods are not limited to time-reversal symmetric
systems: the corresponding action functional can be defined for a large class of Markov chains in a meaningful way.  However, in
order to reduce the number of technical issues we have to deal with, we limit ourselves to reversible systems in this work.  (More
precisely, we consider Markov chains with (in general) time-dependent rates, where the rates at every time obey detailed balance
with respect to an invariant measure that also (in general) depends on time.  This means that we can exploit readily-available
tools from the theory of hydrodynamic limits for these processes, notably the replacement lemma.)  An extension to systems without
detailed balance is left for future work.

The second novel aspect is that we consider particle systems for which the hydrodynamic limit is a \emph{non-linear} diffusion
equation, in contrast (for example) to the symmetric exclusion process studied in~\cite{Fathi2016a}, whose hydrodynamic limit is
linear diffusion. This is a significant difference for rigorous results: within the canonical structure one sees naturally that
the hydrodynamic limit is a (generalised) gradient flow, as expected on physical grounds.  However, in contrast to (linear)
diffusion with a linear mobility, where the (now-)classic Wasserstein evolution provides the natural geometrical setting for the
gradient flow, the analogous setting for diffusions with non-linear mobility is not so well-developed.  In particular, a key
challenge is to establish the validity of a chain rule for the macroscopic entropy functional, which is known for linear
diffusion~\cite{Ambrosio2008b}, but whose extension to the non-linear setting is not at all straightforward.  We show here that
(with some technical effort) the required results for non-linear diffusion can be obtained by casting the evolution into the
classic Wasserstein setting (Theorem~\ref{sec:chain_rule_for_F_V}): this is not the most natural (physical) setting for the
process of interest, but it is sufficient to establish the required results.

This line of research --- linking Markov chains and partial differential equations via canonical structures --- is quite
recent. Consequently, a number of problems remain open. In particular, our approach is not yet a hydrodynamic limit passage: for
this, the macroscopic concentration of the limiting path measure would have to be proved.  Also, the microscopic action converges
in the hydrodynamic limit to a macroscopic action functional that turns out to coincide with a large deviation rate
functional~\cite{Bertini2015a}.  However, in this work we do not establish any links to large deviation theory; this could be a
natural future line of research (e.g.~one could consider similar calculations to the ones in~\cite{Duong2017b} for independent
particles with Langevin dynamics). Another question is whether (and how) the method presented here can provide guidance for limit
passages for non-reversible systems.

Our study combines techniques from a number of different fields: we have attempted to make it self-contained (and hence accessible
to a general reader), at the expense of including some classical material (which expert readers may prefer to skip). This is
indicated in the beginning of the relevant sections.  In Section~\ref{sec:setting}, we describe the particle systems and their
canonical structure. Section~\ref{sec:main_results} states the main results.  Section~\ref{sec:regularity} is entirely devoted to
technical questions of regularity and a proof of the chain rule, while Section~\ref{sec:proof_bounds} contains the proofs of the
main theorems.

\section{Interacting Particle Systems}
\label{sec:setting}

\subsection{Particle Systems on the Discrete Torus}
\label{sec:ips}

The setting we analyse covers a broad class of particle models, as we now describe. This section also collects some classic facts
on particle models. We consider systems with a fixed number of indistinguishable particles, distributed over the $L^d$ sites of
the flat torus $\mathbb T_L^d := \mathbb Z^d/(L\mathbb Z^d)$. Let $\eta(i)$ be the number of particles on site
$i\in\mathbb T_L^d$, so the configuration space of the system is $\Omega_L\subseteq\mathbb N_0^{\mathbb T_L^d}$. Configurations
are denoted with $\eta=(\eta(i))_{i\in \mathbb T_L^d}$.  Let $\eta^{i,i'}$ be the configuration obtained from $\eta$ by moving a
particle from site $i$ to site $i'$.  The total number of particles on each site may be bounded by $N_{\max}\in\mathbb N_0$, that
is, $\Omega_L=\{0,\dots, N_{\max}\}^{\mathbb T_L^d}$, or unbounded. We fix $T>0$ and consider the time interval $[0,T]$. The
(random) state of the system at time $t\in[0,T]$ is denoted by $\eta_t$.

The particles hop between sites of the lattice with some rate $\hat r_{\eta,\eta^{i,i'}}$, which is assumed to be non-zero only if
$i$ and $i'$ are neighbours, $|i-i'|=1$.  We consider a parabolic scaling, so the hydrodynamic limit is obtained by rescaling time
by a factor $L^2$, such that the transition rates for the Markov chain are $r_{\eta,\eta^{i,i'}}= L^2 \hat r_{\eta,\eta^{i,i'}}$.
Let $\Lambda$ be the flat torus $\mathbb T^d=[0,1)^d$. The jump rates for the particle models considered in this article depend on
  an external potential $V\in C^2(\Lambda;\mathbb R)$, and two functions $g_1,g_2\colon\mathbb N_0\to[0,\infty)$, such that
\begin{equation}
  \label{eqn:rates_V}
  \hat r_{\eta,\eta^{i,i'}}^V=g_1(\eta(i))g_2(\eta(i'))\mathrm e^{-\frac 12(V(i'/L)-V(i/L))}.
\end{equation}
We also consider time-dependent potentials $\tilde V\in C^{1,2}([0,T]\times\Lambda;\mathbb R)$ which lead to a time-heterogeneous
Markov chain with transition rates $r^{\tilde V_t}$ at time $t\in[0,T]$.  We write $\tilde V$ for a time-dependent potential and
$V$ for a time-independent potential. The choice in~\eqref{eqn:rates_V} includes many particle processes, such as the zero-range
process and the simple exclusion process. This specific form was chosen to enable the use of existing results from the theory of
hydrodynamic limits, notably the replacement lemma employed below.

An interacting particle system has \emph{gradient dynamics} (or is of \emph{gradient type}) if there exists a function $\mathrm
d\colon \mathbb N_0\to [0,\infty)$ such that (for $V=0$) $r^0_{\eta,\eta^{i,i'}}-r^0_{\eta,\eta^{i',i}}=\mathrm d(\eta(i))-\mathrm
  d(\eta(i'))$.  In this case we define $\hat \phi_i(\mu):=\sum_{\eta\in\Omega_L} \mu(\eta)\;\!\mathrm d(\eta(i))$.  (Note that
  this is the simplest form of a gradient system, which in more generality can consist of differences of finite cylinder
  functions, cf.~\cite{Kipnis1999a}).

\subsubsection{Invariant Measures, Initial Conditions, and Microscopic Free Energy}
\label{sec:invariant_meas}

The number of particles is conserved by the dynamics, so these systems have many possible invariant measures.  The hydrodynamic
limit relies on a particular structure for these measures, as follows. Let $\nu_*$ be a (not necessarily normalised) reference
measure on $\Omega_L$, with $\nu_*(\eta)>0$ for all $\eta\in\Omega_L$, which is assumed to have a product structure in the sense
that $\nu_*(\eta)=\prod_{i\in\mathbb T_L^d}\nu_{*,1}(\eta(i))$ for some probability measure $\nu_{*,1}$ on $\mathbb N_0$. We
assume that the process with rates $\hat r^0$ satisfies the \emph{detailed balance condition}
\begin{equation}
  \label{eqn:db}
  \nu_*(\eta)\;\!\hat r_{\eta,\eta^{i,i+e_k}}^0 = \nu_*(\eta^{i,i+e_k})\;\!\hat r_{\eta^{i,i+e_k},\eta}^0
\end{equation}
for all $\eta\in\Omega_L$, $i\in\mathbb T_L^d$ and $k=1,\dots,d$.  This implies that $\nu_*$ is invariant for the dynamics $\hat
r^0$ and that these dynamics are time reversal-symmetric with respect to $\nu_*$.  To avoid technical difficulties, we further
assume that the one site \emph{partition function} is finite, i.e.~for all $\theta\in \mathbb R$
\begin{equation}
  \label{eqn:finite_moments}
Z_1(\theta):=\sum_{n\in\mathbb N_0}\mathrm e^{\theta n}\nu_{*,1}(n)<\infty.
\end{equation}

In classical statistical mechanics (see for example~\cite[Section 3]{Bertini2007a} or~\cite{Chandler1987a}), the \emph{local free
  energy density} is given by the Legendre dual of the \emph{cumulant generating function (or pressure)} of $\nu_{*,1}$, i.e.
\begin{equation}
  \label{eqn:free_energy_legendre}
  f(a) = \sup_{\theta\in\mathbb R}\;\!\bigl(a \;\!\theta - \log Z_1(\theta)\bigr)=a\;\!f'(a) - \log Z_1(f'(a)),
\end{equation}
which implies that $f$ is convex. In the following, we will assume that $f\in C^2([0,N_{\max}];\mathbb R)$ and that a.e.~$f''>0$,
see Section~\ref{sec:new-assump}). Now, for $\alpha\in (0, N_{\max})$, we define the probability measures
\begin{equation}
  \label{eqn:stationary_measure}
  \nu_{\alpha,1}(n):=\frac {\mathrm e^{f'(\alpha)n}}{Z_1(f'(\alpha))}\nu_{*,1}(n)
\end{equation}
and $\nu_\alpha := \prod_{i\in\mathbb T_L^d}\nu_{\alpha,1}$. For each $\alpha\in (0,N_{\max})$ this choice implies that
$E_{\nu_\alpha}\bigl[\sum_{i\in\mathbb T_L^d}\eta(i)/L^d\bigr]=\alpha$ (where $E_{\nu_\alpha}$ denotes the expectation with
respect to $\nu_\alpha$) and that $\nu_{\alpha}$ is stationary and satisfies~\eqref{eqn:db} for the process with rates $\hat
r^0$. For an external potential $V\in C^2(\Lambda;\mathbb R)$ the process with rates $\hat r^V$ satisfies detailed balance with
respect to the probability measures $\nu_\alpha^V(\eta)\propto \nu_\alpha(\eta)\mathrm e^{-\sum_{i\in\mathbb
    T_L^d}V(i/L)\eta(i)}$.  For the measure $\nu_\alpha^V$, the expected number of particles at $u\in\Lambda$ is defined as
\begin{equation}
  \label{eqn:local_no_of_particles}
  \bar\rho_{\alpha,V}(u):=\frac{E_{\nu_{\alpha,1}}\bigl[\eta(0)\mathrm e^{-V(u)\eta(0)}\bigr]}{E_{\nu_{\alpha,1}}
    \bigl[\mathrm e^{-V(u)\eta(0)}\bigr]}<\infty.
\end{equation}
Combining~\eqref{eqn:local_no_of_particles} with~\eqref{eqn:stationary_measure} allows to show that
$\bar\rho_{\alpha,V}(u) = (f')^{-1}(-V(u) + f'(\alpha))$, or equivalently
%The latter allows to show that
$f'(\bar\rho_{\alpha,V}(u)) = -V(u) + f'(\alpha)$. Consequently~\eqref{eqn:local_no_of_particles} is
strictly monotonically increasing in $\alpha$. Since the number of particles is conserved, its distribution is fully determined by
the initial condition for the model.  In everything that follows, we restrict to initial distributions $(\mu^L_0)_{L\in\mathbb N}$
for which the total density of particles is bounded uniformly: there exists $C_{\rm tot}\in(0, N_{\max}]$ such that for all
$L\in\mathbb N$
\begin{equation}
  \label{eqn:particle_bound}
  \mu^L_0\biggl(\eta\in\Omega_L\;\!\Bigl|\;\!\frac 1{L^d}\sum_{i\in\mathbb T_L^d}\eta(i)\le C_{\rm tot}\biggr)=1.
\end{equation}
This means that the Markov chain is supported on finitely many configurations, allowing us to treat each particle system as a
finite state Markov chain. Finally, for any $V\in C^2(\Lambda;\mathbb R)$ and any $\alpha$, define the relative entropy (or
\emph{microscopic free energy}) as
\begin{equation}
  \label{eqn:free_energy_alpha}
  \mathcal F_{L,\alpha}^V(\mu) :=\mathcal H\bigl(\mu|\nu_\alpha^V\bigr)=\sum_{\eta\in\Omega_L} \mu(\eta) \log 
  \Bigl(\frac{\mu(\eta)}{\nu_\alpha^V(\eta)}\Bigr),
\end{equation}
where $\mu$ is a probability measure (on $\Omega_L$).  If $\mu$ is the probability measure for our interacting particle system at
some time $t$ then $\mathcal F_{L,\alpha}^V(\mu)<\infty$, by~\eqref{eqn:particle_bound}, since $\nu_*(\eta)>0$ for all
$\eta\in\Omega_L$.

\subsubsection{Canonical Structure for Markov Chains}

We now describe a $\Psi$-$\Psi^\star$ structure for finite state Markov chains which is related to a relative entropy between path
measures~\cite{Kaiser2018a}.  This structure is central to this article (see also~\cite{Maes2008a,Maes2008b}). Let $\mu$ be a
probability measure on $\Omega_L$ supported on finitely many configurations.  We think of this measure as a (generic) distribution
of the particle system.  For $\eta,\eta'\in\Omega_L$ we define the \emph{probability current} from $\eta$ to $\eta'$ as
\begin{equation}
  \label{eqn:probab_current}
  J_{\eta,\eta'}(\mu) := \mu(\eta)r^V_{\eta,\eta'}-\mu(\eta')r^V_{\eta',\eta} .
\end{equation}
The divergence at $\eta$ is $\operatorname{div}J(\mu)(\eta):= \sum_{\eta'\in\Omega_L}
J_{\eta,\eta'}(\mu)$. Following~\cite{Kaiser2018a}, define a \emph{mobility}
\begin{equation}
a_{\eta,\eta'}(\mu):=2\bigl[\mu(\eta)r^V_{\eta,\eta'}\mu(\eta')r^V_{\eta',\eta}\bigr]^{1/2}
\end{equation}
which is independent of $V$ since $\hat r_{\eta,\eta'}^V\hat r_{\eta',\eta}^V=\hat r_{\eta,\eta'}^0\hat r_{\eta',\eta}^0$.  Let
the discrete gradient of a function $h$ on $\Omega_L$ be $\nabla^{\eta,\eta'}h:=h(\eta')-h(\eta)$ and define a \emph{thermodynamic
  force} (cf.~\cite{Maes2008a,Maes2008b,Kaiser2018a}) as
\begin{equation}
  \label{eqn:force}
  F^V_{\eta,\eta'}(\mu):=-\nabla^{\eta,\eta'}\log\Bigl(\frac{\mu}{\nu_\alpha^V}\Bigr),
\end{equation}
which is in fact independent of $\alpha$, as $\nu_\alpha(\eta)/\nu_\alpha(\eta^{i,i'})=\nu_*(\eta)/\nu_*(\eta^{i,i'})$. For a
general interpretation of the mobility and the force and their physical relation to thermodynamic quantities, such as entropy
production and housekeeping heat, we refer the reader to~\cite{Kaiser2018a}.

The canonical structure is based on a dual paring between currents and thermodynamic forces.  We consider generic currents $j$ and
forces $F$, which are arbitrary anti-symmetric functions on $\Omega_L\times\Omega_L$ with $j_{\eta,\eta'}=-j_{\eta',\eta}$ and
$F_{\eta,\eta'}=-F_{\eta',\eta}$. The dual pairing is $\langle j, F\rangle_L := \frac 12 \sum_{\eta,\eta'\in\Omega_L}
j_{\eta,\eta'}F_{\eta,\eta'}\mathbf 1_{\{a_{\eta,\eta'}(\mu) >0\}}$ (which implicitly depends on $\mu$). Here $\mathbf 1_A$ is the
indicator function of the event $A$, which is given by $\mathbf 1_A=1$ if the statement $A$ is satisfied and $\mathbf 1_A=0$
otherwise. Now define
\begin{equation}
  \label{eqn:basic_psi_star}
  \Psi^\star_L(\mu, F):= \sum_{\eta,\eta'\in\Omega_L} a_{\eta,\eta'}(\mu) \Bigl[\cosh\bigl(\tfrac 12 F_{\eta,\eta'}\bigr)-1\Bigr]
\end{equation}
and 
\begin{multline}
  \label{eqn:basic_psi}
  \Psi_L(\mu,j) := \sum_{\eta,\eta'\in\Omega_L} a_{\eta,\eta'}(\mu) 
  \biggl[\frac{j_{\eta,\eta'}}{a_{\eta,\eta'}(\mu)}\operatorname{arcsinh}\Bigl(\frac{j_{\eta,\eta'}}{a_{\eta,\eta'}(\mu)}
  \Bigr)\\
  -\cosh\biggl(\operatorname{arcsinh}\Bigl(\frac{j_{\eta,\eta'}}{a_{\eta,\eta'}(\mu)}\Bigr)\biggr)+1\biggr],
\end{multline}
where the summands in~\eqref{eqn:basic_psi} have to be interpreted as being equal to zero whenever $a_{\eta,\eta'}(\mu)=0$.  The
two functions~\eqref{eqn:basic_psi_star} and~\eqref{eqn:basic_psi} are both symmetric and strictly convex in their second
argument. Moreover, they are Legendre dual with respect to the dual pairing $\langle j,F\rangle_L$ and give rise to the
\emph{Onsager-Machlup functional},
\begin{equation}
  \label{eqn:om_functional}
  \Phi_L(\mu,j,F) := \Psi_L(\mu,j) - \langle j, F\rangle_L + \Psi^\star_L(\mu, F)\ge 0,
\end{equation}
where the inequality follows from the Fenchel-Young inequality (which directly follows from the Legendre duality of $\Psi$ and
$\Psi^\star$).  This functional will be used in the following to characterise the relative entropy between path measures. In
particular, we will study the convergence of the \emph{non-quadratic} functionals $\Psi$ and $\Psi^\star$ to their quadratic
counterparts to a macroscopic quadratic functional, which has the form of the macroscopic Onsager-Machlup functional.

\subsubsection{Projection onto the Physical Domain}

So far we considered currents and densities on the full configuration space $\Omega_L$. To obtain hydrodynamic behaviour, we
`project' the system onto the physical domain $\mathbb T_L^d$ and also embed the sequence of these domains (indexed by $L$) into
the flat torus $\Lambda$. This section introduces the associated notation.

For a (generic) probability measure $\mu$ on $\Omega_L$ (which we again think of as the current distribution of the particle
system), we can define the averaged number of particles $\hat\rho_i(\mu)$ at site $i\in \mathbb T_L^d$ and an averaged particle
current $\hat\jmath^V_{i,i'}(\mu)$, as
\begin{equation}
  \label{eqn:density_current}
  \hat\rho_i(\mu) := \sum_{\eta\in\Omega_L}\mu(\eta)\eta(i)\qquad\textrm{and}\qquad 
  \hat\jmath^V_{i,i'}(\mu):=\sum_{\eta\in\Omega_L}\mu(\eta)\bigl(\hat r^V_{\eta,\eta^{i,i'}}-\hat r^V_{\eta,\eta^{i',i}}\bigr).
\end{equation}
The current $\hat\jmath^V_{i,i'}(\mu)$ describes the expected net flow of particles from site $i$ to site $i'$ if the distribution
of the particle system is given by $\mu$. For gradient dynamics and $V=0$ the current~\eqref{eqn:density_current} is
\begin{equation}
  \label{eqn:gradient_j0}
  \hat\jmath_{i,i'}^{\;\!0}(\mu) = \hat \phi_i(\mu)-\hat \phi_{i'}(\mu) = -\nabla^{i,i'}\hat\phi(\mu),
\end{equation} 
where the discrete gradient on $\mathbb T_L^d$ is (for $h\colon\mathbb T_L^d\to\mathbb R$) defined as $\nabla^{i,i'}h =
h(i')-h(i)$. Similar to~\eqref{eqn:density_current}, define also two (averaged) mobilities for the edge connecting $i$ and $i'$ as
\begin{equation}
  \label{eqn:hat_chi}
  \hat a_{i,i'}(\mu):%=\sum_{\eta\in\Omega_L} a_{\eta,\eta^{i,i'}}(\mu)
  =\sum_{\eta\in\Omega_L}2\bigl[\mu(\eta)\hat r^V_{\eta,\eta^{i,i'}}\mu(\eta^{i,i'})\hat r^V_{\eta^{i,i'},\eta}\bigr]^{1/2}, \quad
  \hat\chi^V_{i,i'}(\mu):=\frac 12\sum_{\eta\in\Omega_L} \mu(\eta)\bigl(\hat r^V_{\eta,\eta^{i,i'}}+\hat r^V_{\eta,\eta^{i',i}}\bigr),
\end{equation}
which are related by $\hat a_{i,i'}(\mu)\le 2\hat\chi^V_{i,i'}(\mu)$ (with equality for $\mu=\nu_\alpha^V$). Note that the two
mobilities characterise the average particle jumps between $i$ and $i'$ and are therefore symmetric in $i$ and $i'$.

For the embedding on the flat torus, let $\mathcal M_+(\Lambda)$ be the set of finite and non-negative Radon measures on
$\Lambda$, endowed with the weak topology. Define the empirical measure $\Theta_L\colon \Omega_L\to \mathcal M_+(\Lambda)$ as
\begin{equation}
  \label{eqn:empirical_measure}
  \Theta_L(\eta) := \frac 1{L^d}\sum_{i\in\mathbb T_L^d} \eta(i) \;\! \delta_{i/L}.
\end{equation}
Thus, each configuration $\eta$ of an interacting particle system of size $L$ corresponds to a measure $\Theta_L(\eta) \in
\mathcal M_+(\Lambda)$.

\subsubsection{Reference Process and Observed Process}
\label{sec:Refer-proc-observ}

We analyse hydrodynamic limits by comparing different (microscopic) processes.  For any given $L$, the \emph{reference process} is
an interacting particle system on the discrete torus, as defined in Section~\ref{sec:ips}.  The \emph{observed process} is another
interacting particle system on the same space, whose path measure (see below) is absolutely continuous with respect to the
reference process.  Hydrodynamic limits are analysed by considering sequences of observed and reference processes, indexed by $L$.
With slight abuse of terminology, we sometimes refer to the sequence of observed processes as simply ``the observed process'', and
similarly for the reference process.

We consider observed processes with unique hydrodynamic limits.  This leads to a variational characterisation of the hydrodynamic
limit of the reference process, by minimising the relative entropy between the reference process and the observed process.  This
follows the usual approach in the calculus of variations: one considers observed processes with (known) hydrodynamic limits, which
are candidates for the hydrodynamic limit of the reference process.  The optimal candidate is the one that minimises the relative
entropy, and the hydrodynamic limit of this optimal candidate matches the hydrodynamic limit of the reference process (assuming
that it exists).

\subsection{Path Measures on the Microscopic Scale}
\label{sec:path_meas_super}

\subsubsection{Path Measures for the Reference and Observed Processes}
\label{sec:path_meas}

Our analysis of the hydrodynamic limit is based on the convergence of path measures. In this section, we introduce the notation
that allows us to define the path measures $Q_L$ and limit measures $Q^*$ studied in the remainder of the article.

For any topological space $\mathcal S$ we denote with $\mathcal D([0,T];\mathcal S)$ the set of $\mathcal S$ valued c\`adl\`ag
paths (right-continuous paths with left limits) on $[0,T]$.  For details, see~\cite[Chapter 3]{Billingsley1999a}, as well
as~\cite[Chapter 4.1]{Kipnis1999a} and~\cite{Bertini2009c}.  For $t\in[0,T]$ let
$X_t\colon\mathcal D([0,T];\mathcal S)\to \mathcal S$ be the marginal at time $t$, which evaluates a path
$\gamma = (\gamma_t)_{t\in[0,T]}\in\mathcal D([0,T];\mathcal S)$ at time $t$: $X_t(\gamma) = \gamma_t$. We recall that whilst
$X_t$ is measurable for all $t\in[0,T]$, it is continuous only for almost all $t\in(0,T)$, as well as $t=0$ and $t=T$.
In the following, the expression \emph{path measure} will refer to a probability distribution on $\mathcal D([0,T];\mathcal S)$
for some $\mathcal S$.

Given some $L$, the \emph{reference process} is a particle system with a time-dependent potential $\tilde V\in
C^{1,2}([0,T]\times\Lambda;\mathbb R)$, whose path measure [on $\mathcal D([0,T];\Omega_L)$] is denoted by $P_L^{\tilde V}$.  We
can recover the distribution of this Markov chain at time $t$ from $P_L^{\tilde V}$ via the push-forward measure $(X_t)_\#
P_L^{\tilde V}$.

The \emph{observed process} can be any (possibly time-heterogeneous) Markov chain on $\Omega_L$, whose path measure [on $\mathcal
  D([0,T];\Omega_L)$] is denoted by $P_L$.  This process is assumed to have the following properties: the path measure $P_L$ is
absolutely continuous with respect to $P_L^{\tilde V}$, the initial condition of $P_L$ coincides with the one of $P_L^{\tilde V}$,
that is, $(X_0)_\#P_L=(X_0)_\# P_L^{\tilde V}=\mu_0^L$, and the transition rates $r^L_t$ are bounded in time, i.e.~for each
$L\in\mathbb N$, we assume that $\sup_{t\in[0,T]}(r^L_t)_{\eta,\eta'}<\infty$ for all $\eta,\eta'\in\Omega_L$.

We can assign to $P_L$ a unique path $(\mu^L_t,\jmath^L_t)_{t\in[0,T]}$ consisting of the density $\mu^L_t:=(X_t)_\#P_L$ and the
current $(\jmath^L_t)_{\eta,\eta'}:=\mu^L_t(\eta)(r^L_t)_{\eta,\eta'}-\mu^L_t(\eta')(r^L_t)_{\eta',\eta}$, which are again linked
by a continuity equation $\partial_t\mu^L_t = -\operatorname{div}\jmath^L_t$. 

We remark that for the choice $P_L=P_L^{\tilde V}$ the current $\jmath^L_t$ simply coincides with the probability
current~\eqref{eqn:probab_current} for the time-dependent rate $ r^{\tilde V_t}$.  In this case, one can further show that the
associated density and current~\eqref{eqn:density_current} satisfy the continuity equation
$\partial_t \hat\rho_i(\mu_t^L) = - \operatorname{div} \hat\jmath^V(\mu_t^L)(i)$, where the divergence on the physical domain
$\mathbb T_L^d$ is defined as $\operatorname{div}\hat\jmath^V(\mu)(i) := \sum_{i'\in\mathbb T_L^d}\hat\jmath^V_{i,i'}(\mu)$.

Since every $\Omega_L$ can be embedded into the flat torus $\Lambda$ (as a map from $\Omega_L$ to $\mathcal M_L(\Lambda)$), there
is a corresponding embedding of the path space $\mathcal D([0,T];\Omega_L)$ into $\mathcal D([0,T];\mathcal M_L(\Lambda))$. In
particular, each path measure $Q_L$ on $\mathcal D([0,T];\mathcal M_+(\Lambda))$ that is supported on
$\mathcal M_L(\Lambda) :=\{ L^{-d}\sum_{i\in\mathbb T_L^d} k_i \delta_{i/L}\;\! |\;\! k_i \in \mathbb N_0, k_i\le N_{\max}\}$ can
be identified with a unique measure $P_L$ on $\mathcal D([0,T];\Omega_L)$. The measure on
$\mathcal D([0,T];\mathcal M_+(\Lambda))$ that corresponds to the reference process $P_L^{\tilde V}$ is denoted with
$Q_L^{\tilde V}$. Similarly, for the observed process, there is a $Q_L$ corresponding to $P_L$.  No information is lost on
embedding the processes into $\Lambda$, so $\mathcal H(Q_L|Q_L^{\tilde V})=\mathcal H(P_L|P_L^{\tilde V})$, which can be proved by
two applications of Lemma 9.4.5 in~\cite{Ambrosio2008b} with the bijection from $\mathcal M_L(\Lambda)$ to $\Omega_L$.

We summarise this notation, which will be used extensively below: the reference process and the observed processes can be fully
characterised by their path measures [both on $\mathcal D([0,T];\mathcal M_L(\Lambda))$], which are denoted by $Q_L^{\tilde V}$
and $Q_L$ respectively.  There are corresponding path measures on $\mathcal D([0,T];\Omega_L)$ which are denoted by $P_L^{\tilde
  V}$ and $P_L$.

\subsubsection{Microscopic Action Functional}
\label{sec:relent_force_xx}

To compare the reference and the observed process, consider the thermodynamic force for the reference process at time $t$, which
is $F^{\tilde V_t}(\mu^L_t)$, evaluated from~\eqref{eqn:force} with $\mu^L_t=(X_t)_\#P_L$. Since $P_L$ is absolutely continuous
with respect to $P_L^{\tilde V}$, the relative entropy $\mathcal H(P_L|P_L^{\tilde V})$ is under the assumptions in
Section~\ref{sec:path_meas} finite and (cf.~\cite[Appendix]{Kaiser2018a}) coincides with
\begin{equation}
  \label{eqn:rel_ent}
  \mathcal H\bigl(P_L|P_L^{\tilde V}\bigr) = \mathcal H\bigl(\mu^L_0|(X_0)_\# P_L^{\tilde V}\bigr) + \frac 12 \int_0^T 
  \Phi_L\bigl(\mu^L_t,\jmath^L_t,F^{\tilde V_t}(\mu^L_t)\bigr)\;\!\mathrm dt .
\end{equation}
Moreover, $\mathcal H(\mu^L_0|(X_0)_\# P_L^{\tilde V})=0$, since $P_L$ and $P_L^{\tilde V}$ share the same initial condition. We
interpret $\frac 12 \Phi_L(\mu^L_t,\jmath^L_t,F_{\alpha}^{\tilde V_t}(\mu^L_t))$ as an extended Lagrangian~\cite{Kaiser2018a} and
define the \emph{microscopic action} of the path measure $Q_L$ as the relative entropy
\begin{equation}
  \label{eqn:rel_ent_2a}
  \mathbb A_L^{\tilde V}\bigl(Q_L\bigr):= \mathcal H\bigl(Q_L|Q_L^{\tilde V}\bigr)=\mathcal H\bigl(P_L|P_L^{\tilde V}\bigr)
   = \frac 12 \int_0^T 
  \Phi_L\bigl(\mu^L_t,\jmath^L_t,F^{\tilde V_t}(\mu^L_t)\bigr)\;\!\mathrm dt.
\end{equation}
This is the central functional defined on the discrete (lattice) level studied in this article.

\subsection{Macroscopic Quantities}
\label{sec:macro_quant}

In the hydrodynamic scaling limit, the microscopic action~\eqref{eqn:rel_ent_2a} will converge to a macroscopic action, which
is~\eqref{eqn:action_functional}.  (For the macroscopic setting, we restrict our considerations to potentials $V$ that are
constant in time.)  We now show how the macroscopic action functional is constructed.

\subsubsection{The Macroscopic Free Energy}
\label{sec:free_energy_macro}

For $\alpha\in(0, N_{\max}]$ and $V\in C^2(\Lambda;\mathbb R)$, we define the \emph{macroscopic free energy} $\mathcal
  F_\alpha^V\colon\mathcal M_+(\Lambda)\to[0,\infty]$ as
\begin{equation}
  \label{eqn:2nd_free_energy}
  \mathcal F_\alpha^V(\pi):=
  \sup_{h\in C(\Lambda;\mathbb R)}\biggl[ \langle \pi, h\rangle - \int_\Lambda 
  \log\biggl(\frac{Z_1( f'(a) + h(u) - V(u) )}{Z_1( f'(a) - V(u))}\biggr)
  \;\!\mathrm du\biggr].
\end{equation}
This free energy coincides with a rate function: there is a large-deviation principle for the particle configuration $\Theta_L$
sampled from the the steady state $\nu_\alpha^V$; the speed of this LDP is $L^d$ and its rate function is $\mathcal
F_\alpha^V(\pi)$, (see e.g.~Section 5.1, page~75 in~\cite{Kipnis1999a} for the special case of a zero-range
process). From~\eqref{eqn:finite_moments}, $\mathcal F_\alpha^V(\pi)$ is finite only if $\pi(\mathrm du)=\rho(u)\mathrm du$ for
some density $\rho\in\mathcal L^1(\Lambda;[0,\infty))$. In the following we thus write $\mathcal F_\alpha^V(\rho)$ for $\mathcal
  F_\alpha^V(\pi)$. As in Macroscopic Fluctuation Theory~\cite[Section 5.A]{Bertini2015a}, we can represent $\mathcal F_\alpha^V$
  for reversible systems as
\begin{equation}
  \label{eqn:1st_free_energy}
  \mathcal F_\alpha^V(\rho)= \int_\Lambda 
  \Bigl[f(\rho(u))-f(\bar\rho_{\alpha,V}(u))-f'(\bar\rho_{\alpha,V}(u))\bigl(\rho(u)-\bar\rho_{\alpha,V}(u)\bigr)\Bigr]\;\!\mathrm du,
\end{equation}
where $\bar\rho_{\alpha,V}\in \mathcal L^1(\Lambda;[0,\infty))$, introduced in~\eqref{eqn:local_no_of_particles}, is the steady
  state density for the dynamics of the macroscopic system. Note that~\eqref{eqn:1st_free_energy} inherits the convexity of $f$.

\subsubsection{The Hydrodynamic Current and the Hydrodynamic Equation}
\label{sec:hyrdo_eqn}

In the hydrodynamic limit, the particle density at time $t$ is given by some $\rho_t\in\mathcal L^1(\Lambda;[0,\infty))$. The
  hydrodynamic current describes the resulting particle flow:
\begin{equation}
  \label{eqn:hydro_current}
  J(\rho):=-\nabla\phi(\rho)-\chi(\rho)\nabla V,
\end{equation}
where $\phi$ and $\chi$ are functions that depend on the system of interest and are discussed later in this section.  The
hydrodynamic equation is then
\begin{equation}
  \label{eqn:hydro_eqn_xxx}
  \dot\rho_t = - \nabla\cdot J(\rho_t) = \Delta \phi(\rho_t)+\nabla\cdot(\chi(\rho_t)\nabla V).
\end{equation}
In this article, we consider weak solutions to~\eqref{eqn:hydro_eqn_xxx}, in the sense that for all $G\in
C^{1,2}([0,T]\times\Lambda;\mathbb R)$
\begin{multline}
  \label{eqn:weak_sol_xxx}
  \qquad\qquad\int_\Lambda \rho_T G_T\;\!\mathrm du -\int_\Lambda \rho_0G_0\;\!\mathrm du
  - \int_0^T\int_\Lambda \rho_t\;\! \partial_t G_t\;\!\mathrm du \;\!\mathrm dt \\
  =\int_0^T\int_\Lambda \phi(\rho_t)\Delta G_t\;\!\mathrm du\;\!\mathrm dt
  - \int_0^T\int_\Lambda \chi(\rho_t)\nabla V\cdot \nabla G_t\;\!\mathrm du\;\!\mathrm dt. \qquad\qquad
\end{multline}

The dynamics on the macroscopic scale are characterised by the functions $\phi,\chi$ in~\eqref{eqn:hydro_eqn_xxx}. To relate these
quantities to the microscopic dynamics, we consider the case $V=0$, so that $E_{\nu_{\alpha,1}}[\eta(0)]=\alpha$. Define the
macroscopic mobility $\chi\colon [0, N_{\max}]\to[0,\infty)$ as
\begin{equation}
  \label{eqn:chi_indep}
  \chi(\alpha):= \hat\chi_{i,i+e_k}^0(\nu_\alpha)= \frac 12\hat a_{i,i+e_k}(\nu_\alpha),
\end{equation}
which is independent of $i$ and $e_k$ (and thus well-defined). To see this, note from~\eqref{eqn:db} and~\eqref{eqn:hat_chi} that
$\hat\chi_{i,i+e_k}^0(\nu_\alpha) = \sum_{\eta\in\Omega_L}\nu_\alpha(\eta)\hat r_{\eta,\eta^{i,i+e_k}}^0
=E_{\nu_{\alpha,1}}[g_1(\eta(0))]E_{\nu_{\alpha,1}}[g_2(\eta(0))]$,
where we used~\eqref{eqn:rates_V} and the product structure of $\nu_\alpha$.  Similarly, define
$\phi\colon[0, N_{\max}]\to[0,\infty)$ by $\phi(\alpha):=\hat \phi_i(\nu_\alpha)=E_{\nu_{\alpha,1}}[\mathrm d(\eta(0))]$, which is
by construction independent of $i$. One then can prove the \emph{local Einstein relation}
\begin{equation}
  \label{eqn:local_einstein_rel}
  \phi'(\alpha) = f''(\alpha)\chi(\alpha),
\end{equation}
which relates $\phi$ and $\chi$ to the free energy $f$ from
Section~\ref{sec:free_energy_macro}. Equation~\eqref{eqn:local_einstein_rel} can be obtained by differentiating $\phi(\alpha)
=E_{\nu_{*,1}}[\mathrm d(\eta(0))\mathrm e^{f'(\alpha)\eta(0)}]/E_{\nu_{*,1}}[\mathrm e^{f'(\alpha)\eta(0)}]$. Note that $
\phi'(\alpha) = \frac 12 f''(\alpha)\sum_\eta \nu_{\alpha}(\eta)\bigl[\mathrm d(\eta(i))-\mathrm
  d(\eta(i'))\bigr](\eta(i)-\eta(i')) $ (for $i,i'\in\mathbb T_L^d$ arbitrary with $i\not=i'$). Further, the gradient structure
and detailed balance yield $\frac 12 \sum_\eta \nu_{\alpha}(\eta)\bigl[\hat r^0_{\eta,\eta^{i,i'}}-\hat
  r^0_{\eta,\eta^{i',i}}\bigr](\eta(i)-\eta(i')) =\frac 12\sum_\eta \nu_{\alpha}(\eta)\bigl[\hat r^0_{\eta,\eta^{i,i'}}+\hat
  r^0_{\eta,\eta^{i',i}}\bigr]=\chi(\alpha)$.

\subsubsection{The Macroscopic Action Functional and the Chain Rule}
\label{sec:macro_action_functional}

For $\rho\in \mathcal L^1(\Lambda;[0,\infty))$ and $h\colon\Lambda\to\mathbb R^d$, we introduce the norm $\|h\|_{\chi(\rho)}^2:=
  \int_\Lambda \chi(\rho(u))|h(u)|^2\;\!\mathrm du$ (for full details and associated spaces, see Section~\ref{sec:regularity}
  below). The macroscopic analogues of the (time integrals of the) microscopic functions $\Psi_L$ and $\Psi^\star_L$
  from~\eqref{eqn:basic_psi_star},~\eqref{eqn:basic_psi} are
\begin{multline}
  \label{eqn:e_op_norm}
  \mathcal E\bigl((\rho_t)_{t\in[0,T]}\bigr):=
  \sup_G \biggl[\biggl(\int_\Lambda \rho_TG_T \;\!\mathrm du -\int_\Lambda \rho_0G_0\;\!\mathrm du \\
  - \int_0^T\int_\Lambda \rho_t\;\! \partial_t G_t \;\!\mathrm du \;\!\mathrm dt
  \biggr)-\frac 12\int_0^T \|\nabla G_t\|_{\chi(\rho_t)}^2\;\!\mathrm dt \biggr]
\end{multline}
and 
\begin{multline}
  \label{eqn:e_star_op_norm}
  \mathcal E^\star\bigl((\rho_t)_{t\in[0,T]}\bigr):=\sup_G \biggl[ 
  \biggl(\int_0^T\int_\Lambda \phi(\rho_t)\Delta G_t\;\!\mathrm du\;\!\mathrm dt
  - \int_0^T\int_\Lambda \chi(\rho_t)\nabla V\cdot \nabla G_t\;\!\mathrm du\;\!\mathrm dt\biggr)\\
  -\frac 12\int_0^T\|\nabla G_t\|_{\chi(\rho_t)}^2\;\!\mathrm dt\biggr],
\end{multline}
where the supremum is in both cases over $C^{1,2}([0,T]\!\times\! \Lambda;\mathbb R)$. We will show in
Propositions~\ref{prop:e_rep} and~\ref{prop:e_star_rep} that, under certain assumptions, these functionals can be expressed as
time integrals of suitably defined norms
\begin{equation*}
  \label{eqn:e_minus_one-sec-2}
  \mathcal E\bigl((\rho_t)_{t\in[0,T]}\bigr) = \frac 12\int_0^T \|\dot\rho_t\|_{-1,\chi(\rho_t)}^2 \;\!\mathrm dt
\end{equation*}
and
\begin{multline*}
  \label{eqn:e_star_minus_one-sec-2}
  \mathcal E^\star\bigl((\rho_t)_{t\in[0,T]}\bigr)
  = \frac12 \int_0^T \|\Delta\phi(\rho_t)+\nabla\cdot(\chi(\rho_t)\nabla V)\|_{-1,\chi(\rho_t)}^2 \;\!\mathrm dt\\
  = \frac12 \int_0^T \|f''(\rho_t)\nabla \rho_t+\nabla V\|_{\chi(\rho_t)}^2 \;\!\mathrm dt.
\end{multline*}
In particular, we will show that non-quadratic $\Psi$ and $\Psi^\star$ of~\eqref{eqn:basic_psi} and~\eqref{eqn:basic_psi_star} can
be bounded by the quadratic expressions $\mathcal E$ and $\mathcal E^\star$, respectively.

Finally, for $(\pi_t)_{t\in[0,T]}$ absolutely continuous with respect to the Lebesgue measure, we define the \emph{macroscopic
  action} as
\begin{equation}
  \label{eqn:action_functional}
  \mathbb A\bigl((\pi_t)_{t\in[0,T]}\bigr)
  :=\frac 12\bigl[\mathcal F_\alpha^V(\rho_T) - \mathcal F_\alpha^V(\rho_0)
  +\mathcal E\bigl((\rho_t)_{t\in[0,T]}\bigr)
  +\mathcal E^\star\bigl((\rho_t)_{t\in[0,T]}\bigr)\bigr].
\end{equation}
If $(\pi_t)_{t\in[0,T]}$ is not absolutely continuous with respect to the Lebesgue measure, we set $\mathbb
A\bigl((\pi_t)_{t\in[0,T]}\bigr)=+\infty$.

In a nutshell, the main results of this article are twofold: Firstly, we establish relations between suitably scaled $\mathbb
A_L^{\tilde V}$ of~\eqref{eqn:rel_ent_2a} and the continuum limit~\eqref{eqn:action_functional}: see
Theorems~\ref{thm:lower-bound-thm} to~\ref{thm:final_thm}.  Secondly, we show that under suitable regularity assumptions, in
particular if the free energy $\mathcal F_\alpha^V$ satisfies a chain rule (see Equation~\eqref{eqn:chain_rule_formal}), the
macroscopic action can be re-written in a way which reveals the hydrodynamic limit as minimiser of this functional,
see~\eqref{eqn:A_alternative} below.

\subsection{Assumptions on the Particle Systems Studied}
\label{sec:Assumpt-Syst-Stud}

\subsubsection{Local Equilibrium Assumption and the Replacement Lemma} 
\label{sec:local_eqm}

When taking the hydrodynamic limit, one must prove a local equilibration condition, which means that the system resembles --- in a
small neighbourhood around any point --- an equilibrium system. To make this precise, take $\ell\in\mathbb N$ and define the
average number of particles in a box with diameter $2\ell +1$ as
\begin{equation*}
  \eta^\ell(i) := \frac 1{(2\ell+1)^d}\sum_{|m|\le \ell}\eta(i\! +\! m).
\end{equation*}
Similarly, we also define the averages
$\hat\chi_{i,i+e_k}^\ell(\mu):=(2\ell +1)^{-d}\sum_{|m|\le \ell}\hat\chi_{i+m,i+m+e_k}(\mu)$ and
$\hat \phi_i^\ell(\mu):=(2\ell +1)^{-d}\sum_{|m|\le \ell}\hat \phi_{i+m}(\mu)$.

Now assume that $L\gg 1$ and $\epsilon\ll 1$ and that the state of the system is given by $\eta\in \Omega_L$. Define $\ell=\lfloor
\epsilon L\rfloor$, which is the size of a macroscopic box with diameter $\approx 2\epsilon$ (measured on the macroscopic scale).
Hence $\hat\chi_{i,i+e_k}^{\lfloor \epsilon L\rfloor}(\delta_\eta)$ is a locally averaged mobility.  \emph{Local equilibration}
means that $\hat\chi_{i,i+e_k}(\nu_{\eta^{\lfloor \epsilon L\rfloor}(i)})$ is close to the expected mobility for an equilibrium
distribution $\nu_\alpha$ with the same (locally-averaged) particle density.  That is, the time averaged distributions
$\mu^L_{[0,T]}:=\frac 1T\int_0^T \mu^L_t\;\!\mathrm dt$ satisfy in local equilibrium
\begin{equation}
  \label{eqn:local_equilibrium}
  \limsup_{\epsilon\to 0}\limsup_{L\to\infty}\frac 1{L^d}\sum_{i\in\mathbb T_L^d}\sum_{k=1}^d 
  \sum_{\eta\in\Omega_L}\mu^L_{[0,T]}(\eta)\;\!\Bigl|\hat\chi_{i,i+e_k}^{\lfloor \epsilon L\rfloor}(\delta_\eta) 
  - \hat\chi_{i,i+e_k}(\nu_{\eta^{\lfloor \epsilon L\rfloor}(i)})\Bigr|= 0,
\end{equation}
as well as
\begin{equation}
  \label{eqn:local_equilibrium_2}
  \limsup_{\epsilon\to 0}\limsup_{L\to\infty}\frac 1{L^d}\sum_{i\in\mathbb T_L^d}
  \sum_{\eta\in\Omega_L}\mu^L_{[0,T]}(\eta)\;\!\Bigl|\hat \phi_i^{\lfloor \epsilon L\rfloor}(\delta_\eta) 
  - \hat \phi_i(\nu_{\eta^{\lfloor \epsilon L\rfloor}(i)})\Bigr|= 0.
\end{equation}

\paragraph{Remark (Replacement Lemma).}
Note that results like~\eqref{eqn:local_equilibrium} and~\eqref{eqn:local_equilibrium_2} are classically obtained by proving the
stronger replacement lemma, which in our notation amounts to proving for $\hat\chi$ (and analogously for $\hat\phi$)
\begin{equation}
  \label{eqn:replacement_lemma}
  \limsup_{\epsilon\to 0}\limsup_{L\to\infty}\sup_{\mu}\frac 1{L^d}\sum_{i\in\mathbb T_L^d}\sum_{k=1}^d 
  \sum_{\eta\in\Omega_L}\mu(\eta)\;\!\Bigl|\hat\chi_{i,i+e_k}^{\lfloor \epsilon L\rfloor}(\delta_\eta) 
  - \hat\chi_{i,i+e_k}(\nu_{\eta^{\lfloor \epsilon L\rfloor}(i)})\Bigr|= 0,
\end{equation}
where the supremum is taken over a class of measures $\mu$ satisfying certain bounds on the relative entropy (i.e.~the free
energy) and the Dirichlet form, which can be identified with $\frac 12\Psi^\star(\mu,F^V(\mu))$ (see e.g.~the remark in the proof
of Proposition~\ref{prop:lower_bound_finite} below). In the following, we will follow the classical approach and work
with~\eqref{eqn:replacement_lemma}. We state sufficient conditions for the replacement lemma in Section~\ref{sec:suff_local_eqm}
below and establish in this way the validity of~\eqref{eqn:local_equilibrium} and~\eqref{eqn:local_equilibrium_2}.

\subsubsection{Assumptions on the Path Measures $P_L^{\tilde V}$}
\label{sec:new-assump}

We have presented a general framework for interacting particles on lattices and their hydrodynamic scaling limits.  The results of
the next section are similarly general and can be applied to a range of systems, including the symmetric simple exclusion process
and certain zero-range processes, as discussed in Section~\ref{sec:examples} below. However, our results for hydrodynamic limits
clearly do not apply to all interacting-particle systems. We summarise here the main assumptions on the reference process
$P_L^{\tilde V}$ required in the following analysis: these need to be verified in order to apply our results to a particular
system.

On the microscopic scale, we assume that the transition rates are given by~\eqref{eqn:rates_V} and are of gradient type.  The
initial conditions and invariant measures are as described in Section~\ref{sec:invariant_meas}.  We note that many of the proofs
given below make use of assumption~\eqref{eqn:particle_bound}. Despite the fact that it is a non-standard assumption for
hydrodynamic limits (unless $N_{\max}<\infty$, in which case~\eqref{eqn:particle_bound} holds trivially), it is not too
restrictive, in the sense that the typical initial conditions $(\mu^L_0)_{L\in\mathbb N}$ can be shown to satisfy (cf.~equation
(1.4) in Section~5.1 on page~71 in~\cite{Kipnis1999a}) $\lim_{A\to \infty} \limsup_{L\to\infty}\mu^L_0(\eta\in\Omega_L\;\!|\;\!
L^{-d}\sum_{i\in\mathbb T_L^d}\eta(i)\ge A)=0$.

When taking the hydrodynamic limit, we assume that for any sequence of measures $(\mu^L)_{L\in\mathbb N}$
satisfying~\eqref{eqn:particle_bound}, it holds that
\begin{equation}
  \label{eqn:C_hat} 
  C_{\hat\chi}:=\limsup_{L\to\infty}\frac 1{L^d}\sum_{i\in\mathbb T_L^d}\sum_{k=1}^d\hat\chi_{i,i+e_k}(\mu^L)<\infty,
\end{equation}
which ensures that the total rate of particle jumps for the reference process stays controlled as $L\to\infty$.  Similarly we
suppose that any sequence of measures $(\mu^L)_{L\in\mathbb N}$ obeying~\eqref{eqn:particle_bound} also satisfies
\begin{equation}
  \label{eqn:phi_hat} 
  C_{\hat\phi}:=\limsup_{L\to\infty}\frac 1{L^d}\sum_{i\in\mathbb T_L^d}\hat\phi_i(\mu^L)<\infty.
\end{equation}

In addition, our proofs require the following technical assumptions on the functions $f$, $\phi$ and $\chi$ that characterise the
hydrodynamic limit itself: We assume that $f\in C^2([0,N_{\max}];\mathbb R)$ with $f(0)=0$, $f''>0$ a.e.~and that
$\lim_{r\to 0}f'(r)=-\infty$ and $\lim_{r\to N_{\max}}f'(r)=\infty$. Note that this implies by \eqref{eqn:stationary_measure}
that $\phi(0)=0=\chi(0)$.  Further, we assume that $\phi,\chi>0$ on $(0, N_{\max})$ and that both $\phi$ and $\chi$ are Lipschitz
continuous on $[0, N_{\max}]$, without loss of generality with common Lipschitz constant $C_{\rm Lip}>0$. Since
$\phi(0)=\chi(0)=0$, we have in particular $0<\phi(a),\chi(a)\le C_{\rm Lip} a$ for $a\in (0, N_{\max}]$.  We further assume that
$\phi$ is continuously differentiable on $(0, N_{\max})$ (by the above Lipschitz condition with bounded derivative) and also
strictly monotonically increasing. This implies the existence of a continuous inverse
$\phi^{-1}\colon \phi([0, N_{\max}])\to [0, N_{\max}]$, where $\phi([0, N_{\max}])=\{\phi(a): a\in [0, N_{\max}]\}$. We also
suppose that $\phi^{-1}$ has a bounded derivative (which is by the inverse function theorem equivalent to saying that there exists
$C_*>0$ such that $\phi'(a)\ge C_*$ for all $a\in(0, N_{\max}]$).

\section{Statement of the Results} 
\label{sec:main_results}

In this section, we discuss the behaviour of the microscopic action in the limit $L\to\infty$, and the implications of this
behaviour for hydrodynamic limits.  Sections~\ref{sec:act-fcts} and~\ref{sec:suff_local_eqm} derive preliminary results, which
establish properties of the action functionals and sufficient conditions for local equilibration.  Section~\ref{sec:main_section}
states the main results, consisting of three theorems (Theorems~\ref{thm:lower-bound-thm}--\ref{thm:final_thm}). Finally
Section~\ref{sec:examples} discusses the applications of these theorems in two specific particle systems, and their implications
for hydrodynamic limits.

\subsection{Properties of the Microscopic and Macroscopic Action Functions}
\label{sec:act-fcts}

\subsubsection{Chain rule on Microscopic Scale}

Consider $(\mu^L_t,\jmath^L_t)_{t\in[0,T]}$ as in Section~\ref{sec:path_meas}. The force $F^V(\mu^L_t)$ can be linked to the free
energy~\eqref{eqn:free_energy_alpha} via the classical chain rule formula (cf.~Theorem 9.2 of Appendix 1 in~\cite{Kipnis1999a},
Proposition~2.2 in~\cite{Fathi2016a} and also~\cite{Kaiser2018a}) $\mathcal F_{L,\alpha}^V(\mu^L_{t_2}) - \mathcal
F_{L,\alpha}^V(\mu^L_{t_1}) = - \int_{t_1}^{t_2}\langle \jmath^L_t,F^V(\mu^L_t)\rangle_L\;\!\mathrm dt$, which is a special case
of the following result (proved in Section~\ref{sec:proof_xx_new} below).

\begin{proposition}[Chain rule for the microscopic free energy]
  \label{prop:micro_chain_rule}
  Let $\tilde V\in C^{1,2}([0,T]\times\Lambda;\mathbb R)$ and consider a path measure $P_L$ on $\Omega_L$, as described in
  Section~\ref{sec:path_meas}, with associated density and current $(\mu^L_t,\jmath^L_t)_{t\in[0,T]}$. Then the map
  $t\mapsto \mathcal F_{L,\alpha}^{\tilde V_t}(\mu^L_t)$ is absolutely continuous for $t\in[0,T]$ and satisfies the following
  chain rule. For all $0\le t_1<t_2\le T$
  \begin{equation}
    \label{eqn:free_energy_diff}
    \mathcal F_{L,\alpha}^{\tilde V_{t_2}}(\mu^L_{t_2}) - \mathcal F_{L,\alpha}^{\tilde V_{t_1}}(\mu^L_{t_1})
    = - \int_{t_1}^{t_2}\langle \jmath^L_t,F^{\tilde V_t}(\mu^L_t)\rangle_L\;\!\mathrm dt 
    + \int_{t_1}^{t_2}  \sum_{i\in\mathbb T_L^d} \bigl(\hat\rho_i(\mu^L_t)-\bar\rho_{\alpha,\tilde V_t}(i)\bigr)\;\!
    \partial_t\tilde V_t(\tfrac iL)\;\!\mathrm dt.
  \end{equation}
\end{proposition}

Now fix some $\alpha\in(0,N_{\max})$ and combine Proposition~\ref{prop:micro_chain_rule} with~\eqref{eqn:om_functional}
and~\eqref{eqn:rel_ent}, which yields
\begin{multline}
  \label{eqn:rel_ent_2b}
  \mathbb A_L^{\tilde V}\bigl(Q_L\bigr)
  =\frac 12 \bigl[\mathcal F_{L,\alpha}^{\tilde V_T}(\mu^L_T) - \mathcal F_{L,\alpha}^{\tilde V_0}(\mu^L_0)\bigr] 
  + \frac 12\int_0^T \Psi_L(\mu^L_t,\jmath^L_t)\;\!\mathrm dt
  + \frac 12\int_0^T \Psi^\star_L\bigl(\mu^L_t,F^{\tilde V_t}(\mu^L_t)\bigr)\;\!\mathrm dt
  \\-\frac 12\int_0^T \sum_{i\in\mathbb T_L^d} \bigl(\hat\rho_i(\mu^L_t)-\bar\rho_{\alpha,\tilde V_t}(i)\bigr)\;\!
  \partial_t\tilde V_t(\tfrac iL)\;\!\mathrm dt\ge 0.
\end{multline}

\subsubsection{Macroscopic Action}
\label{sec:macro_action}

We now establish some properties of $\mathbb{A}$, as defined in~\eqref{eqn:action_functional}. If $\mathbb
A\bigl((\pi_t)_{t\in[0,T]}\bigr)<\infty$ one can show that
\begin{multline}
  \label{eqn:A_finite}
  \mathbb A\bigl((\pi_t)_{t\in[0,T]}\bigr)
  =\frac 12\bigl[\mathcal F_\alpha^V(\rho_T) - \mathcal F_\alpha^V(\rho_0)\bigr] \\
  + \frac 14\int_0^T \bigl(\|\dot\rho_t\|_{-1,\chi(\rho_t)}^2 + \|\Delta\phi(\rho_t) 
  +\nabla\cdot(\chi(\rho_t)\nabla V)\|_{-1,\chi(\rho_t)}^2\bigr)\;\!\mathrm dt,
\end{multline}
see Proposition~\ref{prop:e_rep} and Proposition~\ref{prop:e_star_rep}. For a definition of the norm $\|\cdot
\|_{-1,\chi(\rho_t)}$ (and the associated inner product $\langle \cdot,\cdot\rangle _{-1,\chi(\rho_t)}$) we also refer to
Section~\ref{sec:regularity} below.

Note that $\mathbb A((\pi_t)_{t\in[0,T]})$ as defined here might in general be negative.  A sufficient condition for
non-negativity of $\mathbb A((\pi_t)_{t\in[0,T]})$ is ensured by the validity of the following chain rule, which can be seen as a
macroscopic counterpart to~\eqref{eqn:free_energy_diff} for potentials constant in time. A \emph{formal} calculation yields for
$0\le t_1< t_2\le T$ the chain rule
\begin{equation}
  \label{eqn:chain_rule_formal}
  \mathcal F_\alpha^V(\rho_{t_2}) - \mathcal F_\alpha^V(\rho_{t_1})
  =\int_{t_1}^{t_2} \bigl\langle \dot\rho_t, \frac{\delta \mathcal F_\alpha^V}{\delta \rho_t}\bigr\rangle \;\!\mathrm dt
  =-\int_{t_1}^{t_2} \langle \dot\rho_t, \Delta\phi(\rho_t)+\nabla\cdot(\chi(\rho_t)\nabla V)\rangle _{-1,\chi(\rho_t)}\;\!\mathrm dt.
\end{equation}
Combined with~\eqref{eqn:A_finite} this allows us to (formally!) rewrite the macroscopic action functional~\eqref{eqn:A_finite} as
\begin{equation}
  \label{eqn:A_alternative}
  \mathbb A\bigl((\pi_t)_{t\in[0,T]}\bigr)
  =\frac 14\int_0^T \bigl\|\dot\rho_t -\Delta\phi(\rho_t) -\nabla\cdot(\chi(\rho_t)\nabla V)\bigr\|_{-1,\chi(\rho_t)}^2\;\!\mathrm dt.
\end{equation}

In Section~\ref{sec:chain_rule_for_F_V} we summarise some geometrical properties of the relevant function spaces and we establish
sufficient conditions for the chain rule:

\begin{theorem}
  \label{thm:chain-rule}
  Let the assumptions from Section~\ref{sec:new-assump} hold and additionally assume that $\chi'(a)\geq C_*$ for all
  $a\in(0,N_{\rm max}]$ (for some $ C_*>0$). If $d>1$, then further assume that the free energy density $f$ satisfies the McCann
  condition for geodesic convexity (stated in Equation~\eqref{eqn:mccann} below).  Then any path $(\pi_t)_{t\in[0,T]}$ with
  $\mathbb{A}((\pi_t)_{t\in[0,T]})<\infty$ and $\mathcal F_\alpha^V(\rho_0)<\infty$ satisfies the identities in
  Equation~\eqref{eqn:chain_rule_formal}.
\end{theorem}

Note that the McCann condition is always satisfied in one spatial dimension (where it reduces to convexity of $f$).  We further
stress that in Macroscopic Fluctuation Theory the validity of the chain rule is implicitly assumed by Equation (2.15)
in~\cite{Bertini2015a}, which relates the large deviation rate for a forward path to its time-reversed counterpart.

\subsection{Sufficient Conditions for Local Equilibration}
\label{sec:suff_local_eqm}

The following theorem, proved in Section~\ref{sec:proof_xx_new} below, yields a sufficient condition for the local equilibration
discussed in Section~\ref{sec:local_eqm} in terms of the free energy~\eqref{eqn:free_energy_alpha} of the initial condition and
the action functional~\eqref{eqn:rel_ent_2a}.

\begin{theorem}
  \label{thm:repl_lemma_thm}
  Let $(P_L)_{L\in\mathbb N}$ be as in Section~\ref{sec:path_meas} with densities $(\mu^L_t)_{t\in[0,T]}$, for $L\in\mathbb N$,
  and associated path measures $(Q_L)_{L\in\mathbb N}$ on $\mathcal D([0,T];\mathcal M_+(\Lambda))$.  Assume there exist $V\in
  C^2(\Lambda;\mathbb R)$ and $\alpha\in[0,N_{\max})$ such that
  \begin{equation}
    \label{eqn:initial_bound_free_energy_xx}
    \limsup_{L\to\infty}\frac 1{L^d}\mathcal F_{L,\alpha}^{V}(\mu^L_0)<\infty
  \end{equation}
  and $\tilde V\in C^{1,2}([0,T]\times\Lambda;\mathbb R)$ such that
  \begin{equation}
    \label{eqn:suff_local_eqm_action}
    \limsup_{L\to\infty}\frac 1{L^d}\mathbb A_L^{\tilde V}\bigl(Q_L\bigr)<\infty.
  \end{equation} 
  Then $(\mu^L_{[0,T]})_{L\in\mathbb N}$ satisfies the local equilibrium assumption,~\eqref{eqn:local_equilibrium}
  and~\eqref{eqn:local_equilibrium_2}.  Moreover, Equations~\eqref{eqn:initial_bound_free_energy_xx}
  and~\eqref{eqn:suff_local_eqm_action} are independent of $V$, $\tilde V$ and $\alpha$, such that these conditions can
  equivalently be stated as $\limsup_{L\to\infty} L^{-d}\mathcal H(Q_L | Q_{\nu_\alpha})<\infty$, where $Q_{\nu_\alpha}$ denotes
  the measure on $\mathcal D([0,T];\Omega_L)$ with marginals equal to $\nu_\alpha$, in the sense that $(X_t)_\#Q_{\nu_\alpha} =
  (\Theta_L)_\#\nu_\alpha$ for all $t\in[0,T]$.
\end{theorem}

\subsection{Particle Systems on Hydrodynamic Scale}
\label{sec:main_section}

We now present our main results. We consider sequences of path measures $(Q_L^V)_{L\in\mathbb N}$ and $(Q_L)_{L\in\mathbb N}$ on
$\mathcal D([0,T];\mathcal M_+(\Lambda))$, as defined in Section~\ref{sec:path_meas}, as well as the corresponding sequences
$(P_L^V)_{L\in\mathbb N}$ and $(P_L)_{L\in\mathbb N}$. We define $Q^*$ as a (possibly non-unique) limit point of the sequence of
observed processes $(Q_L)_{L\in\mathbb N}$ and we establish various properties of this limit. The physical idea is that the path
on which $Q^*$ is supported is a \emph{candidate} for the hydrodynamic limit for the reference process $(Q_L^V)_{L\in\mathbb N}$.
By analysing the large-$L$ behaviour of the microscopic action $\mathbb{A}_L^V(Q_L)$, the aim is to show that the only admissible
candidate path is the true hydrodynamic limit.  For specific examples, see Section~\ref{sec:examples}, below.

\subsubsection{Assumptions for Scaling Limits}
\label{sec:assumptions_scaling_limit}

To apply the results of this section to a specific interacting particle system (reference process), several assumptions have to be
satisfied.  We assume that the conditions given in Section~\ref{sec:new-assump} have been verified.  We assume also that the
initial distributions $(\mu^L_0)_{L\in\mathbb N}$ of $(P_L^V)_{L\in\mathbb N}$ converge to a fixed density $\rho_0\in\mathcal
L^1(\Lambda;[0,\infty))$ in the sense that $(\Theta_L)_\#\mu^L_0\to \delta_{\pi_0}$ with $\pi_0(\mathrm du)=\rho_0(u)\mathrm
  du$. For the rest of this Section~\ref{sec:main_section}, we fix $\alpha$ uniquely by requiring that $\int_\Lambda
  \rho_0(u)\;\!\mathrm du = \int_\Lambda \bar\rho_{\alpha,V}(u)\;\!\mathrm du$.

Further, we assume that the observed processes $(Q_L)_{L\in\mathbb N}$ are relatively
compact~\cite{Billingsley1999a,Kipnis1999a}. Then there is a measure $Q^*$ on $\mathcal D([0,T];\mathcal M_+(\Lambda))$ and a
subsequence of $(Q_L)_{L\in\mathbb N}$ converging to $Q^*$ (such that the marginal at time $t=0$ satisfies
$(X_0)_\#Q^\ast=\delta_{\pi_0}$). Finally, we assume that the measure $Q^*$ is concentrated on paths that are absolutely
continuous with respect to the Lebesgue measure,
\begin{equation}
  \label{eqn:concentration_of_measure}
  Q^*\Bigl((\pi_t)_{t\in[0,T]}\in\mathcal D([0,T];\mathcal M_+(\Lambda)) 
  : \pi_t(\mathrm du) = \rho_t(u)\;\! \mathrm du\textrm{ for a.a.}~t\in[0,T]\Bigr)=1.
\end{equation}
We note that the paths in~\eqref{eqn:concentration_of_measure} satisfy $\rho_t \in \mathcal L^1(\Lambda;[0,\infty))$.  Moreover,
if $N_{\max}<\infty$, then clearly also $\rho_t\le N_{\max}$ a.e.~on $\Lambda$ for almost all $t\in[0,T]$.  However, the limit
$Q^*$ is not assumed to be unique: there could exist other subsequences of $(Q_L)_{L\in\mathbb N}$ with different limits.

Given a specific model, the compactness of the sequence $(Q_L)_{L\in\mathbb N}$ and the support on absolutely continuous
paths~\eqref{eqn:concentration_of_measure} often follow from~\eqref{eqn:initial_bound_free_energy_xx} in combination with an
assumptions on the transition rates of the particle system. This is the case for the examples considered in
Section~\ref{sec:examples} below.

\subsubsection{Comparison with classical proofs of the Hydrodynamic Limit}
\label{sec:hydro_discussion}

To provide context for our analysis, we briefly summarise the classical approach to hydrodynamic limits.  Here, we consider
separately the observed process and the reference process, but the classical approach takes $(P_L)_{L\in\mathbb
  N}=(P_L^V)_{L\in\mathbb N}$.  The task of proving a hydrodynamic limit for $(Q_L)_{L\in\mathbb N}$ then consists of
characterising all limiting distributions. The first step is to establish relative
compactness~\cite{Billingsley1999a,Kipnis1999a}, which ensures the existence of a (possibly non-unique) limit $Q^*$.  One then
shows that $Q^*$ is unique and that it is concentrated on a single path $(\rho_t)_{t\in[0,T]}$
(i.e.~$Q^*=\delta_{(\pi_t)_{t\in[0,T]}}$ and $\pi_t(\mathrm du)=\rho_t(u)\mathrm du$ for almost all $t\in[0,T]$). This general
approach includes both the entropy method and the relative entropy method~\cite{Kipnis1999a}: note that it \emph{first}
establishes that $Q^*$ is supported on weak solutions to~\eqref{eqn:weak_sol_xxx} and \emph{then} uses a uniqueness result for
this solution to infer that $Q^*$ is supported on this unique solution, see e.g.~\cite[Chapter~4]{Kipnis1999a}.

Our approach here differs in two main points: We consider an observed process that is different from the reference process
($P_L\neq P_L^V$ in general) and we assume that the sequence $(Q_L)_L$ has a unique limiting distribution $Q^*$ that is
concentrated on a single path, as in~\eqref{eqn:concentration_of_measure}.  (As a special case, one may take $P_L=P_L^V$, under
the assumption that the hydrodynamic limit exists, but the following results are not restricted to this case.)  These assumptions
mean that the results in this work do not prove the existence of a hydrodynamic limit, neither for the observed process nor the
reference process.  Rather, they assume the existence of such a limit, and they establish properties of the associated path
$(\pi_t)_{t\in[0,T]}$ and its macroscopic action $\mathbb A\bigl((\pi_t)_{t\in[0,T]}\bigr)$.

\subsubsection{Convergence of Free Energy and Action for Deterministic Limits}

The following first main theorem yields regularity results for $(P_L)_{L\in\mathbb N}$ under the assumptions of
Section~\ref{sec:assumptions_scaling_limit} and those of Theorem~\ref{thm:repl_lemma_thm}. In particular, it shows that the
macroscopic action (and its individual contributions) are asymptotically dominated by their (more detailed) microscopic
counterparts.

\begin{theorem}[Regularity of the limit and asymptotic lower bounds]
  \label{thm:lower-bound-thm}
  Let $(P_L)_{L\in\mathbb N}$ be a sequence as in Section~\ref{sec:assumptions_scaling_limit}, with density and current
  $(\mu^L_t,\jmath^L_t)_{t \in[0,T]}$, for $L\in\mathbb N$. We suppose that the associated sequence $(Q_L)_{L\in\mathbb N}$ has a
  unique limit point $Q^*=\delta_{(\pi_t)_{t\in[0,T]}}$ for some $(\pi_t)_{t\in[0,T]}\in \mathcal D([0,T];\mathcal M_+(\Lambda))$
  and that the initial condition is well prepared in the sense that the free energies converge
  (cf.~\cite{Serfaty2011a,Fathi2016a,Mielke2016b})
  \begin{equation}
    \label{eqn:liminf_1}
    \lim_{L\to\infty} \frac 1{L^d}\mathcal F_{L,\alpha}^V\bigl(\mu^L_0\bigr) = \mathcal F_\alpha^V(\rho_0).
  \end{equation} 
  Further assume that $(Q_L)_{L\in\mathbb N}$ satisfies~\eqref{eqn:suff_local_eqm_action} for $\tilde V_t=V$, such that
  \begin{equation}
    \label{eqn:liminf_1_1}
    \limsup_{L\to\infty}\frac 1{L^d}\mathbb A_L^V\bigl(Q_L\bigr)<\infty.
  \end{equation}
  Then $(\pi_t)_{t\in[0,T]}$ is narrowly continuous, i.e.~$(\pi_t)_{t\in[0,T]}\in C([0,T];\mathcal M_+(\Lambda))$ and the action
  satisfies the lower bound
  \begin{equation}
    \label{eqn:conv_A}
    \liminf_{L\to\infty}\frac 1{L^d}\mathbb A_L^V\bigl(Q_L\bigr)\ge
    \mathbb A ((\pi_t)_{t\in[0,T]}).
  \end{equation}
  Further, the free energy satisfies for all $t\in[0,T]$
  \begin{equation}
    \label{eqn:T_limit_free_energy_liminf}
    \liminf_{L\to\infty} \frac 1{L^d}\mathcal F_{L,\alpha}^V\bigl(\mu^L_t\bigr)\ge\mathcal F_\alpha^V(\rho_t),
  \end{equation}
  as well as
  \begin{equation}
    \label{eqn:lower_bound_psi_thm}
    \liminf_{L\to\infty} \frac 1{L^d}\int_0^T \Psi_L\bigl(\mu^L_t,\jmath^L_t\bigr) \;\!\mathrm dt 
    \ge \frac 12\int_0^T \|\dot\rho_t\|_{-1,\chi(\rho_t)}^2\;\!\mathrm dt
  \end{equation}
  and
  \begin{equation}
    \label{eqn:lower_bound_psi_star_thm}
    \liminf_{L\to\infty} \frac 1{L^d}\int_0^T \Psi_L^\star\bigl(\mu^L_t,F^V(\mu^L_t)\bigr) \;\!\mathrm dt 
    \ge \frac 12\int_0^T \|\Delta\phi(\rho_t)+\nabla\cdot(\chi(\rho_t)\nabla V)\|_{-1,\chi(\rho_t)}^2\;\!\mathrm dt.
  \end{equation}
\end{theorem}

In this theorem, we see for the first time a connection between the non-quadratic microscopic functionals $\Psi$ and $\Psi^\star$
and their macroscopic quadratic counterparts, see~\eqref{eqn:lower_bound_psi_thm} and~\eqref{eqn:lower_bound_psi_star_thm}.

\begin{proof}
Note that the assumptions of Theorem~\ref{thm:repl_lemma_thm} are satisfied, so that the local equilibration
assumptions~\eqref{eqn:local_equilibrium} and~\eqref{eqn:local_equilibrium_2} hold.  The result~\eqref{eqn:conv_A} follows from
the representation of $\mathbb A_L^V$ in~\eqref{eqn:rel_ent_2b}, the definition of $\mathbb A$ in~\eqref{eqn:action_functional}
combined with~\eqref{eqn:liminf_1} and the following three inequalities (for which the proofs will be given in
Section~\ref{sec:assumptions_2}). Firstly, for the free energy at the final time $T$, we obtain from
Proposition~\ref{prop:free_energy_lower_bound_general} and the continuity of $X_T$ (the evaluation of the path at the final time
$t=T$) that
\begin{equation}
  \label{eqn:liminf_2}
  \liminf_{L\to\infty}\frac 1{L^d}\mathcal F_{L,\alpha}^V(\mu^L_T)\ge \mathcal F_\alpha^V(\rho_T).
\end{equation}
Secondly, 
\begin{equation}
  \label{eqn:lower_bound_psi}
  \liminf_{L\to\infty} \frac 1{L^d}\int_0^T \Psi_L\bigl(\mu^L_t,\jmath^L_t\bigr) \;\!\mathrm dt
  \ge \mathcal E\bigl((\rho_t)_{t\in[0,T]}\bigr),
\end{equation}
which follows from Proposition~\ref{prop:proof_psi_bound}, and thirdly
\begin{equation}
  \label{eqn:lower_bound_psi_star}
  \liminf_{L\to\infty} \frac 1{L^d}\int_0^T \Psi_L^\star\bigl(\mu^L_t,F^V(\mu^L_t)\bigr) \;\!\mathrm dt
  \ge \mathcal E^\star\bigl((\rho_t)_{t\in[0,T]}\bigr),
\end{equation}
which is proved in Proposition~\ref{prop:another_ineq}. Proposition~\ref{prop:e_rep} and Proposition~\ref{prop:e_star_rep} then
yield~\eqref{eqn:lower_bound_psi_thm} and~\eqref{eqn:lower_bound_psi_star_thm},
respectively. Proposition~\ref{lem:continuity_lemma} further shows that the path is 2-absolutely continuous in the Wasserstein
sense (see~\eqref{eqn:abs_cont_wasserstein} in Section~\ref{sec:regularity}), from which we can deduce the narrow continuity using
Lemma~\ref{lem:ac_equiv}. The inequality~\eqref{eqn:T_limit_free_energy_liminf} for the free energy at any time $t\in[0,T]$ then
follows from another application of Proposition~\ref{prop:free_energy_lower_bound_general}.
\end{proof}

It is instructive to consider Theorem~\ref{thm:lower-bound-thm} in the case where the observed process is equal to the reference
process $P_L=P_L^V$.  In this case the microscopic action $\mathbb A_L^V\bigl(Q_L\bigr)=0$ and the theorem has implications for
the hydrodynamic limit of the reference process, as follows. Either $Q^*$ does not concentrate on a single path, in which case the
theorem is inapplicable; or $Q^*$ does concentrate on a single path, and the theorem shows that the macroscopic action of that
path satisfies $\mathbb A ((\pi_t)_{t\in[0,T]})\leq0$, by~\eqref{eqn:conv_A}.  In the examples that we consider below, this
macroscopic action is zero, see below.

We consider a special case for the observed process $P_L$. We keep the reference process $P_L^V$ as outlined in
Section~\ref{sec:assumptions_scaling_limit} and consider for some (possibly time-dependent) potential $\tilde H\in
C^{1,2}([0,T]\times\Lambda;\mathbb R)$ the process $P_L=P^{\tilde V}_L$ for the potential $\tilde V_t=V + \tilde H_t$ as defined
in Section~\ref{sec:path_meas}.  Note that both processes have the same initial condition $\mu^L_0$ and their transition rates
$r^{V+\tilde H_t}$ and $r^V$ coincide up to a change of the external potential (i.e.~the functions $g_1$ and $g_2$
in~\eqref{eqn:rates_V} coincide for both processes).  We assume that the corresponding path measures $(Q^{V\!+\!\tilde
  H}_L)_{L\in\mathbb N}$ satisfy, as in Section~\ref{sec:hyrdo_eqn} above, a hydrodynamic limit with hydrodynamic equation
\begin{equation}
  \label{eqn:hydro_limit_time_dependent}
  \dot\rho_t =\Delta\phi(\rho_t) + \nabla\cdot(\chi(\rho_t)\nabla (V+\tilde H_t)).
\end{equation}
In this case one can improve the result~\eqref{eqn:conv_A} from Theorem~\ref{thm:lower-bound-thm} by showing that the action
functionals $\mathbb A_L^V(Q^{V\!+\!\tilde H}_L)$ converge, as described by the following second main theorem.

\begin{theorem}
  \label{thm:Q_H_limit}
  Assume that $P_L=P^{V\!+\!\tilde H}_L$ for some $\tilde H\in C^{1,2}([0,T]\times \Lambda;\mathbb R)$ and that
  $(P_L)_{L\in\mathbb N}$ satisfies the assumptions in Theorem~\ref{thm:lower-bound-thm}. Moreover, assume that the density of the
  path $(\pi_t)_{t\in[0,T]}$ is a weak solution to~\eqref{eqn:hydro_limit_time_dependent}, in the sense
  of~\eqref{eqn:weak_sol_xxx}. Then
  \begin{multline}
    \label{eqn:limit_action}
    \lim_{L\to\infty}\frac 1{L^d}\mathbb A_L^V\bigl(Q^{V\!+\!\tilde H}_L\bigr) 
    = \frac 14\int_0^T \bigl\|\nabla \tilde H_t\bigr\|_{\chi(\rho_t)}^2\;\!\mathrm dt\\
    =  \frac 14\int_0^T \bigl\|\dot\rho_t -\Delta\phi(\rho_t) -\nabla\cdot(\chi(\rho_t)\nabla V)\bigr\|_{-1,\chi(\rho_t)}^2\;\!\mathrm dt.
  \end{multline}
\end{theorem}

We postpone the proof of Theorem~\ref{thm:Q_H_limit} to Section~\ref{sec:theorem_2} below. See also Section~10
in~\cite{Kipnis1999a} for the specific calculations for the simple exclusion process, which can be seen as a special case of our
computations. We further stress that for measures of the form $(P_L^{\tilde V})_{L\in\mathbb N}$ the assumption
on~\eqref{eqn:liminf_1_1} in Theorem~\ref{thm:lower-bound-thm} is satisfied trivially, since
$\mathbb A_L^{\tilde V}\bigl(Q_L^{\tilde V}\bigr)=0$.

Theorem~\ref{thm:Q_H_limit} clarifies the relationship between the microscopic and macroscopic action functionals.  It shows how
the non-quadratic ($\Psi$-$\Psi^\star$) form of the microscopic action $\mathbb{A}_L^V$ converges to a (simpler) quadratic form,
when viewed on the macroscopic scale.  Of course, this convergence requires some information about the regularity of the path that
dominates $Q^*$: this comes from the assumption~\eqref{eqn:hydro_limit_time_dependent}.

Recall that the lower bound~\eqref{eqn:conv_A} in Theorem~\ref{thm:lower-bound-thm} and the limit~\eqref{eqn:limit_action} in
Theorem~\ref{thm:Q_H_limit} coincide (by~\eqref{eqn:A_alternative}) if and only if the chain rule~\eqref{eqn:chain_rule_formal}
holds. The validity of the chain rule~\eqref{eqn:chain_rule_formal} for the path $(\pi_t)_{t\in[0,T]}$ in
Theorem~\ref{thm:Q_H_limit} can be shown to be equivalent to the case where the limits
in~\eqref{eqn:T_limit_free_energy_liminf},~\eqref{eqn:lower_bound_psi_thm} and~\eqref{eqn:lower_bound_psi_star_thm} exist and all
three inequalities are equalities.

\begin{theorem}
  \label{thm:final_thm}
  Let the assumptions in Theorem~\ref{thm:Q_H_limit} hold. Further assume that $\mathcal F_\alpha^V$ satisfies the chain
  rule~\eqref{eqn:chain_rule_formal} for the path $(\rho_t)_{t\in[0,T]}$. Then the free energy converges for all $t\in[0,T]$,
  \begin{equation}
    \label{eqn:T_limit_free_energy}
    \lim_{L\to\infty} \frac 1{L^d}\mathcal F_{L,\alpha}^V\bigl(\mu^L_t\bigr) = \mathcal F_\alpha^V(\rho_t).
  \end{equation} 
  Moreover,
  \begin{equation}
    \label{eqn:T_limit_psi}
    \lim_{L\to\infty} \frac 1{L^d}\int_0^T\Psi_L\bigl(\mu^L_t,\jmath^L_t\bigr)\;\!\mathrm dt\\ 
    = \frac 12\int_0^T\|\dot\rho_t\|_{-1,\chi(\rho_t)}^2\;\!\mathrm dt
  \end{equation}
  and
  \begin{equation}
    \label{eqn:T_limit_psi_star}
    \lim_{L\to\infty} \frac 1{L^d}\int_0^T\Psi^\star_L\bigl(\mu^L_t,F^V(\mu^L_t)\bigr)\;\!\mathrm dt\\ 
    = \frac 12\int_0^T \|\Delta\phi(\rho_t)+\nabla\cdot(\chi(\rho_t)\nabla V)\|_{-1,\chi(\rho_t)}^2\;\!\mathrm dt.
  \end{equation}
  Also the opposite implication holds: If~\eqref{eqn:T_limit_free_energy},~\eqref{eqn:T_limit_psi}
  and~\eqref{eqn:T_limit_psi_star} are satisfied, then $\mathcal F_\alpha^V$ satisfies the chain
  rule~\eqref{eqn:chain_rule_formal} for $(\rho_t)_{t\in[0,T]}$.
\end{theorem}

\begin{proof} 
This proof is similar to calculations performed in~\cite{Kosygina2001a} and~\cite{Fathi2016a}, where the authors
establish~\eqref{eqn:T_limit_free_energy} for the hydrodynamic limit of the simple exclusion process.  Note
that~\eqref{eqn:limit_action},~\eqref{eqn:A_finite},~\eqref{eqn:A_alternative} and the chain rule~\eqref{eqn:chain_rule_formal}
imply
\begin{multline*}
  \qquad\lim_{L\to\infty} \frac 1{L^d}\biggl(\mathcal F_{L,\alpha}^V\bigl(\mu^L_T\bigr) 
  + \int_0^T \Psi_L\bigl(\mu^L_t,\jmath^L_t\bigr)\;\!\mathrm dt 
  + \int_0^T \Psi^\star_L\bigl(\mu^L_t,F^V(\mu^L_t)\bigr)\;\!\mathrm dt\biggr)\\ 
  = \mathcal F_\alpha^V(\rho_T)+\frac 12\int_0^T \|\dot\rho_t\|_{-1,\chi(\rho_t)}^2\;\!\mathrm dt 
  + \frac 12\int_0^T \|\Delta\phi(\rho_t)+\nabla\cdot(\chi(\rho_t)\nabla V)\|_{-1,\chi(\rho_t)}^2\;\!\mathrm dt.\qquad
\end{multline*}
We apply the inequality
$\limsup_{n\to\infty}(a_n+b_n+c_n)\ge \limsup_{n\to\infty} a_n + \liminf_{n\to\infty} b_n + \liminf_{n\to\infty}c_n$ to the
expression on the left hand side to obtain the inequality
\begin{equation*}
  \limsup_{L\to\infty} \frac 1{L^d}\mathcal F_{L,\alpha}^V\bigl(\mu^L_T\bigr)\le\mathcal F_\alpha^V(\rho_T).
\end{equation*}
The result for an arbitrary time $t\in[0,T]$ then follows for repeating the above proof for the time interval $[0,t]$. The
remaining two limits~\eqref{eqn:T_limit_psi} and~\eqref{eqn:T_limit_psi_star} follow in a similar way by a slight modification of
the above steps.

For the opposite implication, we assume that~\eqref{eqn:T_limit_free_energy},~\eqref{eqn:T_limit_psi}
and~\eqref{eqn:T_limit_psi_star} hold. In this case we have
\begin{multline*}
  \frac 12\int_0^T \bigl\|\dot\rho_t -\Delta\phi(\rho_t) -\nabla\cdot(\chi(\rho_t)\nabla V)\bigr\|_{-1,\chi(\rho_t)}^2\;\!\mathrm dt\\
  = \mathcal F_\alpha^V(\rho_T)-\mathcal F_\alpha^V(\rho_0)+\frac 12\int_0^T \|\dot\rho_t\|_{-1,\chi(\rho_t)}^2\;\!\mathrm dt \\
  + \frac 12\int_0^T \|\Delta\phi(\rho_t)+\nabla\cdot(\chi(\rho_t)\nabla V)\|_{-1,\chi(\rho_t)}^2\;\!\mathrm dt,
\end{multline*}
which is equivalent to~\eqref{eqn:chain_rule_formal} for $t_1=0$ and $t_2=T$. Repeating the above steps for $[0,t]$ (for any
$t\in[0,T]$) then finishes the proof.
\end{proof}

\paragraph{Remark on Chain Rule} In summary, we have seen that there are at least three ways to verify the chain
rule~\eqref{eqn:chain_rule_formal}. One way is to prove the assumptions of Theorem~\ref{thm:chain-rule}. Alternatively, one can
derive a Large Deviation Principle, as in Macroscopic Fluctuation Theory (cf.~the discussion below Theorem~\ref{thm:chain-rule});
or one can directly calculate the limits in Theorem~\ref{thm:final_thm}.

Now recall the case where the observed process and the reference process coincide, $P_L=P_L^V$.  One sees
that~\eqref{eqn:T_limit_free_energy}--\eqref{eqn:T_limit_psi_star} in Theorem~\ref{thm:final_thm} are similar
to~\eqref{eqn:T_limit_free_energy_liminf}--\eqref{eqn:lower_bound_psi_star_thm} in Theorem~\ref{thm:lower-bound-thm}, but
Theorem~\ref{thm:final_thm} is stronger, in that the limits have been shown to exist.  To prove this, the additional
assumption~\eqref{eqn:hydro_limit_time_dependent} was required, as well as~\eqref{eqn:chain_rule_formal}.  For the example systems
considered below, these assumptions can be proven by other means.  This establishes that the macroscopic action $\mathbb A
((\pi_t)_{t\in[0,T]})$ is non-negative, as long as the density $\rho$ associated to $\pi$ is a solution
of~\eqref{eqn:hydro_limit_time_dependent}, for some $\tilde H$.  In this case one sees that the hydrodynamic limit of the
reference system can be characterised as the unique zero of $\mathbb A$, within this class of paths.

Moreover, the quadratic structure of $\mathbb{A}$ together with the macroscopic chain rule means that the minimiser of
$\mathbb{A}$ can be identified as a gradient flow for the free energy.  Such gradient flows are widespread in macroscopic
descriptions of physical systems: we speculate that the structure presented here is similarly general.  That is, it is natural to
expect gradient flows as macroscopic descriptions of physical systems whose microscopic descriptions are reversible Markov chains,
because the non-quadratic $\Psi$-$\Psi^\star$ form of the microscopic action often converges to a quadratic functional on the
macroscopic scale.

\subsection{Examples}
\label{sec:examples}

Standard examples of particle models described by the class of models in Section~\ref{sec:ips} are (i) the zero-range process
(ZRP) for which $\Omega_L=\mathbb N_0^{\mathbb T_L^d}$, and $g_1$ is a function that satisfies $g_1(0)=0$ and $g_2=1$; and (ii)
the (symmetric) simple exclusion process (SEP), where $\Omega_L=\{0,1\}^{\mathbb T_L^d}$, $g_1(n) = \mathbf 1_{\{n=1\}}$ and
$g_2(n) = \mathbf 1_{\{n=0\}}$; and (iii) the generalised exclusion processes, where $\Omega_L=\{0,\cdots,m\}^{\mathbb T_L^d}$,
$g_1(n) = \mathbf 1_{\{n\ge 1\}}$ and $g_2(n) = \mathbf 1_{\{n\le m\}}$ for some fixed $m\in\mathbb N$~\cite{Kipnis1999a}. The
latter is an example of a non-gradient system. We focus on the two gradient models ZRP and SEP, which have $\mathrm d(k)=g_1(k)$
and $\mathrm d(k)=k$, respectively.

\subsubsection{Zero-Range Process}

The ZRP satisfies the assumptions of Section~\ref{sec:new-assump} if we assume that the rates are strictly monotonically
increasing and sub-linear.  That is, we assume that there exists $g^*>0$ such that $0 < g_1(k+1)-g_1(k)\le g^*$. Since $g_1(0)=0$
we have $g_1(k)\le g^*k$.  The mobility for the ZRP is given by $\chi(a)=\phi(a)$, where $E_{\nu_\alpha}[g_1(\eta(0))] =
\phi(\alpha)$.  The reference measure is $\nu_{*,1}(n)=1/(\prod_{k=1}^n g(k))$ and the $\alpha$-dependent invariant distribution
is for $z(\phi(\alpha)):=\sum_{n=0}^\infty \phi(\alpha)^n \nu_{*,1}(n)$ given by
\begin{equation*}
  \nu_{\alpha,1}(\eta(0))=\frac{\phi(\alpha)^{\eta(0)}}{z(\phi(\alpha))}\nu_{*,1}(\eta(0)).
\end{equation*}
Finally, the free energy is
\begin{equation*}
  \mathcal F_\alpha^V(\rho)
  =\int_\Lambda \biggl[\rho (u)\log \biggl(\frac{\phi(\rho(u))}{\mathrm e^{-V(u)}\phi(\alpha)}\biggr) 
  - \log \biggl(\frac{z(\phi(\rho(u)))}{z(\mathrm e^{-V(u)}\phi(\alpha))}\biggr)\biggr]\;\!\mathrm du
\end{equation*}
for $f(a) = \rho\log\phi(a)-\log z(\phi(a))$ and $\bar\rho_{\alpha,V}(u)=\phi^{-1}(\mathrm e^{-V(u)}\phi(\alpha))$.

These considerations establish that Theorems~\ref{thm:lower-bound-thm} to~\ref{thm:final_thm} can be applied to the ZRP.  We now
consider the implications of these theorems for hydrodynamic limits.  We first compare the path measures for the ZRP (that is, the
sequence of $P^V_L$ indexed by $L$) with some sequence of path measures $P_L$ which concentrate on an absolutely continuous path
$(\pi_t)_{t\in[0,T]}$ and satisfies the assumptions of Theorem~\ref{thm:repl_lemma_thm}. In this case one may apply
Theorem~\ref{thm:lower-bound-thm}, which establishes an asymptotic lower bound on the rescaled microscopic action $L^{-d}\mathbb
A_L^V(Q_L)$.  If $(\pi_t)_{t\in[0,T]}$ is the hydrodynamic limit of the ZRP then $P_L^V$ has to concentrate on
$(\pi_t)_{t\in[0,T]}$, but one also has (in general) that $L^{-d}\mathbb A_L^V(Q_L^V)=0$. Hence, if $L^{-d}\mathbb A_L^V(Q_L)$ is
bounded away from zero then the path $(\pi_t)_{t\in[0,T]}$ associated to $P_L$ can be ruled out as a possible hydrodynamic limit.

In fact the hydrodynamic limit of the ZRP is known to be given by~\eqref{eqn:hydro_limit_time_dependent} with $\tilde H = 0$ (see
Section~5 in~\cite{Kipnis1999a}), in which case Theorem~\ref{thm:lower-bound-thm} bounds the macroscopic action by zero:
$\mathbb A ((\pi_t)_{t\in[0,T]})\le 0$.  However this bound is not yet sufficient to show that $P_L^V$ concentrates on
$(\pi_t)_{t\in[0,T]}$, so it does not prove the hydrodynamic limit.

We now restrict our consideration to measures of the form $P_L=P_L^{V\!+\!\tilde H}$ that concentrate on paths which
satisfy~\eqref{eqn:hydro_limit_time_dependent}, for some $\tilde H$. In this case, Theorem~\ref{thm:Q_H_limit} may be applied.
This establishes that the limit of $L^{-d}\mathbb A_L^V(Q_L^{V\!+\!\tilde H})$ exists. We moreover can verify the assumptions of
Theorem~\ref{thm:chain-rule} (at least for $d=1$) or alternatively rely on the existence of the pathwise LDP
(see~\cite{Benois1995a}), which shows that also Theorem~\ref{thm:final_thm} holds -- this establishes a lower bound $\mathbb
A((\pi_t)_{t\in[0,T]}) \geq 0$ for any path $(\pi_t)_{t\in[0,T]}$ that solves~\eqref{eqn:hydro_limit_time_dependent}, with some
$\tilde H$.  This means that $(\pi_t)_{t\in[0,T]}$ is only admissible as a candidate for the hydrodynamic limit of the ZRP, if it
is a (weak) solution to~\eqref{eqn:hydro_limit_time_dependent} with $\tilde H=0$ (otherwise one has the contradiction $0 =
\lim_{L\to\infty} L^{-d} \mathbb A_L^V(Q_L^V) = \mathbb A((\pi_t)_{t\in[0,T]}) > 0$).

\subsubsection{Simple Exclusion Process}

For the SEP the invariant reference measure is $\nu_{*,1}(0)=\nu_{*,1}(1)=1$ and the $\alpha$-dependent invariant product measure
are Bernoulli distributed $\nu_{\alpha,1}(\eta(0)) = \alpha^{\eta(0)}(1-\alpha)^{1-\eta(0)}$. The functions $\phi$ and $\chi$ are
given by $\phi(\alpha)=\alpha$ and $\chi(\alpha) = \alpha(1-\alpha)$. %, which are Lipschitz constant with $C_{\rm Lip}=1$.  
The free energy is given by
\begin{multline*}
  \mathcal F_\alpha^V(\rho)
  =\int_\Lambda \biggl[\rho(u) \log \biggl(\frac{\rho(u)}{\alpha\mathrm e^{-V(u)}}\biggr) 
    +(1-\rho(u))\log \biggl(\frac{1-\rho(u)}{1-\alpha}\biggr)\\
    +\log \Bigl(\alpha\mathrm e^{-V(u)}+(1-\alpha)\Bigr)\biggr]\;\!\mathrm du,
\end{multline*}
which is of the form~\eqref{eqn:1st_free_energy} for the free energy density $f(a) = a\log a + (1-a)\log(1-a)$ and the stationary
density is $\bar\rho_{\alpha,V}(u) = \alpha \mathrm e^{-V(u)}/(\alpha\mathrm e^{-V(u)}+(1-\alpha))$.

For the sequence $P_L^{V+\tilde H}$ the hydrodynamic limit is again given in~\eqref{eqn:hydro_limit_time_dependent}, which has for
suitable initial condition a unique weak solution (see Proposition~5.1 on page~273 in~\cite{Kipnis1999a}).  We can proceed as for
the ZRP and can establish (under suitable assumptions) that the results of Theorem~\ref{thm:lower-bound-thm} and
Theorem~\ref{thm:Q_H_limit} hold.

Note that this process does not satisfy the assumptions of Theorem~\ref{thm:chain-rule} (as the assumption $\chi'(a)\geq C_*$ is
not satisfied). Nonetheless, we can establish the chain rule~\eqref{eqn:chain_rule_formal} if the pathwise LDP holds (cf.~the
discussion at the end of Section~\ref{sec:act-fcts}). This was e.g.~proved in~\cite[Chapter 10]{Kipnis1999a} (see
also~\cite{Bertini2009c}), such that also in this case the results of Theorem~\ref{thm:final_thm} hold.

\section{Regularity of Paths and the Chain Rule}
\label{sec:regularity}

The main aim of this section is to prove Theorem~\ref{thm:chain-rule}. The central difficulty is that classical approaches to
establish chain rules in metric spaces rely on $\lambda$-convexity of the functional under consideration; this property is
delicate and apparently not sufficiently well understood in a context other than the classic (unweighted) Wasserstein setting. The
process considered here are, however, naturally linked to weighted Wasserstein spaces, where important elements of the classic
Wasserstein theory are still missing. We circumvent this problem by showing that while the classic Wasserstein space is not the
natural space for the processes we study, they can be cast in this setting. The analysis is then somewhat technical, but follows
largely arguments in~\cite{Ambrosio2008b}. The novel $\Psi$-$\Psi^\star$-structure is thus less relevant in this section than for
the proofs in Section~\ref{sec:proof_bounds}.

In the following, we consider paths with conserved volume, for which also the action is finite: $\mathbb
A((\rho_t)_{t\in[0,T]})<\infty$. Combined with $\mathcal F_\alpha^V(\rho_0)<\infty$ and~\eqref{eqn:action_functional}, this
implies that $\mathcal E((\rho_t)_{t\in[0,T]})<\infty$ and $\mathcal E^\star((\rho_t)_{t\in[0,T]})<\infty$.  We will see that the
former of the two implies regularity in time (that $(\rho_t)_{t\in[0,T]}$ is absolutely-continuous in the Wasserstein sense) and
the latter yields certain regularity in space (such that e.g.~the weak gradient $\nabla \phi(\rho)$ exists a.e.~in $\Lambda$).

The following steps are based on ideas from Section 4 in~\cite{Dawson1987a}. For a more recent and concise representation of the
following material, we refer to Appendices~D.5 and~D.6)~in~\cite{Feng2006a}. A discussion of similar content in terms of
interacting particle systems can e.g. be found in~\cite{Bertini2009c}.

For any topological space $\mathcal S$, we denote with $\mathscr D(\mathcal S;\mathbb R)=C^\infty_c(\mathcal S;\mathbb R)$ the
vector space of real-valued infinitely often differentiable and compactly supported functions on $\mathcal S$ and equip $\mathscr
D(\mathcal S;\mathbb R)$ with the usual topology for test functions, see e.g.~\cite[Appendix D.1]{Feng2006a}. Its topological
dual, the space of (Schwartz) distributions, will be denoted with $\mathscr D'(\mathcal S;\mathbb R)$. The application of
$g\in\mathscr D(\mathcal S;\mathbb R)$ to a distribution $\vartheta \in \mathscr D'(\mathcal S;\mathbb R)$ is denoted by
$\langle\vartheta,g\rangle$.

The Otto calculus yields a formal interpretation of $\mathcal M_+(\Lambda)$ as an infinite dimensional Riemannian manifold (see
for example Chapter~15 in~\cite{Villani2009a} or Section~8.1.2 in~\cite{Villani2003a}). For a measure $\pi\in\mathcal
M_+(\Lambda)$, one can define three isometric spaces $H^1_\pi(\Lambda;\mathbb R)$, $H^{-1}_{\pi}(\Lambda;\mathbb R)$ and $\mathcal
L_{\nabla,\pi}^2(\Lambda;\mathbb R^d)$, which all can play the role of the `tangent space' at $\pi$. We next give precise
definitions of all three spaces. For $h\colon\Lambda\to \mathbb R^d$, we define the norm $\|h\|_\pi^2:=\int_\Lambda
|h(u)|^2\pi(\mathrm du)$. For $g\in W^1_{\rm loc}(\Lambda;\mathbb R)$ this norm gives rise to the semi-norm
$\|g\|_{1,\pi}:=\|\nabla g\|_{\pi}$, where $\nabla g$ denotes the weak derivative of $g$. Since $\{g\in\mathscr D(\Lambda;\mathbb
R): \int_\Lambda g \;\!\mathrm du=0\}$ equipped with $\|\cdot\|_{1,\pi}$ is a normed space, we can define its completion to be
$H^1_\pi(\Lambda;\mathbb R)$. For $\vartheta\in\mathscr D'(\Lambda;\mathbb R)$ the dual norm, which is defined as
\begin{equation}
  \|\vartheta\|_{-1,\pi}^2:=\sup_{g\in H_\pi^1(\Lambda;\mathbb R)} \bigl(2\langle \vartheta,g\rangle - \|g\|_{1,\pi}^2\bigr),
\end{equation}
gives rise to $H^{-1}_\pi(\Lambda;\mathbb R):=\{\vartheta\in\mathscr D'(\Lambda;\mathbb R): \|\vartheta\|_{-1,\pi}<\infty\}$, the
dual of $H^1_\pi(\Lambda;\mathbb R)$. Note that $H^1_\pi(\Lambda;\mathbb R)$ is a Hilbert space (with inner product $\langle
\cdot,\cdot\rangle _{1,\pi}$ defined in the obvious way using the polarisation identity for inner products); it therefore is
reflexive, which implies the existence of a linear and isometric map from $H^1_\pi(\Lambda;\mathbb R)$ to
$H^{-1}_\pi(\Lambda;\mathbb R)$, formally given by $g\mapsto -\nabla\cdot (\pi \nabla g)$. The inner product on
$H^{-1}_\pi(\Lambda;\mathbb R)$ will be denoted with $\langle \cdot,\cdot\rangle _{-1,\pi}$. Finally, let $\mathcal
L_{\nabla,\pi}^2(\Lambda;\mathbb R^d)$ be the completion of $\{\nabla\zeta: \zeta\in \mathscr D(\Lambda;\mathbb R)\}$ with respect
to $\|\cdot\|_\pi$. It is then easy to see that $H_\pi^1(\Lambda;\mathbb R)$ is also isometric to $\mathcal
L_{\nabla,\pi}^2(\Lambda;\mathbb R^d)$ (cf.~page~379 in~\cite{Feng2006a}). We will denote the map from $H_\pi^1(\Lambda;\mathbb
R)$ to $\mathcal L_{\nabla,\pi}^2(\Lambda;\mathbb R^d)$ with $\nabla$.

For our purposes, the spaces $H^{-1}_{\pi}(\Lambda;\mathbb R)$ and $\mathcal L_{\nabla,\pi}^2(\Lambda;\mathbb R^d)$ yield the more
relevant representations. The two prominent cases that will appear in the following are $\pi(\mathrm du) = \rho(u)\mathrm du$ and
$\pi(\mathrm du) = \chi(\rho(u))\mathrm du$. In these cases we will identify the densities $\rho$ and $\chi(\rho)$ as measures and
write $H_\rho^1(\Lambda;\mathbb R)$ and $H_{\chi(\rho)}^1(\Lambda;\mathbb R)$ instead of $H_\pi^1(\Lambda;\mathbb R)$ (and similar
for the other spaces we just introduced).

\subsection{Regularity of Paths on the Hydrodynamic Scale}
\label{sec:reg_paths_hydro_yy}

Now, fix a path $(\pi_t)_{t\in[0,T]}\in \mathcal D([0,T];\mathcal M_+(\Lambda))$ that is absolutely continuous with respect to the
Lebesgue measure with density $(\rho_t)_{t\in[0,T]}$. We equip $C^{1,2}([0,T]\times\Lambda;\mathbb R)$ with the
$(\rho_t)_{t\in[0,T]}$ dependent semi-norm $G\mapsto (\int_0^T\|\nabla G_t\|_{\chi(\rho_t)}^2\mathrm dt)$\textsuperscript{$1/2$},
on which we define the two real valued linear operators
\begin{equation*}
  L_\mathcal E(G) := \int_\Lambda \rho_T G_T\;\!\mathrm du -\int_\Lambda \rho_0G_0\;\!\mathrm du
  - \int_0^T\int_\Lambda \rho_t\;\! \partial_t G_t\;\!\mathrm du \;\!\mathrm dt
\end{equation*}
and
\begin{equation*}
  L_{\mathcal E^\star}(G) := \int_0^T\int_\Lambda \phi(\rho_t)\nabla\cdot \nabla G_t\;\!\mathrm du\;\!\mathrm dt
  - \int_0^T\int_\Lambda \chi(\rho_t)\nabla V\cdot \nabla G_t\;\!\mathrm du\;\!\mathrm dt.
\end{equation*}
Note that these two operators coincide with the left and right hand side of~\eqref{eqn:weak_sol_xxx}, respectively. Moreover, the
corresponding operator norms are given by $\mathcal E((\rho_t)_{t\in[0,T]})$ in~\eqref{eqn:e_op_norm} and $\mathcal
E^\star((\rho_t)_{t\in[0,T]})$ in~\eqref{eqn:e_star_op_norm}, respectively (cf.~e.g.~\cite{Dawson1987a,Feng2006a}).

Under the assumptions of Theorem~\ref{thm:lower-bound-thm}, we have prior information on the regularity of the path
$(\rho_t)_{t\in[0,T]}$, i.e.~we can assume that $\mathcal E((\rho_t)_{t\in[0,T]}),\mathcal E^\star((\rho_t)_{t\in[0,T]})<\infty$
(such that $L_\mathcal E$ and $L_{\mathcal E^\star}$ are bounded linear operators).

Note that $L_{\mathcal E}$ and $L_{\mathcal E^\star}$ are both invariant under addition of a constant in the sense that
$L_{\mathcal E^\star}(G)=L_{\mathcal E^\star}(G+c)$ for any $c\in\mathbb R$.  We thus can (with slight abuse of notation) redefine
$L_\mathcal E$ and $L_{\mathcal E^\star}$ as operators on $\{\nabla G: G\in C^{1,2}([0,T]\times\Lambda;\mathbb R)\}$, equipped
with $\nabla G\mapsto(\int_0^T\|\nabla G_t\|_{\chi(\rho_t)}^2\mathrm dt)$\textsuperscript{$1/2$}, as
\begin{equation*}
  L_\mathcal E(\nabla G) := L_\mathcal E(G) \quad \textrm{and} \quad L_{\mathcal E^\star}(\nabla G):= L_{\mathcal E^\star}(G).
\end{equation*}
Let $\mathcal L^2_{\nabla,\chi}([0,T]\times\Lambda;\mathbb R^d)$ be the $(\rho_t)_{t\in[0,T]}$ dependent completion of $\{\nabla
G: G\in C^{1,2}([0,T]\times\Lambda;\mathbb R)\}$ with respect to $\nabla G\mapsto(\int_0^T\|\nabla G_t\|_{\chi(\rho_t)}^2\mathrm
dt)$\textsuperscript{$1/2$}. Note that if $h=(h_t)_{t\in[0,T]}\in \mathcal L^2_{\nabla,\chi}([0,T]\times\Lambda;\mathbb R^d)$,
then $h_t\in\mathcal L_{\nabla,\chi(\rho_t)}^2(\Lambda;\mathbb R^d)$ for a.a.~$t\in[0,T]$. In Section~\ref{sec:chain_rule_for_F_V}
we will also consider $\mathcal L^2_{\nabla,\rm id}([0,T]\times\Lambda;\mathbb R^d)$, where the norm is replaced with $\nabla
G\mapsto(\int_0^T\|\nabla G_t\|_{\rho_t}^2\mathrm dt)$\textsuperscript{$1/2$}.

Since $\mathcal E((\rho_t)_{t\in[0,T]}),\mathcal E^\star((\rho_t)_{t\in[0,T]})<\infty$ the Bounded Linear Transformation Theorem
(see e.g.~Theorem~I.6~in~\cite{Reed1980a}), allows us to extend $L_\mathcal E(\nabla G)$ and $L_{\mathcal E^\star}(\nabla G)$ to
bounded linear operators on $\mathcal L^2_{\nabla,\chi}([0,T]\times\Lambda;\mathbb R^d)$ with the same operator norms as
above. For $h\in\mathcal L^2_{\nabla,\chi}([0,T]\times\Lambda;\mathbb R^d)$ we have
\begin{equation*}
  L_\mathcal E(h) = \int_\Lambda \rho_T \nabla^{-1} h_T\;\!\mathrm du -\int_\Lambda \rho_0 \nabla^{-1}h_0\;\!\mathrm du
  - \int_0^T\int_\Lambda \rho_t\;\! \partial_t (\nabla^{-1}h_t)\;\!\mathrm du \;\!\mathrm dt , 
\end{equation*}
where $\nabla^{-1}$ denotes (for each $t\in[0,T]$) the isometric map from $\mathcal L_{\nabla,\chi(\rho_t)}^2(\Lambda;\mathbb
R^d)$ to $H_{\chi(\rho_t)}^1(\Lambda;\mathbb R)$. Further
\begin{equation*}
  L_{\mathcal E^\star}(h) = \int_0^T\int_\Lambda \phi(\rho_t)\nabla\cdot h_t\;\!\mathrm du\;\!\mathrm dt
  - \int_0^T\int_\Lambda \chi(\rho_t)\nabla V\cdot h_t\;\!\mathrm du\;\!\mathrm dt.
\end{equation*}
By Riesz' representation theorem (e.g.~Theorem~II.4~in~\cite{Reed1980a}), there exist unique elements
$v, w\in \mathcal L^2_{\nabla,\chi}([0,T]\times\Lambda;\mathbb R^d)$, with $v=(v_t)_{t\in[0,T]}$ and $w=(w_t)_{t\in[0,T]}$, for
which these two bounded operators can be represented by
\begin{equation}
  \label{eqn:op_riesz}
  L_\mathcal E(h) = \int_0^T \int_\Lambda \chi(\rho_t) v_t\cdot h_t \;\!\mathrm du\;\!\mathrm dt,\qquad L_{\mathcal E^\star}(h)
  = \int_0^T \int_\Lambda \chi(\rho_t)w_t\cdot h_t \;\!\mathrm du\;\!\mathrm dt.
\end{equation}
Substituting~\eqref{eqn:op_riesz} in~\eqref{eqn:e_op_norm} and~\eqref{eqn:e_star_op_norm} yields (c.f.~Lemma 4.8
in~\cite{Dawson1987a})
\begin{equation}
  \label{eqn:op_norms}
  \mathcal E\bigl((\rho_t)_{t\in[0,T]}\bigr) = \frac 12\int_0^T\|v_t\|_{\chi(\rho_t)}^2\;\!\mathrm dt,\qquad \mathcal
  E^\star\bigl((\rho_t)_{t\in[0,T]}\bigr) = \frac 12\int_0^T\|w_t\|_{\chi(\rho_t)}^2\;\!\mathrm dt.
\end{equation}

\begin{proposition}
  \label{prop:e_rep}
  Assume that $\mathcal E((\rho_t)_{t\in[0,T]})<\infty$ and that $\chi$ satisfies the assumptions of
  Section~\ref{sec:new-assump}. Then the weak time derivative of $\rho_t$, denoted $\dot\rho_t$, exists in
  $H^{-1}_{\chi(\rho_t)}(\Lambda;\mathbb R)$ for a.a.~$t\in[0,T]$. Moreover,
  \begin{equation}
    \label{eqn:e_minus_one}
    \mathcal E\bigl((\rho_t)_{t\in[0,T]}\bigr) = \frac 12\int_0^T \|\dot\rho_t\|_{-1,\chi(\rho_t)}^2 \;\!\mathrm dt.
  \end{equation}
\end{proposition}

\begin{proof}
Results of this kind are standard and we hence only sketch the proof. Consider the unique $v\in \mathcal
L^2_{\nabla,\chi}([0,T]\times\Lambda;\mathbb R^d)$ from~\eqref{eqn:op_riesz} and recall that $v_t\in\mathcal
L_{\nabla,\chi(\rho_t)}^2(\Lambda;\mathbb R^d)$ for a.a.~$t\in[0,T]$.

Following e.g.~Lemma 4.8 in~\cite{Dawson1987a} (see also~\cite{Duong2013b}), one shows that $\mathcal
E((\rho_t)_{t\in[0,T]})<\infty$ implies that $t\mapsto\langle\rho_t,\cdot\rangle$ is absolutely continuous in the sense of
distributions, such that the distributional derivative $\dot\rho_t\in\mathscr D'(\Lambda;\mathbb R)$ exists for
a.a.~$t\in(0,T)$. In our case, the latter satisfies for $G\in\mathscr D(\Lambda;\mathbb R)$ and a.a.~$t\in(0,T)$
\begin{equation}
  \label{eqn:chain_rule_2}
  \frac {\mathrm d}{\mathrm dt}\int_\Lambda \rho_t \;\! G\;\!\mathrm du
  =\langle \dot \rho_t,G\rangle = \int_\Lambda \chi(\rho_t) v_t\cdot \nabla G\;\!\mathrm du.
\end{equation}
Thus $\dot \rho_t = -\nabla\cdot (\chi(\rho_t) v_t)$ in the distributional sense for a.a.~$t\in(0,T)$, such that $v_t\in \mathcal
L_{\nabla,\chi(\rho_t)}^2(\Lambda;\mathbb R^d)$ can uniquely be identified with $\dot\rho_t$. Further the isometry from $\mathcal
L_{\nabla,\chi(\rho_t)}^2(\Lambda;\mathbb R^d)$ to $H_{\chi(\rho_t)}^{-1}(\Lambda;\mathbb R)$ (for a.a.~$t\in[0,T]$) implies that
$\dot\rho_t\in H_{\chi(\rho_t)}^{-1}(\Lambda;\mathbb R)$ and~\eqref{eqn:e_minus_one} also follows.
\end{proof}

Let $p\in[1,\infty]$. We say a path $(\pi_t)_{t\in[0,T]}$ is \emph{$p$-absolutely continuous (in the Wasserstein sense)}, if there
exists a function $m\in\mathcal L^p([0,T];\mathbb R)$, such that for any $0\le t_1<t_2\le T$
\begin{equation}
  \label{eqn:abs_cont_wasserstein}
  W_2(\pi_{t_1},\pi_{t_2}) \le \int_{t_1}^{t_2} m(s)\;\!\mathrm ds,
\end{equation}
where $W_2$ denotes the 2-Wasserstein distance~\cite{Villani2003a,Ambrosio2008b}. In this case, the metric derivative
(cf.~equation (1.1.3) in~\cite{Ambrosio2008b}) exists for a.a.~$t\in(0,T)$,
\begin{equation*}
  |\pi'_t|:= \limsup_{h\to 0} \biggl(\frac{W_2(\pi_{t},\pi_{t+h})}{h}\biggr) <\infty
\end{equation*}
and $t\mapsto |\pi_t'|$ is the minimal function that satisfies~\eqref{eqn:abs_cont_wasserstein}, see Theorem 1.1.2
in~\cite{Ambrosio2008b}. In other words, $(\pi_t)_{t\in[0,T]}$ is $p$-absolutely continuous if and only if the map $t\mapsto
|\pi_t'|$ is an element of $\mathcal L^p([0,T];\mathbb R)$.  From now on we consider the case $p=2$.

\begin{lemma}
  \label{lem:ac_equiv}
  A path $(\pi_t)_{t\in[0,T]}\in \mathcal D([0,T];\mathcal M_+(\Lambda))$ is 2-absolutely continuous if and only if there exists a
  vector field $\tilde v=(\tilde v_t)_{t\in[0,T]}$ with $\tilde v_t\in \mathcal L_{\nabla,\pi_t}^2(\Lambda;\mathbb R^d)$ and
  $\int_0^T \|\tilde v_t\|_{\pi_t}\;\!\mathrm dt<\infty$ that satisfies $\dot\pi_t + \nabla\cdot(\pi_t\tilde v_t) =0$ in the
  distributional sense for almost all $t\in[0,T]$.  In this case we have in particular $(\pi_t)_{t\in[0,T]}\in C([0,T];\mathcal
  M_+(\Lambda))$.
\end{lemma}

\begin{proof}
The result follows from a modification of Lemma 8.1.2 and Theorem 8.3.1 in~\cite{Ambrosio2008b} to the domain $\Lambda$.  Assume
first that $(\pi_t)_{t\in[0,T]}$ is 2-absolutely continuous. Then Theorem 8.3.1 implies that the continuity equation $\dot\pi_t +
\nabla\cdot(\pi_t\tilde v_t) =0$ holds for some $\tilde v_t$, which can, by Lemma 8.4.2 in~\cite{Ambrosio2008b}, without loss of
generality be chosen to satisfy $\tilde v_t\in \mathcal L_{\nabla,\pi_t}^2(\Lambda;\mathbb R^d)$.

For the opposite implication we assume that the continuity equation holds and that moreover $\int_0^T \|\tilde
v_t\|_{\pi_t}\;\!\mathrm dt<\infty$. An application of the H\"older inequality combined with
$\sup_{t\in[0,T]}\pi_t(\Lambda)<\infty$ ensures that $\int_0^T\int_\Lambda |\tilde v_t(u)|\;\!\pi_t(\mathrm du)\;\!\mathrm
dt<\infty$. Lemma 8.1.2 thus implies that the curve has a weakly continuous modification $(\tilde\pi_t)_{t\in[0,T]}\in
C([0,T];\mathcal M_+(\Lambda))$. Now, since every right-continuous path that admits a continuous modification already has to be
continuous, we have $(\pi_t)_{t\in[0,T]}=(\tilde\pi_t)_{t\in[0,T]}$. This allows us to apply the reverse implication of Theorem
8.3.1 to $(\pi_t)_{t\in[0,T]}$, which yields that $(\pi_t)_{t\in[0,T]}$ is 2-absolutely continuous.
\end{proof}

The Wasserstein distance $W_2$ has a fluid dynamical representation in terms of the Brenier-Benamou formula (compare
Equation~(8.0.3) in~\cite{Ambrosio2008b} and Section~8.1 in~\cite{Villani2003a}). The distance of two measures
$\pi,\hat\pi\in\mathcal M_+(\Lambda)$ with $\pi(\Lambda)=\hat\pi(\Lambda)>0$ is given by
\begin{equation*}
  W_2^2(\pi,\hat\pi) = \inf \biggl\{ \int_0^1 \|\tilde v_t\|_{\mu_t}^2 \;\!\mathrm dt~ \Big|~ \mu_0=\pi,~ \mu_1
  =\hat\pi,~ \dot \mu_t + \nabla\cdot (\mu_t \tilde v_t)=0\biggr\},
\end{equation*}
where the infimum is taken over all 2-absolutely continuous paths of measures $(\mu_t)_{t\in[0,T]}$ {and velocities $\tilde v_t
  \in \mathcal L_{\nabla,\mu_t}^2(\Lambda;\mathbb R^d)$ satisfying the continuity equation above}.

Let $(\pi_t)_{t\in[0,T]}$ be absolutely continuous with respect to the Lebesgue measure with density $(\rho_t)_{t\in[0,T]}$. We
say that $(\rho_t)_{t\in[0,T]}$ is 2-absolutely continuous if $(\pi_t)_{t\in[0,T]}$ is 2-absolutely continuous. Moreover, we will
identify densities with their associated measures. In particular, we write $W_2^2(\rho,\hat\rho)=W_2^2(\pi,\hat\pi)$ for
$\pi(\mathrm du)=\rho(u)\mathrm du$ and $\pi(\mathrm du)=\rho(u)\mathrm du$.

\begin{proposition}
  \label{lem:continuity_lemma}
  Assume that $\mathcal E((\rho_t)_{t\in[0,T]})<\infty$ and that $\chi$ satisfies the assumptions of
  Section~\ref{sec:new-assump}. Then $(\rho_t)_{t\in[0,T]}$ is 2-absolutely continuous in the Wasserstein sense.
\end{proposition}

\begin{proof}
We choose the time rescaling $\bar t = t(t_2-t_1)+t_1$ and set $\mu_t = \rho_{\bar t}$ and $\tilde v_t = (t_2-t_1)(\chi(\rho_{\bar
  t}) v_{\bar t})/\rho_{\bar t}$, such that $\dot \mu_t + \nabla\cdot (\mu_t \tilde v_t)=0$ by construction. We obtain for all
$0\le t_1<t_2\le T$
\begin{equation*}
  W_2^2(\rho_{t_1},\rho_{t_2})\le (t_2-t_1)\int_{t_1}^{t_2} \| (\chi(\rho_t)v_t)/\rho_t\|_{\rho_t}^2 \;\!\mathrm dt 
  \le  (t_2-t_1) \int_{t_1}^{t_2} C_{\rm Lip} \| v_t\|_{\chi(\rho_t)}^2 \;\!\mathrm dt<\infty,
\end{equation*}
such that the metric derivative satisfies for almost all $t\in[0,T)$
\begin{equation}
  \label{eqn:metric_deriv}
  |\rho'_t|= \limsup_{h\to 0} \biggl(\frac{W_2(\rho_{t},\rho_{t+h})}{h}\biggr) \le \sqrt{C_{\rm Lip}}\|v_t\|_{\chi(\rho_t)}.
\end{equation}
The square integrability of the right hand side now implies that $(\rho_t)_{t\in[0,T]}$ is 2-absolutely continuous.
\end{proof} 

\begin{proposition}
  \label{prop:e_star_rep}
  Assume that $\mathcal E^\star\bigl((\rho_t)_{t\in[0,T]}\bigr)<\infty$ and that $f,\phi$ and $\chi$ satisfy the assumptions of
  Section~\ref{sec:new-assump}. Then
  \begin{multline}
    \label{eqn:e_star_minus_one}
    \mathcal E^\star\bigl((\rho_t)_{t\in[0,T]}\bigr)
    = \frac12 \int_0^T \|\Delta\phi(\rho_t)+\nabla\cdot(\chi(\rho_t)\nabla V)\|_{-1,\chi(\rho_t)}^2 \;\!\mathrm dt\\
    = \frac12 \int_0^T \|f''(\rho_t)\nabla \rho_t+\nabla V\|_{\chi(\rho_t)}^2 \;\!\mathrm dt.
  \end{multline}
\end{proposition}

\begin{proof}
$\mathcal E^\star\bigl((\rho_t)_{t\in[0,T]}\bigr)<\infty$ implies that the distributional derivative of $\phi(\rho_t)\in \mathcal
  L^1_{\rm loc}(\Lambda;\mathbb R)$ satisfies $\nabla\phi(\rho_t)\in \mathcal L^1_{\rm loc}(\Lambda;\mathbb R^d)$ for
  a.a.~$t\in[0,T]$ (cf.~Appendix D.6~in~\cite{Feng2006a}). Equivalently, $\phi(\rho_t)\in W^{1,1}_{\rm loc}(\Lambda;\mathbb R)$
  for a.a.~$t\in[0,T]$. The first identity in~\eqref{eqn:e_star_minus_one} can be established as in Appendix
  D.6~in~\cite{Feng2006a} (for the choice $\mu(\mathrm du) = \chi(\rho_t(u))\mathrm du$). We turn to the second identity. Since
  $\phi^{-1}$ is continuously differentiable with bounded derivative, we obtain by the chain rule for functions in $W^{1,1}_{\rm
    loc}(\Lambda;\mathbb R)$ with bounded derivative (see e.g.~Theorem~4~(ii) in~\cite{Evans1992a}) that also $\nabla \rho_t \in
  \mathcal L^1_{\rm loc}(\Lambda;\mathbb R)$, and thus $\rho_t \in W^{1,1}_{\rm loc}(\Lambda;\mathbb R)$, for almost all
  $t\in[0,T]$. The derivative is for almost all $u\in\Lambda$ given by
\begin{equation}
  \label{eqn:sobolev_chain_rule_0}
  \nabla \rho_t(u) = (\phi^{-1})'(\phi(\rho_t(u)))\nabla \phi(\rho_t(u)) = \frac{\nabla \phi(\rho_t(u))}{\phi'(\rho_t(u))},
\end{equation}
where the last identity follows from the Implicit Function Theorem. Multiplying with $\phi'(\rho_t)$ and using the local Einstein
relation~\eqref{eqn:local_einstein_rel} we obtain that almost everywhere
\begin{equation}
  \label{eqn:sobolev_chain_rule}
  \nabla \phi(\rho_t) = \phi'(\rho_t)\nabla \rho_t = \chi(\rho_t) f''(\rho_t)\nabla \rho_t.
\end{equation}
Combined with $w$ in~\eqref{eqn:op_norms}, we have for any $G\in\mathscr D(\Lambda;\mathbb R)$ and almost all $t\in[0,T]$ that
\begin{equation*}
  \int_\Lambda \chi(\rho_t) w_t\cdot \nabla G \;\!\mathrm du
  =\int_\Lambda \bigl(\nabla \phi(\rho_t)+\chi(\rho_t)\nabla V\bigr)\cdot \nabla G\;\!\mathrm du
  =\int_\Lambda \chi(\rho_t) [f''(\rho_t)\nabla\rho_t+\nabla V]\cdot \nabla G\;\!\mathrm du
\end{equation*}
such that we can identify $w_t=f''(\rho_t)\nabla\rho_t+\nabla V$. Substituting this identity in~\eqref{eqn:op_norms} yields the
final result.
\end{proof}

\subsection{Chain Rule for the Free Energy}
\label{sec:chain_rule_for_F_V}

In this section, we prove Theorem~\ref{thm:chain-rule}, which establishes rigorously the validity of the macroscopic chain
rule~\eqref{eqn:chain_rule_formal}, for which we so far gave only a formal derivation. Consider a given path
$(\rho_t)_{t\in[0,T]}$ that satisfies $\mathbb A((\rho_t)_{t\in[0,T]})<\infty$.  We restrict ourselves to densities $\rho,\hat\rho
\in \mathcal L^1(\Lambda;[0,\infty))$ s.t.~$\int_\Lambda\rho\;\!\mathrm du=\int_\Lambda\hat \rho\;\!\mathrm du>0$ and continue to
  identify densities with measures. The constant volume implies that free energy differences do not depend on $\alpha$. Indeed,
  defining $\mathcal F(\rho) := \int_\Lambda f(\rho(u))\mathrm du$ and $\mathcal V(\rho) := \int_\Lambda V(u)\rho(u)\mathrm du$
  (for $V\in C^2(\Lambda;\mathbb R)$), we can define an $\alpha$-independent modification of the free energy
\begin{equation}
  \label{eqn:mod_free_energy}
  \mathcal F^V(\rho) :=  \mathcal F(\rho) + \mathcal V(\rho),
\end{equation}
which is (with~\eqref{eqn:1st_free_energy}) easily seen to satisfy $\mathcal F_\alpha^V(\hat\rho)-\mathcal
F_\alpha^V(\rho)=\mathcal F^V(\hat\rho)-\mathcal F^V(\rho)$.

We assume that $f\in C^2([0,\infty);\mathbb R)$ satisfies the assumptions in Section~\ref{sec:new-assump}, such that the
  functional $\mathcal F\colon \mathcal L^1(\Lambda;[0,\infty))\to (-\infty,\infty]$ is proper and lower-semicontinuous (see
  Remark 9.3.8 in~\cite{Ambrosio2008b}).  Note that for $N_{\max}=\infty$ the assumption $\lim_{r\to N_{\max}}f'(r)=\infty$
  implies super linearity of $f$.

We set
\begin{equation*}
  L_f(a):=a f'(a)-f(a) = \int_0^a r f''(r) \;\!\mathrm dr
\end{equation*}
and note the similarity to $\phi(a) = \int_0^a \phi'(r)\;\!\mathrm dr = \int_0^a \chi(r) f''(r)\;\!\mathrm dr$ (where we again
used the local Einstein relation~\eqref{eqn:local_einstein_rel}); in particular $L_f'(a)/a = f''(a)= \phi'(a)/\chi(a)$. The
quantity $L_f$ is sometimes referred to as a `pressure' function due to its relation to the thermodynamic pressure in classical
thermodynamics, see e.g.~Remark 5.18 (ii) in~\cite{Villani2003a}.

We denote the (2-)Wasserstein distance between $\rho$ and $\hat\rho$ with $W_2(\rho,\hat\rho)$. A constant speed geodesic
(connecting $\rho$ to $\hat\rho$) is a curve $(\rho_t)_{t\in[0,1]}$ such that ($\rho_0=\rho$, $\rho_1=\hat\rho$ and)
$W_2(\rho_s,\rho_t)=|t-s|W_2(\rho,\hat\rho)$ for all $s,t\in[0,T]$. With this, a functional $\mathcal G$ is called
$\lambda$-convex (also called semi-convex) for $\lambda\in\mathbb R$ if the inequality
\begin{equation}
  \label{eqn:lambda_convexity_xx}
  \mathcal G(\rho_t) \le (1-t)\mathcal G(\rho_0) + t \mathcal G(\rho_1) -\frac \lambda 2 t(1-t)W_2^2(\rho_0,\rho_1)
\end{equation}
holds for each constant speed geodesic $(\rho_t)_{t\in[0,1]}$. Note that if two functionals $\mathcal G_i$ are $\lambda_i$-convex
for $i=1,2$, then clearly $\mathcal G_1+\mathcal G_2$ is $\lambda$-convex with $\lambda = \min(\lambda_1,\lambda_2)$.

We call $\mathcal G$ geodesically convex if the map $t\mapsto \mathcal G(\rho_t)$ is convex for any geodesic
$(\rho_t)_{t\in[0,1]}$ (which is equivalent to $\lambda$-convexity for $\lambda=0$). A useful criterion for geodesic convexity of
the free energy $\mathcal F$ is the McCann condition (see Proposition~9.3.9 and equation~(9.3.11) in~\cite{Ambrosio2008b}): A
convex function $f\in C^2([0,\infty);\mathbb R)$ with $f(0)=0$ satisfies the McCann condition (in $d$ dimensions) if the map
\begin{equation}
  \label{eqn:mccann}
  s\mapsto s^d f(s^{-d})
\end{equation}
is convex on $(0,\infty)$ (cf.~the discussion in Section~9.3~in~\cite{Ambrosio2008b}). In the case $d=1$, convexity of $f$ is
sufficient to establish geodesic convexity. For a potential energy of the form $\mathcal V(\rho) = \int_\Lambda V(u)\rho(u)\mathrm
du$ $\lambda$-convexity is equivalent to $\lambda$-convexity (also called strong convexity) of $V$ on $\Lambda$ (see equation
(9.3.3) and Proposition~9.3.2 in~\cite{Ambrosio2008b}), which is $V((1-t)x + ty) \le (1-t)V(x) + t V(y) -(\lambda/2)
t(1-t)\|x-y\|^2$.  For $V\in C^2(\Lambda;\mathbb R)$ the Hessian matrix is bounded and this assumption is trivially
satisfied. Note that under the assumption that $\mathcal F$ is geodesically-convex and $\mathcal V$ is $\lambda$-convex for some
$\lambda\le 0$, also $\mathcal F^V$ is $\lambda$-convex.

\subsubsection{Assumptions for Chain Rule}
\label{sec:assumpt_chainrule_yy}

To our knowledge, minimal sufficient conditions for the validity of a chain rule of the form~\eqref{eqn:chain_rule_formal} are
still an open question. One difficulty is that the existing theory requires $\lambda$-convexity of the functional in question. In
the case of independent particles (with $\chi(a)=\phi(a)=a$) sufficient conditions for $\lambda$-convex functionals can be
obtained from the general theory for gradient flows in Wasserstein spaces, which was established in~\cite{Ambrosio2008b} (see
also~\cite{Villani2003a,Santambrogio2015a}).  We note that generalisations of the gradient flow theory in Wasserstein spaces with
non-linear (usually concave) mobilities have been considered in the literature, see
e.g.~\cite{Lisini2006a,Lisini2009a,Lisini2010a,Dolbeault2009a,Dirr2016a}. Yet, establishing the chain rule in a weighted
Wasserstein metric is fraught with technical difficulties, in particular $\lambda$-convexity of the functional. We overcome this
difficulty here by showing that in the setting studied here, where a weighted Wasserstein metric is the natural space, the chain
rule can be established in an unweighted (classical) Wasserstein setting, where strong tools are available.

In this section, we establish the chain rule~\eqref{eqn:chain_rule_formal} in the special case that {the density $f$ of the free
  energy $\mathcal F^V$ satisfies the McCann condition for geodesic convexity~\eqref{eqn:mccann}} and the particle process is `not
too far away' from the process with independent particles (where $\chi(a)=\phi(a)=a)$: We consider the case $N_{\max}=\infty$ and
assume there exists $C_*>0$ (without loss of generality the same constant which bounds $\phi'(a)$ from below) such that
\begin{equation}
  \label{eqn:lower_bound_chi}
  C_*\le \chi'(a)
\end{equation}
for almost all $a\in(0,\infty)$.  This implies that $C_*\le \chi'(a),\phi'(a)\le C_{\rm Lip}$, such that also $C_*a\le
\chi(a),\phi(a)\le C_{\rm Lip}a$. We obtain for any $\rho\in \mathcal L^1(\Lambda;[0,\infty))$ that the norms $\|\cdot \|_\rho$
  and $\|\cdot \|_{\chi(\rho)}$ are equivalent,
\begin{equation}
  \label{eqn:equiv_of_norms}
  C_*\|\cdot \|_{\rho}\le \|\cdot\|_{\chi(\rho)}\le C_{\rm Lip}\|\cdot\|_\rho.
\end{equation}
In this case also the limit points coincide such that $\mathcal L_{\nabla,\chi(\rho)}^2(\Lambda;\mathbb R^d)=\mathcal
L_{\nabla,\rho}^2(\Lambda;\mathbb R^d)$.  This will allow us to leverage results from the classical Wasserstein framework
in~\cite{Ambrosio2008b}.

\paragraph{Remark.} The Lipschitz continuity of $\chi(a)$ implies that $\mathcal L_{\nabla,\rho}^2(\Lambda;\mathbb
R^d)\subseteq \mathcal L_{\nabla,\chi(\rho)}^2(\Lambda;\mathbb R^d)$.  In general, this is a strict inclusion (consider e.g.~the
case of the SEP with $\chi(a)=a(1-a)$ and $\rho=1$ on a subset $O\subseteq\Lambda$ with positive Lebesgue measure).  A (weaker,
density $\rho$ dependent) condition for the opposite inclusion to hold is
\begin{equation*}
  \inf_{u\in\Lambda}\frac{\chi(\rho(u))}{\rho(u)}>0,
\end{equation*}
which can in this case replace the constant in the lower bound of~\eqref{eqn:equiv_of_norms}. Note that this is a density specific
condition, whereas the above condition~\eqref{eqn:lower_bound_chi} is a model specific condition (which is independent of
$\rho$). For the SEP, this condition is satisfied precisely in the case when $\rho$ is bounded away from the maximal possible
local particle density, i.e.~$\rho\le N_{\max}-\epsilon$ (for some $\epsilon>0$).  The same considerations show that in general
$\mathcal L^2_{\nabla,\rm id}([0,T]\times\Lambda;\mathbb R^d)\subseteq \mathcal L^2_{\nabla,\chi}([0,T]\times\Lambda;\mathbb R^d)$
and that~\eqref{eqn:lower_bound_chi}, or alternatively
\begin{equation*}
  \inf_{(t,u)\in[0,T]\times\Lambda}\frac{\chi(\rho_t(u))}{\rho_t(u)}>0,
\end{equation*}
ensures that $\mathcal L^2_{\nabla,\rm id}([0,T]\times\Lambda;\mathbb R^d) = \mathcal L^2_{\nabla,\chi}([0,T]\times\Lambda;\mathbb R^d)$.

\subsubsection{Validity of the Chain Rule}
\label{sec:Validity-Chain-Rule}

The following results, which are mainly based on Chapter~9 and~10 in~\cite{Ambrosio2008b}, relate $L_f(\rho)$ to the directional
derivative, the Fr\'echet-subdifferential, and the metric slope of $\mathcal F(\rho)$. Below we sketch results which can be
obtained by a suitable modification of the results in~\cite{Ambrosio2008b}. More precisely, we are interested in the case where
the domain is $\Lambda=\mathbb T^d$ and the measures of interest are absolutely continuous with respect to the Lebesgue measure.

As shown in Theorem 1.25 in~\cite{Santambrogio2015a} there exists for any $\rho,\hat\rho\in\mathcal L^1(\Lambda;[0,\infty))$ with
  $\int_\Lambda \rho\;\!\mathrm du=\int_\Lambda \hat\rho\;\!\mathrm du>0$ a unique optimal transport map from $\rho$ to $\hat\rho$
  of the form $r=i-\nabla \varphi$, where $\varphi$ is semi-concave (i.e.~there exists a constant $C>0$ such that $\varphi(u)-C
  |u|^2$ is concave).  Moreover, the interpolation $r_t:=(1-t)i+tr$ between $r$ and the identity $i$ on $\Lambda$ is such that
  $(r_t)_\#\rho$ has a Lebesgue density for all $t\in[0,1]$ (which can e.g. be shown by a modification of the proof of
  Proposition~9.3.9. in~\cite{Ambrosio2008b}).

Now, assume that $f$ satisfies the McCann condition for geodesic convexity~\eqref{eqn:mccann}, that $\mathcal F(\rho),\mathcal
F(\hat\rho)<\infty$, and that $L_f(\rho)\in W^{1,1}(\Lambda;\mathbb R)$. Then
\begin{equation*}
  \int_\Lambda \nabla [ L_f(\rho) ]\cdot(r-i)\;\!\mathrm du
  \le -\int_\Lambda L_f(\rho) \operatorname{tr}\tilde \nabla (r-i)\;\!\mathrm du
  = \lim_{t\searrow 0}\frac{\mathcal F((r_t)_\#\rho)-\mathcal F(\rho)}{t}<\infty,
\end{equation*}
where $\tilde \nabla r$ denotes the approximate derivative (see Definition~5.5.1 in~\cite{Ambrosio2008b}) and $i$ is the identity
on $\Lambda$. This result can be obtained from a modification of the proofs of Lemma~10.4.4 and Lemma~10.4.5
in~\cite{Ambrosio2008b}.

For a $\lambda$-convex functional $\mathcal G$, the Fr\'echet-subdifferential $\partial \mathcal G (\rho)$ at $\rho\in\mathcal
L^1(\Lambda;[0,\infty))$ with $\int_\Lambda \rho\;\!\mathrm du>0$ consists of all vectors $\zeta\in \mathcal
  L^2_\rho(\Lambda;\mathbb R^d):=\{\zeta\colon\Lambda\to\mathbb R^d: \|\zeta\|_\rho<\infty\}$ such that for all
  $\hat\rho\in\mathcal L^1(\Lambda;[0,\infty))$ with $\int_\Lambda \rho\;\!\mathrm du=\int_\Lambda \hat\rho\;\!\mathrm du$
\begin{equation}
  \label{eqn:frechet_subdiff}
  \mathcal G(\hat\rho)-\mathcal G(\rho) \ge \int_\Lambda \zeta\cdot (r-i) \rho\;\!\mathrm du + \frac \lambda 2 W_2^2(\rho,\hat\rho),
\end{equation}
where $r$ is the optimal transport map from $\rho$ to $\hat\rho$ (see Equation~(10.1.7) in~\cite{Ambrosio2008b}).

\begin{lemma}[Slope and subdifferential, cf.~Theorem 10.4.6 in~\cite{Ambrosio2008b}]
  \label{lem:slope_diff}
  Assume that $f$ satisfies the McCann condition for geodesic convexity~\eqref{eqn:mccann}. For $\rho\in\mathcal
  L^1(\Lambda;[0,\infty))$ with $\int_\Lambda \rho\;\!\mathrm du>0$ and $\mathcal F(\rho)<\infty$ the following statements are
    equivalent.
  \begin{enumerate}
  \item \label{it:slope1} The Fr\'echet-subdifferential~\eqref{eqn:frechet_subdiff} is non-empty,
    $\partial \mathcal F^V(\rho)\not=\emptyset$.
  \item \label{it:slope2} The metric derivative at $\rho$ is finite,
    \begin{equation*}|\partial \mathcal F^V|(\rho):=\limsup_{W_2(\rho,\hat\rho)\to 0} \frac{(\mathcal F^V(\rho)
        -\mathcal F^V(\hat\rho))^+}{W_2(\rho,\hat\rho)} <\infty.
    \end{equation*}
  \item \label{it:slope3} $L_f(\rho)\in W^{1,1}_{\rm loc}(\Lambda;\mathbb R)$ with $\nabla [L_f(\rho)] + \rho\nabla V=\rho w$ for
    some $w\in \mathcal L_{\nabla,\rho}^2(\Lambda;\mathbb R^d)$.
    \newcounter{enumiold} %% Only the first time in a document 
    \setcounter{enumiold}{\value{enumi}}
  \end{enumerate}
  If either of the above holds we have $w\in\partial\mathcal F(\rho)$ and $\|w\|_\rho = |\partial \mathcal F|(\rho)$. Moreover, if
  the additional assumption~\eqref{eqn:lower_bound_chi} holds, then the above conditions are also equivalent to
  \begin{enumerate}
    \addtocounter{enumi}{\value{enumiold}}
  \item \label{it:slope4} $\phi(\rho)\in W^{1,1}_{\rm loc}(\Lambda;\mathbb R)$ with $\nabla [\phi(\rho)] + \chi(\rho)\nabla
    V=\chi(\rho)w$ for some $w\in \mathcal L_{\nabla,\chi(\rho)}^2(\Lambda;\mathbb R^d)$.
  \end{enumerate}
\end{lemma}

\begin{proof}
The equivalence between~\ref{it:slope1} and~\ref{it:slope2} holds since (by Lemma 10.1.5 in~\cite{Ambrosio2008b}) the metric slope
for (regular and thus in particular) $\lambda$-convex functionals is given by
\begin{equation}
  \label{eqn:representation_metric_slope}
  |\partial \mathcal F|(\rho) = \min\{\|\zeta\|_\rho \;\! : \;\! \zeta\in\partial \mathcal F(\rho)\}.
\end{equation}
We next show that~\ref{it:slope2} implies~\ref{it:slope3}. The result follows from a standard calculation, cf.~e.g.~the proof of
Lemma 3.5 in~\cite{Lisini2009a}. Consider a smooth function $\xi\in C^\infty_c(\Lambda;\mathbb R)$. We define the flow associated
to $\nabla\xi$ as the unique solution $X(t,u)$ to $\dot X(t,u) = \nabla\xi(X(t,u)),\quad X(0,u)=u$ for $u\in\Lambda$ and
$t\in(0,1)$. For $\rho_t^\xi:= X(t,\cdot)_\#\rho$ we have (cf.~(3.32) in~\cite{Lisini2009a})
\begin{equation}
  \label{eqn:W_2_bound_xi}
  W_2^2(\rho,\rho_t^\xi)\le t\int_0^t \|\nabla\xi\|_{\rho_s^\xi}^2\;\!\mathrm ds
  =t^2 (\|\nabla\xi\|_\rho^2 + o(1)).
\end{equation}
Similar to (3.35) and (3.36) in~\cite{Lisini2009a} one finds
\begin{equation}
  \label{eqn:grad_F_V}
  \lim_{t \to 0} \frac{\mathcal F(\rho_t^\xi)-\mathcal F(\rho)}t
  = \int_\Lambda \nabla [ L_f(\rho) ]\cdot \nabla\xi\;\!\mathrm du\quad\textrm{ and }\quad
  \lim_{t \to 0} \frac{\mathcal V(\rho_t^\xi)-\mathcal V(\rho)}t = \int_\Lambda \rho\nabla V\cdot \nabla\xi \;\!\mathrm du.
\end{equation}
Using~\eqref{eqn:W_2_bound_xi} and $\mathcal F^V= \mathcal F + \mathcal V$ we obtain (cf.~(3.33) in~\cite{Lisini2009a})
\begin{equation*}
  |\partial \mathcal F^V|(\rho)\ge \frac{1}{\|\nabla\xi\|_\rho}\lim_{t \to 0} \frac{\mathcal F^V(\rho_t^\xi)-\mathcal F^V(\rho)}t
  = \frac{1}{\|\nabla\xi\|_\rho}\int_\Lambda (\nabla [ L_f(\rho) ] +\rho \nabla V)\cdot \nabla\xi\;\!\mathrm du.
\end{equation*}
Similar to the discussion at the beginning of Section~\ref{sec:reg_paths_hydro_yy}, $|\partial \mathcal F^V|(\rho)<\infty$ implies
that the linear operator $v\mapsto \int_\Lambda (\nabla [L_f(\rho)]+\rho\nabla V)\cdot v\;\!\mathrm du$ from
$\mathcal L^2_{\nabla,\rho}(\Lambda;\mathbb R^d)$ to $\mathbb R$ is bounded, such that Riesz' representation theorem implies the
existence of $w\in \mathcal L^2_{\nabla,\rho}(\Lambda;\mathbb R^d)$ for which $\nabla[L_f(\rho)] +\rho \nabla V = \rho w$, such
that $L_f(\rho)\in W^{1,1}_{\rm loc}(\Lambda;\mathbb R)$. In particular $|\partial \mathcal F^V|(\rho)\ge \|w\|_\rho$.

For the implication~\ref{it:slope3} to~\ref{it:slope2} consider any $\hat\rho\in\mathcal L^1(\Lambda;[0,\infty))$ with
  $\int_\Lambda \rho\;\!\mathrm du=\int_\Lambda \hat\rho\;\!\mathrm du$ and $\mathcal F(\hat\rho)<\infty$. Then
\begin{equation*}
  \mathcal F(\hat\rho)-\mathcal F(\rho) \ge \lim_{t\to 0}\frac{\mathcal F((r_t)_\#\rho)-\mathcal F(\rho)}t 
  \ge\int_\Lambda \nabla[L_F(\rho)]\cdot (r-i)\;\!\mathrm du,
\end{equation*}
where the fist inequality follows from the monotonicity of the difference quotient (see Equation~(10.4.24)
in~\cite{Ambrosio2008b}). The $\lambda$-convexity of $\mathcal V$ yields (cf.~\eqref{eqn:lambda_convexity_xx})
\begin{equation*}
  \mathcal V(\hat\rho)-\mathcal V(\rho) \ge \lim_{t\to 0}\frac{\mathcal V((r_t)_\#\rho)-\mathcal V(\rho)}t
  + \frac \lambda 2 W_2^2(\rho,\hat\rho)
  =\int_\Lambda \rho\nabla V\cdot (r-i)\;\!\mathrm du + \frac \lambda 2 W_2^2(\rho,\hat\rho).
\end{equation*}
This implies that $w=(\nabla[L_F(\rho)]/\rho + \nabla V)\in \partial\mathcal F^V(\rho)$ and thus $|\partial \mathcal F^V|(\rho)
\le \|w\|_\rho<\infty$ by eqn.~\eqref{eqn:representation_metric_slope}.

The equivalence between~\ref{it:slope3} and~\ref{it:slope4} can be seen as follows: Recall that $C_*L_f'(a)\le \phi'(a)\le C_{\rm
  Lip} L_f'(a)$ and also $C_* L_f(a) \le \phi(a) \le C_{\rm Lip} L_f(a)$. With the same argument as in the proof of
Proposition~\ref{prop:e_star_rep} we obtain that the chain rule holds as in~\eqref{eqn:sobolev_chain_rule_0},
i.e.~$L_f'(\rho)\nabla\rho = \nabla[L_f(\rho)]$ and $\phi'(\rho)\nabla\rho = \nabla[\phi(\rho)]$, such that
$C_*\|\nabla[L_f(\rho)]\|\le \|\nabla[\phi(\rho)]\|\le C_{\rm Lip} \|\nabla[L_f(\rho)]\|$. This proves that $\phi(\rho)\in
W^{1,1}(\Lambda;\mathbb R)$ if and only if $L_f(\rho)\in W^{1,1}(\Lambda;\mathbb R)$.  Moreover $w=\nabla [L_f(\rho)]/\rho =
\nabla [\phi(\rho)]/\chi(\rho)$.
\end{proof}

Finally, we can outline a proof for Theorem~\ref{thm:chain-rule}, which follows ideas from~\cite{Ambrosio2008b,Lisini2009a}. Since
we work on the torus $\Lambda=\mathbb T^d$ (rather than $\mathbb R^d$), we sketch the argument.

\begin{sketchofproofof}{Theorem~\ref{thm:chain-rule}} 
Since $\mathbb{A}$ is finite and the assumptions of Section~\ref{sec:new-assump} are valid Propositions~\ref{prop:e_rep}
and~\ref{lem:continuity_lemma} and~\ref{prop:e_star_rep} hold. Moreover, since $f$ satisfies the McCann
condition~\eqref{eqn:mccann} and also the assumption~\eqref{eqn:lower_bound_chi} on $\chi'$ holds we can apply
Lemma~\ref{lem:slope_diff}. Combining all these results we have that the map $t\mapsto |\rho'_t||\partial \mathcal F^V|(\rho_t)$
is in $\mathcal L^1_{\rm loc}([0,T];\mathbb R)$. This then implies that $t\mapsto \mathcal F^V(\rho_t)$ is locally absolutely
continuous (see e.g.~Lemma 3.4 in~\cite{Lisini2009a}), with a.e.~derivative
\begin{equation*}
  \frac{d}{dt} \mathcal F^V(\rho_t) 
  =-\langle v_t,w_t\rangle_{\chi(\rho_t)}
  =-\langle \dot\rho_t, \Delta(\rho_t)+\nabla\cdot(\chi(\rho_t)\nabla V)\rangle_{-1,\chi(\rho_t)},
\end{equation*}
which implies the chain rule~\eqref{eqn:chain_rule_formal}.
\end{sketchofproofof}

\section{Proofs and Supplementary Content}
\label{sec:proof_bounds}

For nearest neighbour transitions, the following proposition yields a special representation for symmetric summands.
\begin{proposition}
  \label{prop:rep}
  Let $A_{\eta,\eta'}$ be a symmetric function (such that $A_{\eta,\eta'}=A_{\eta',\eta}$) with $A_{\eta,\eta}=0$ and
  $A_{\eta,\eta^{i,j}}=0$ whenever $|i-j|\not =1$. If either $\sum_{\eta,\eta'\in\Omega_L}|A_{\eta,\eta'}|<\infty$ or
  $A_{\eta,\eta'}\ge 0$ for all $\eta,\eta'\in\Omega_L$, then
  \begin{equation}
    \label{eqn:sym_A}
    \sum_{\eta,\eta'\in\Omega_L}A_{\eta,\eta'}=2 \sum_{i\in\mathbb T_L^d}\sum_{k=1}^d\sum_{\eta\in\Omega_L} A_{\eta,\eta^{i,i+e_k}}
    \mathbf 1_{\{\eta(i)>0\}}.
  \end{equation}
\end{proposition}

\begin{proof}
Note that by definition $\sum_{\eta,\eta'\in\Omega_L}A_{\eta,\eta'}= \sum_{i\in\mathbb
  T_L^d}\sum_{k=1}^d\sum_{\eta\in\Omega_L}\bigl( A_{\eta,\eta^{i,i+e_k}} + A_{\eta,\eta^{i,i-e_k}}\bigr)\mathbf
1_{\{\eta(i)>0\}}$.  Using symmetry, the second summand is equal to $A_{\eta^{i,i-e_k},\eta}$, such that first replacing the
configuration $\eta$ with $\eta^{i-e_k,i}$ before replacing the index $i$ with $i+e_k$ yields~\eqref{eqn:sym_A}.
\end{proof}

Following~\cite{Kipnis1999a} Chapter 5, we define for $\epsilon>0$ the approximation of the identity
$\iota_\epsilon:=(2\epsilon)^{-d}\mathbf 1_{[-\epsilon,\epsilon)^d}(\cdot)$. Recall that the convolution of a measure
  $\pi\in\mathcal M_+(\Lambda)$ with a function $f\in \mathcal L^1(\Lambda;\mathbb R)$ is defined as $[\pi*f](u):= \int_\Lambda
  f(u'\!-\!u)\;\! \pi(\mathrm du')$. The convolution of $\iota_\epsilon$ with the empirical measure~\eqref{eqn:empirical_measure}
  is the function
\begin{equation}
  \label{eqn:conv_empirical_const}
        [\Theta_L(\eta)*\iota_\epsilon](u) = (2\epsilon L)^{-d}  \sum_{i\in\mathbb T_L^d} \mathbf 1_{[\frac{2i-1}{2L},\frac{2i+1}{2L})^d}(u)
          \sum_{j:|i-j|\le \lfloor\epsilon L\rfloor} \eta(j),
\end{equation}
which is piecewise constant on $\{[\frac{2i-1}{2L},\frac{2i+1}{2L})^d\}_{i\in\mathbb T_L^d}$. This allows us to represent the
  averaged particle density as a function of the empirical distribution, i.e.
\begin{equation*}
  [\Theta_L(\eta)*\iota_\epsilon](i/L) = \Bigl(\frac{2{\lfloor \epsilon L\rfloor}+1}{2\epsilon L}\Bigr)^d \eta^{\lfloor \epsilon L\rfloor}(i).
\end{equation*}
For $\pi(\mathrm du)=\rho(u)\;\!\mathrm du$ the convolution yields $[\pi*\iota_\epsilon](u) =
(2\epsilon)^{-d}\int_{[u-\epsilon,u+\epsilon)^d}\rho(u')\;\!\mathrm du'$.  Since $\lim_{\epsilon\to
    0}[\pi*\iota_\epsilon](u)=\rho(u)$ for almost all $u\in\Lambda$, we define $[\pi*\iota_0](u):= \rho(u)$.

\subsection{Proofs of the Statements in Section~\ref{sec:act-fcts}}
\label{sec:proof_xx_new}

\begin{proofof}{Proposition~\ref{prop:micro_chain_rule}}  Recall that $(\mu^L_t)_{t\in[0,T]}$ is finitely supported in the sense that
the set $\mathcal N_0:= \{ \eta\in\Omega_L| \mu^L_t(\eta)>0$ for some $t\in[0,T]\}$ is finite. Since $r^L_t$ consists of nearest
neighbour transitions, also the set $\mathcal N_1:=\{(\eta,\eta')\in\Omega_L\times\Omega_L | \mu^L_t(\eta)(r^L_t)_{\eta,\eta'}>0$
or $\mu^L_t(\eta')(r^L_t)_{\eta',\eta}>0$ for some $t\in[0,T]\}$ is finite.  Thus the left hand side
of~\eqref{eqn:free_energy_diff} is equal to
\begin{multline*}
  \sum_{\eta\in\mathcal N_0}\Biggl[
    \mu^L_{t_2}(\eta)\log\biggl(\frac{\mu^L_{t_2}(\eta)}{\nu_\alpha(\eta)}\biggr)
    -\mu^L_{t_1}(\eta)\log\biggl(\frac{\mu^L_{t_1}(\eta)}{\nu_\alpha(\eta)}\biggr)
    \Biggr]\\
  +\sum_{\eta\in\mathcal N_0}\sum_{i\in\mathbb T_L^d} \Bigl(\mu^L_{t_2}(\eta)  \eta(i) \tilde V_{t_2}(\tfrac iL)
  -\mu^L_{t_1}(\eta) \eta(i) \tilde V_{t_1}(\tfrac iL)\Bigr)\\
  +\log \biggl(\sum_{\eta\in\Omega_L} \nu_\alpha(\eta)\mathrm e^{-\sum_{i\in\mathbb T_L^d}\tilde V_{t_2}(i/L)\eta(i)}\biggr)
  -\log \biggl(\sum_{\eta\in\Omega_L} \nu_\alpha(\eta)\mathrm e^{-\sum_{i\in\mathbb T_L^d}\tilde V_{t_1}(i/L)\eta(i)}\biggr).
\end{multline*}
Similar to Theorem~9.2 of Appendix~1 in~\cite{Kipnis1999a}, one then shows using~\eqref{eqn:finite_moments} that the latter is
equal to
\begin{multline*}
  \qquad\sum_{\eta\in\mathcal N_0}\int_{t_1}^{t_2}\frac {\mathrm d}{\mathrm dt}
  \biggl[ \mu^L_t(\eta)\log\biggl(\frac{\mu^L_t(\eta)}{\nu_\alpha(\eta)}\biggr)\biggr]\mathrm dt
  +\sum_{\eta\in\mathcal N_0}\sum_{i\in\mathbb T_L^d} \int_{t_1}^{t_2}\frac{\mathrm d}{\mathrm dt}\Bigl[\mu^L_t(\eta)  \eta(i) \tilde V_t(\tfrac iL)
    \Bigr]\mathrm dt\\
  -\int_{t_1}^{t_2} \sum_{\eta\in\Omega_L}\nu_\alpha^{\tilde V_t}(\eta) \sum_{i\in\mathbb T_L^d}\eta(i)\;\!\partial_t\tilde V_t(\tfrac iL)
  \;\!\mathrm dt.\qquad
\end{multline*}
A straightforward calculation (using $\partial_t\mu^L_t(\eta)=-\operatorname{div}\jmath^L_t(\eta)$, the fact that the transition
rates $r^L_t$ are bounded, and the fact that $\mu^L_t$ is supported on a finite number of configurations) allows to show that
\begin{multline}
  \label{eqn:free_energy_diff_chainrule_mc}
  \mathcal F_{L,\alpha}^{\tilde V_{t_2}}(\mu^L_{t_2}) - \mathcal F_{L,\alpha}^{\tilde V_{t_1}}(\mu^L_{t_1})
  =-\sum_{\eta\in\mathcal N_0}\int_{t_1}^{t_2}\operatorname{div}\jmath^L_t(\eta) \biggl(\log\biggl(\frac{\mu^L_t(\eta)}{\nu_\alpha(\eta)}\biggr)
  +1\biggr)\mathrm dt\\
  -\sum_{\eta\in\mathcal N_0}\int_{t_1}^{t_2}\operatorname{div}\jmath^L_t(\eta)
  \sum_{i\in\mathbb T_L^d}\eta(i)\;\!\tilde V_t(\tfrac iL)\;\!\mathrm dt
  + \int_{t_1}^{t_2} \sum_{\eta\in\Omega_L}\bigl(\mu^L_t(\eta)-\nu_\alpha^{\tilde V_t}(\eta)\bigr)
  \sum_{i\in\mathbb T_L^d}\eta(i)\;\!\partial_t\tilde V_t(\tfrac iL)\;\!\mathrm dt.
\end{multline}
Using once more the boundedness of the nearest neighbour transition rates and that $\mu_0$ is supported on finitely many
configurations, we can show, employing the bound $\log(\mu^L_t(\eta)/\nu_\alpha(\eta))\le |\log(\nu_\alpha(\eta))|$, that
\begin{multline*}
  \int_0^T\!\!\sum_{\eta,\eta'\in\Omega_L} \Bigl|(\jmath^L_t)_{\eta,\eta'}
  \log\biggl(\frac{\mu^L_t(\eta)}{\nu_\alpha(\eta)}\biggr)\Bigr|\mathrm dt\\
  \le
  \int_0^T\!\!\!\sum_{(\eta,\eta')\in\mathcal N_1} \bigl(\mu^L_t(\eta)(r^L_t)_{\eta,\eta'}+\mu^L_t(\eta')(r^L_t)_{\eta',\eta}\bigr)
  |\log(\nu_\alpha(\eta))|\mathrm dt<\infty.
\end{multline*}
The latter allows us to combine the first two summands on the right hand side of~\eqref{eqn:free_energy_diff_chainrule_mc}, which
are equal to $-\sum_{\eta\in\Omega_L}\operatorname{div}\jmath^L_t(\eta) \log({\mu^L_t(\eta)}/{\nu_\alpha^{\tilde V_t}(\eta)})
=-\langle \jmath^L_t,F^{\tilde V_t}(\mu^L_t)\rangle_L$, where the last identity follows by a summation by parts (cf.~Equation~(15)
in~\cite{Kaiser2018a}). This finishes the proof.
\end{proofof}

The proof of Theorem~\ref{thm:repl_lemma_thm} relies on an auxiliary statement of independent interest, which we prove first. The
result gives sufficient conditions for local equilibration.
\begin{lemma}
  \label{lem:suff_loc_eqm}
  Consider $(P_L)_{L\in\mathbb N}$ from Section~\ref{sec:path_meas} with associated density $(\mu^L_t)_{t\in[0,T]}$. Assume there
  exists $\tilde V\in C^{1,2}([0,T]\times\Lambda;\mathbb R)$ such that the inequalities
  \begin{equation}
    \label{eqn:suff_replacement_lemma_1}
    \limsup_{L\to\infty}\frac 1{L^{d}}\int_0^T\mathcal F_{L,\alpha}^{\tilde V_t}(\mu^L_t)\;\!\mathrm dt<\infty
  \end{equation}
  and
  \begin{equation}
    \label{eqn:suff_replacement_lemma_2}
    \limsup_{L\to\infty}\frac 1{L^d}\int_0^T \Psi_L^\star\bigl(\mu^L_t,F^{\tilde V_t}(\mu^L_t)\bigr) \;\!\mathrm dt<\infty
  \end{equation}
  are satisfied. Then $(\mu^L_{[0,T]})_{L\in\mathbb N}$ (where again $\mu^L_{[0,T]}:=\frac 1T \int_0^T \mu^L_t\;\!\mathrm dt$) is
  in the class considered by the replacement lemma~\eqref{eqn:replacement_lemma}. In particular~\eqref{eqn:local_equilibrium}
  and~\eqref{eqn:local_equilibrium_2} are satisfied for $(\mu^L_{[0,T]})_{L\in\mathbb N}$. Moreover, these assumptions are
  independent of the choices of $\tilde V$ and $\alpha$: We can replace $\tilde V$ with $\tilde V+\tilde H$ for some $\tilde H\in
  C^{1,2}([0,T]\times\Lambda;\mathbb R)$ and also replace $\alpha$ with $\alpha'\in(0,N_{\max})$
  in~\eqref{eqn:suff_replacement_lemma_1} arbitrary. Then~\eqref{eqn:suff_replacement_lemma_1}
  and~\eqref{eqn:suff_replacement_lemma_2} are satisfied for $\tilde V$ and $\alpha$ if and only if they are satisfied for $\tilde
  V + \tilde H$ and $\alpha'$.
\end{lemma}

\begin{proof}
The bound~\eqref{eqn:suff_replacement_lemma_1} for $\tilde V+\tilde H$ and $\alpha'$ follows similar to Remark~1.2 on page~70
of~\cite{Kipnis1999a}. For~\eqref{eqn:suff_replacement_lemma_2} note that the basic estimate $\cosh(x+y)\le \cosh(x)\mathrm
e^{|y|}$ combined with~\eqref{eqn:C_hat} yields
\begin{equation}
  \label{eqn:psi_star_estim}
  \frac 1{L^d}\int_0^T \Psi_L^\star\bigl(\mu^L_t,F_\alpha^{\tilde V_t+\tilde H_t}(\mu^L_t)\bigr) \;\!\mathrm dt 
  \le  \frac {C_{\tilde H}}{L^d}\int_0^T \Psi_L^\star\bigl(\mu^L_t,F^{\tilde V_t}(\mu^L_t)\bigr) \;\!\mathrm dt + 2(C_{\tilde H}-1)T C_{\hat\chi}
\end{equation}
for some $C_{\tilde H}>0$ that only depends on $H$. We thus can restrict to the special case $\tilde V_t = 0$. The two bounds
needed for the replacement lemma~\eqref{eqn:replacement_lemma} then follow from convexity, i.e.~$\mathcal
F_{L,\alpha}^0(\mu^L_{[0,T]})\le \frac 1T\int_0^T\mathcal F_{L,\alpha}^0(\mu^L_t)\;\!\mathrm dt$ and
$\Psi_L^\star\bigl(\mu^L_{[0,T]},F^0(\mu^L_{[0,T]})\bigr) \le \frac 1T\int_0^T \Psi_L^\star\bigl(\mu^L_t,F^0(\mu^L_t)\bigr)
\;\!\mathrm dt$ (cf.~the discussion in Chapter 5.3 near equation~(3.1) on page 81 in~\cite{Kipnis1999a}).
\end{proof}

With this result at hand, we can turn to the proof of Theorem~\ref{thm:repl_lemma_thm}.

\begin{proofof}{Theorem~\ref{thm:repl_lemma_thm}} Since the relative entropy is non-negative, we obtain with a modification
of~\eqref{eqn:rel_ent_2b} to the time interval $[t,T]$ (for each $t\in[0,T]$) that
\begin{equation}
  \label{eqn:free_energy_estimate_bw}
  \begin{split}
    \mathcal F_{L,\alpha}^{\tilde V_t}(\mu^L_t)
    &\le 
    \mathcal F_{L,\alpha}^{\tilde V_T}(\mu^L_T) 
    +\int_t^T \Psi_L(\mu^L_s,\jmath^L_s)\;\!\mathrm ds+ \int_t^T
    \Psi^\star_L\bigl(\mu^L_s,F_\alpha^{\tilde V_s}(\mu^L_s)\bigr)\;\!\mathrm ds 
    \\
    &\qquad\qquad\qquad\qquad
    -\int_t^T\sum_{i\in\mathbb T_L^d} \bigl(\hat\rho_i(\mu^L_s)
    -\bar\rho_{\alpha,\tilde V_s}(i)\bigr)\;\!\partial_s\tilde V_s(\tfrac iL)\;\!\mathrm ds
    \\
    &\le  \mathbb A_L^{\tilde V}\bigl(Q_L\bigr) + \mathcal F_{L,\alpha}^{\tilde V_0}(\mu^L_0)
    +  C_{\tilde V}\biggl( TL^dC_{\rm tot} + \int_0^T\sum_{i\in\mathbb T_L^d} \bar\rho_{\alpha,\tilde V_t}(i/L)\;\!\mathrm dt \biggr) ,
  \end{split}
\end{equation}
where $C_{\tilde V}$ is a constant that only depends on $\tilde V$. Thus
\begin{multline}
  \label{eqn:limsup_thm_31}
  \qquad\limsup_{L\to\infty}\frac 1{L^d} \int_0^T \mathcal F_{L,\alpha}^{\tilde V_t}(\mu^L_t)\;\!\mathrm dt
  \le  \limsup_{L\to\infty}\frac T{L^d}\mathbb A_L^{\tilde V}\bigl(Q_L\bigr)
  +\limsup_{L\to\infty}\frac T{L^d}\mathcal F_{L,\alpha}^{\tilde V_0}(\mu^L_0)\\
  + T^2 C_{\tilde V} C_{\rm tot}
  + TC_{\tilde V} \int_0^T\int_\Lambda \bar\rho_{\alpha,\tilde V_t}(u)\;\!\mathrm du\;\!\mathrm dt
  <\infty.\qquad
\end{multline}
The second inequality follows from a similar estimate to~\eqref{eqn:free_energy_estimate_bw}: Consider the second inequality
in~\eqref{eqn:free_energy_estimate_bw} for $t=0$ and drop the term $\mathcal F_{L,\alpha}^V(\mu^L_T) + \int_0^T
\Psi_L(\mu^L_t,\jmath^L_t)\;\!\mathrm dt \ge 0$. Then
\begin{multline*}
  \int_0^T \Psi^\star_L\bigl(\mu^L_t,F^{\tilde V_t}(\mu^L_t)\bigr)\;\!\mathrm dt\\
  \le  \mathbb A_L^{\tilde V}\bigl(Q_L\bigr) + \mathcal F_{L,\alpha}^{\tilde V_0}(\mu^L_0)
  +  2 C_{\tilde V}\biggl( TL^dC_{\rm tot} + \int_0^T\sum_{i\in\mathbb T_L^d} \bar\rho_{\alpha,\tilde V_t}(i/L)\;\!\mathrm dt \biggr)
\end{multline*}
and we can conclude as in~\eqref{eqn:limsup_thm_31}. We then apply Lemma~\ref{lem:suff_loc_eqm} to obtain that the
equations~\eqref{eqn:local_equilibrium} and~\eqref{eqn:local_equilibrium_2} are satisfied for $(\mu^L_{[0,T]})_{L\in\mathbb
  N}$. The independence of $V$, $\tilde V$ and $\alpha$ follows from the considerations in Lemma~\ref{lem:suff_loc_eqm}.
\end{proofof}

\subsection{Proofs of Liminf Inequalities}
\label{sec:assumptions_2}

This section is devoted to the proof of the liminf inequalities in the proof of Theorem~\ref{thm:lower-bound-thm}. Many of the
ideas of the following proofs are borrowed from the entropy method developed in~\cite{Guo1988a}. We here follow the presentation
of this method in Chapter 5 of the book by Kipnis and Landim~\cite{Kipnis1999a}.  The results we want to prove are of the form
$\liminf_{L\to\infty} B_L \ge B_*$.  The general strategy involves replacing $B_L$ by some (possibly $\epsilon$ dependent)
$C_L^\epsilon$ and to show that
\begin{equation*}
  \liminf_{\epsilon\to 0}\liminf_{L\to\infty} C_L^\epsilon \ge B_* \qquad\textrm{and}\qquad \limsup_{\epsilon\to 0}
  \limsup_{L\to\infty} |B_L - C_L^\epsilon| = 0.
\end{equation*}

\subsubsection{Bounds for $\Psi_L$ and $\Psi^\star_L$}
\label{sec:Bounds-PsiL-PsistarL*}
In order to achieve the projection to the physical domain anticipated in Section~\ref{sec:ips} we consider functions which are
linear in $\eta$. For this we fix a function $G\in C^1(\Lambda;\mathbb R)$ and define $\tilde G_L\colon\Omega_L\to\mathbb R$ by
$\tilde G_L(\eta) := L^d\langle\Theta_L(\eta),G\rangle = \sum_{i\in\mathbb T_L^d}G(i/L)\eta(i)$, for which the discrete derivative
satisfies the identity $\nabla^{\eta,\eta^{i,i+e_k}}\tilde G_L=\nabla^{i,i+e_k}G(\cdot/L)$. Note that this last identity allows us
to reduce the dependence on the configuration space to a dependence on the physical domain. Choosing the `force' $F=\nabla \tilde
G_L$, we obtain with Proposition~\ref{prop:rep} (since all summands are non-negative) that
\begin{equation}
  \label{eqn:psi_star_phys}
  \Psi^\star_L(\mu, \nabla \tilde G_L)
  = 2\sum_{i\in\mathbb T_L^d}\sum_{k=1}^d \hat a_{i,i+e_k}(\mu) L^2\Bigl[\cosh\bigl(\tfrac 12 \nabla^{i,i+e_k}G(\cdot/L)\bigr)-1\Bigr]
\end{equation}
and similar, for the current $j_{\eta,\eta'}^G=a_{\eta,\eta'}(\mu) \sinh\bigl(\tfrac 12 \nabla^{\eta,\eta'}\tilde G_L\bigr)$
associated to the above force (cf.~\cite{Kaiser2018a})
\begin{multline}
  \label{eqn:psi_phys}
  \qquad\Psi_L(\mu,j^G) 
  =2\sum_{i\in\mathbb T_L^d}\sum_{k=1}^d \hat a_{i,i+e_k}(\mu) L^2\Bigl[\sinh\bigl(\tfrac 12 \nabla^{i,i+e_k}G(\cdot/L)\bigr)\tfrac 12
    \nabla^{i,i+e_k}G(\cdot/L)\\-\Bigl(\cosh\bigl(\tfrac 12 \nabla^{i,i+e_k}G(\cdot/L)\bigr)-1\Bigr)
    \Bigr].\qquad
\end{multline}
We next derive upper bounds for~\eqref{eqn:psi_star_phys} and~\eqref{eqn:psi_phys} and a lower bound for $\Psi_L^\star(\mu,
F^V(\mu))$.

\begin{proposition}[Upper bounds for $\Psi_L$ and $\Psi^\star_L$] 
  \label{prop:upper_bound_finite}
  Let $\mu$ be a measure on $\Omega_L$. Further let $f_{\eta,\eta'}:=\nabla^{\eta,\eta'} \tilde G_L$ for some
  $G\colon\Lambda\to\mathbb R$ and $j^G_{\eta,\eta'}:=a_{\eta,\eta'}(\mu) \sinh\bigl(\tfrac 12 \nabla^{\eta,\eta'}\tilde
  G_L\bigr)$. Then
  \begin{equation}
    \label{eqn:bounds_psi_psi_star}
    \Psi^\star_L(\mu, \nabla \tilde G_L) \le \Psi_L(\mu,j^G) \le \frac 12\sum_{i\in\mathbb T_L^d}\sum_{k=1}^d
    \hat \chi^0_{i,i+e_k}(\mu) \bigl[2L \sinh\bigl(\tfrac 12 \nabla^{i,i+e_k}G(\cdot/L)\bigr)\bigr]^2.\qquad
  \end{equation}
\end{proposition}

\begin{proof}
The proof follows from the basic inequalities $\cosh(x)-1 \le x\sinh(x)-(\cosh(x)-1) \le \frac 12\sinh(x)^2$ applied
to~\eqref{eqn:psi_star_phys} and~\eqref{eqn:psi_phys}, together with the inequality $\hat a_{i,i+e_k}(\mu)\le 2\hat
\chi^0_{i,i+e_k}(\mu)$ stated below~\eqref{eqn:hat_chi}.
\end{proof}

\begin{proposition}[Lower bound for $\Psi^\star_L$]
  \label{prop:lower_bound_finite}
  Let $\mu$ be a measure on $\Omega_L$, $\alpha\in(0, N_{\max})$ and $V\in C^2(\Lambda;\mathbb R)$. Then, for any
  $G\colon\Lambda\to\mathbb R$ we have the following lower bound on $\Psi^\star_L\bigl(\mu,F^V(\mu)\bigr)$ uniform in $\alpha$
  \begin{multline}
    \label{eqn:Psi_star_Fs_lower_bound}
    \Psi^\star_L\bigl(\mu,F^V(\mu)\bigr) \\
    \ge \sum_{i\in\mathbb T_L^d}\sum_{k=1}^d \Bigl[\bigl(L\hat\jmath^V_{i,i+e_k}(\mu)\bigr)(
      L\nabla^{i,i+e_k} G(\cdot / L)) -\frac 12 \hat\chi^V_{i,i+e_k}(\mu) \bigl[L\nabla^{i,i+e_k}G(\cdot / L)\bigr]^2\Bigr].
  \end{multline}
\end{proposition}
\begin{proof}
We use the notation $\rho:=\mu/\nu_\alpha^V$ (s.t.~$\rho$ is the density of $\mu$ with respect to $\nu_\alpha^V$) and
$q_{\eta,\eta'}:=\nu_\alpha^V(\eta)r^V_{\eta,\eta'}$, such that the relation $q_{\eta,\eta'}=q_{\eta',\eta}$ (detailed balance)
holds. Then $F^V_{\eta,\eta'}(\mu)=-\nabla^{\eta,\eta'}\log \rho$ and
$a_{\eta,\eta'}(\mu)=2\sqrt{\rho(\eta)q_{\eta,\eta'}\rho(\eta')q_{\eta',\eta}}$. Further, $a_{\eta,\eta'}(\mu)[\cosh(\tfrac 12
  F^V_{\eta,\eta'}(\mu))-1] =\sqrt{q_{\eta,\eta'}q_{\eta',\eta}}(\sqrt{\rho(\eta)}-\sqrt{\rho(\eta')}\;\!)^2$.  Using the
representation in Proposition~\ref{prop:rep} and $q_{\eta,\eta'}=\sqrt{q_{\eta,\eta'}q_{\eta',\eta}}=q_{\eta',\eta}$, we obtain
\begin{equation*}
  \Psi^\star_L\bigl(\mu,F^V(\mu)\bigr)
  = \sum_{\eta\in\Omega_L}\sum_{i\in\mathbb T_L^d}\sum_{k=1}^d  2q_{\eta,\eta^{i,i+e_k}}\;\!\Bigl(\sqrt{\rho(\eta)}
  -\sqrt{\rho(\eta^{i,i+e_k})}\;\!\Bigr)^2.
\end{equation*}
Define $H_{\eta,\eta'}=\frac 14\bigl(\sqrt{\rho(\eta)}+\sqrt{\rho(\eta')}\bigr)\nabla^{\eta,\eta'}\tilde G_L$.  Using
$\nabla^{\eta,\eta^{i,i+e_k}}\tilde G_L=\nabla^{i,i+e_k}G(\cdot/L)$ one easily establishes
\begin{multline*}
  2\Bigl(\sqrt{\rho(\eta)}-\sqrt{\rho(\eta^{i,i+e_k})}\;\!\Bigr)^2\ge
  4\Bigl(\sqrt{\rho(\eta)}-\sqrt{\rho(\eta^{i,i+e_k})}\Bigr)H_{\eta,\eta^{i,i+e_k}}
  - 2 H_{\eta,\eta^{i,i+e_k}}^2\\
  = \bigl(\rho(\eta)-\rho(\eta^{i,i+e_k})\bigr)\nabla^{i,i+e_k}G(\cdot/L)
  -\frac 18 \bigl(\sqrt{\rho(\eta)}+\sqrt{\rho(\eta^{i,i+e_k})}\bigr)^2(\nabla^{i,i+e_k}G(\cdot/L))^2.
\end{multline*}
Using $q_{\eta,\eta'}=q_{\eta',\eta}$, the inequality $\tfrac 12(x+y)^2\le x^2+y^2$, and
$\mu(\eta)r^V_{\eta,\eta'}=\rho(\eta)q_{\eta,\eta'}$ thus allows to bound
$2q_{\eta,\eta^{i,i+e_k}}(\sqrt{\rho(\eta)}-\sqrt{\rho(\eta^{i,i+e_k})})^2$ from below by
\begin{multline*}
  \qquad\qquad\qquad\bigl(\mu(\eta)r^V_{\eta,\eta^{i,i+e_k}}-\mu(\eta^{i,i+e_k})r^V_{\eta^{i,i+e_k},\eta}\bigr)\nabla^{i,i+e_k}G(\cdot/L)\\
  -\frac 14\bigl(\mu(\eta)r^V_{\eta,\eta^{i,i+e_k}}+\mu(\eta^{i,i+e_k})r^V_{\eta^{i,i+e_k},\eta}\bigr)(\nabla^{i,i+e_k}G(\cdot/L))^2.
  \qquad\qquad\qquad
\end{multline*}
Note that $\sum_{\eta\in\Omega_L} \mu(\eta)r^V_{\eta,\eta^{i+e_k,i}}=\sum_{\eta\in\Omega_L}
\mu(\eta^{i,i+e_k})r^V_{\eta^{i,i+e_k},\eta}$ implies that
\begin{multline}
  \label{eqn:Psi_star_Fs_lower_bound_alternative}
  \Psi^\star_L\bigl(\mu,F^V(\mu)\bigr)\ge
  \sum_{i\in\mathbb T_L^d}\sum_{k=1}^d \biggl[\Bigl(\sum_{\eta\in\Omega_L} \mu(\eta)\bigl(r^V_{\eta,\eta^{i,i+e_k}}
    -r^V_{\eta,\eta^{i+e_k,i}}\bigr)\Bigr)\nabla^{i,i+e_k} G(\cdot/L)\\
    -\frac 14\Bigl(\sum_{\eta\in\Omega_L} \mu(\eta)\bigl(r^V_{\eta,\eta^{i,i+e_k}}+r^V_{\eta,\eta^{i+e_k,i}}\bigr)\Bigr)
    \bigl(\nabla^{i,i+e_k}G(\cdot/L)\bigr)^2\biggr],
\end{multline}
which coincides by~\eqref{eqn:density_current} and~\eqref{eqn:hat_chi} with the right hand side
of~\eqref{eqn:Psi_star_Fs_lower_bound}.
\end{proof}

\subsubsection{Asymptotic Lower Bound for the Free Energy}
\label{sec:Asympt-Lower-Bound}

\begin{proposition}
  \label{prop:free_energy_lower_bound_general}
  Let the assumptions of Theorem~\ref{thm:lower-bound-thm} hold and let $t\in[0,T]$ be such that the path $(\pi_t)_{t\in[0,T]}$ is
  continuous at $t$. Then
  \begin{equation}
    \label{eqn:liminf_general_t}
    \liminf_{L\to\infty}\frac 1{L^d}\mathcal F_{L,\alpha}^V\bigl(\mu^L_t\bigr)\ge \mathcal F_\alpha^V(\rho_t).
  \end{equation}
\end{proposition}

\begin{proof}
For each $h\in C(\Lambda;\mathbb R)$ the entropy inequality (a special case of the Fenchel inequality, see Proposition~8.1 and
page~340 in Appendix~1 in~\cite{Kipnis1999a}) implies
\begin{equation*}
  \begin{split}
    \frac 1{L^d}\mathcal F_{L,\alpha}^V\bigl(\mu^L_t\bigr) &\ge \frac 1{L^d}\biggr[\sum_{\eta\in\Omega_L} \mu^L_t(\eta)
      \sum_{i\in\mathbb T_L^d}h(i/L)\eta(i)-\log\biggl(\sum_{\eta\in\Omega_L}\nu_\alpha^V(\eta)\mathrm e^{\sum_{i\in\mathbb
          T_L^d}h(i/L)\eta(i)}\biggr)\biggr]\\ &=\sum_{\eta\in\Omega_L} \mu^L_t(\eta) \langle \Theta_L(\eta),h\rangle -\frac
    1{L^d}\sum_{i\in\mathbb T_L^d}\log \biggl(\frac{E_{\nu_{\alpha,1}}[ \mathrm e^{(h(i/L)-V(i/L))\eta(0)}]}{E_{\nu_{\alpha,1}}[
        \mathrm e^{-V(i/L)\eta(0)}]}\biggr).
  \end{split}
\end{equation*}
By the assumption of finite moments in~\eqref{eqn:finite_moments} the dominated convergence theorem yields that $u\mapsto
E_{\nu_{\alpha,1}}[\mathrm e^{(h(u)-V(u))\eta(0)}]$ is continuous.

By~\eqref{eqn:particle_bound}, we can restrict to measures with bounded volume, such that a truncation argument, combined with the
weak convergence $Q^L\to Q^*=\delta_{(\pi_t)_{t\in[0,T]}}$ and the continuity of the projection/evaluation at time $t$ implies $
\sum_{\eta\in\Omega_L} \mu^L_t(\eta) \langle \Theta_L(\eta),h\rangle = \mathbb E_{Q_L}[\langle \pi_t,h\rangle]\to \mathbb
E_{Q_L}[\langle \pi_t,h\rangle] = \langle \pi_t,h\rangle$.  Thus
\begin{equation}
  \liminf_{L\to\infty}\frac 1{L^d}\mathcal F_{L,\alpha}^V\bigl(\mu^L_t\bigr)
  \ge \langle \pi_t, h\rangle - \int_\Lambda \log \biggl(\frac{E_{\nu_{\alpha,1}}[ \mathrm e^{(h(u)-V(u))\eta(0)}]}{E_{\nu_{\alpha,1}}
    [ \mathrm e^{-V(u)\eta(0)}]}\biggr)\;\!\mathrm du.
\end{equation}
Taking the supremum with respect to $h\in C(\Lambda;\mathbb R)$ combined with~\eqref{eqn:2nd_free_energy} then finishes the proof.
\end{proof}

\subsubsection{Asymptotic Lower Bound for $\Psi$}
\label{sec:asymp_lower_bound_psi}

The following proofs will depend on uniform continuity of functions (which follows here from continuity and the compactness of the
domain $\Lambda$ (or $[0,T]\times\Lambda$)).

\begin{lemma}
  Under the assumptions of Theorem~\ref{thm:lower-bound-thm}, we have for any $G\in C^{1,2}([0,T]\times\Lambda;\mathbb R)$
  \begin{multline}
    \label{eqn:replace_empirical_1}
    \qquad\limsup_{\epsilon\to 0}\limsup_{L\to\infty}\biggl|\int_0^T \frac 1{L^d}\sum_{i\in\mathbb T_L^d}\sum_{k=1}^d
    \hat\chi_{i,i+e_k}(\mu^L_t) \bigl[L\nabla^{i,i+e_k}G_t(\cdot / L)\bigr]^2\;\!\mathrm dt\\ -\int_0^T\int_\Lambda
    \sum_{\eta\in\Omega_L}\mu^L_t(\eta) \chi\bigl([\Theta_L(\eta)*\iota_\epsilon](u)\bigr) |\nabla G_t(u)|^2\;\!\mathrm du\;\!\mathrm
    dt\biggr|=0.\qquad
  \end{multline}
\end{lemma}

\begin{proof}
We first show that without loss of generality we can set $V=0$ for the rates~\eqref{eqn:rates_V}. We denote with $\hat\chi^V$ the
mobility for a smooth potential $V$ and with $\hat\chi^0$ the mobility for $V=0$. Note that
\begin{multline}
  \label{eqn:some_other_bound}
  \biggl|\int_0^T\frac 1{L^d}\sum_{i\in\mathbb T_L^d}\sum_{k=1}^d \Bigl(\hat\chi_{i,i+e_k}^V(\mu^L_t)
  -\hat\chi_{i,i+e_k}^0(\mu^L_t) \Bigr)\bigl[L\nabla^{i,i+e_k}G_t(\cdot / L)\bigr]^2\mathrm dt\biggr|\\
  \le \int_0^T\frac 1{L^d}\sum_{i\in\mathbb T_L^d}\sum_{k=1}^d \hat\chi_{i,i+e_k}^0(\mu^L_t)
  2\bigl(\cosh\bigl(\tfrac 12\nabla^{i,i+e_k}V(\cdot/L)\bigr)-1\bigr)\bigl[L\nabla^{i,i+e_k}G_t(\cdot / L)\bigr]^2\mathrm dt.
\end{multline}
Taylor's theorem enables us to find for each $t\in[0,T]$ a number $\xi\in(i/L,(i\!+\!e_k)/L)$ for which
$L\nabla^{i,i+e_k}G_t(\cdot / L) = \partial_k G_t(\xi)$.  Defining $C_G:=\sum_{k=1}^d\sup_{t\in[0,T]}\|\partial_k
G_t\|_\infty^2<\infty$ allows us to bound the right hand side of~\eqref{eqn:some_other_bound} from above by
\begin{equation}
  \label{eqn:some_bound}
  \frac {2 C_G T}{L^d}\sum_{i\in\mathbb T_L^d}\sum_{k=1}^d \bigl(\cosh\bigl(\tfrac 12\nabla^{i,i+e_k}V(\cdot/L)\bigr)-1\bigr)
  \hat\chi_{i,i+e_k}^0\biggl(\frac 1T\int_0^T\mu^L_t\;\!\mathrm dt\biggr).
\end{equation}
Using the uniform continuity of $V$ (on the compact set $\Lambda$), we obtain for each $\epsilon>0$ that
$|\nabla^{i,i+e_k}V(\cdot/L)|<\epsilon$ as $L\to\infty$ independent of $i$ and $e_k$, such that~\eqref{eqn:some_bound} is (for $L$
large enough) with~\eqref{eqn:C_hat} bounded by $2C_G C_{\hat\chi} T (\cosh(\epsilon/2)-1)$.  Thus, taking the limit superior
$\epsilon\to 0$ after taking $L\to\infty$ in~\eqref{eqn:some_bound} shows that the left hand side of~\eqref{eqn:some_other_bound}
vanishes. This justifies the replacement of $V$ with $V=0$ in the mobility. We thus drop the indices $V$ and $0$ and simply write
$\hat\chi$ for the mobility with $V=0$.

To prove~\eqref{eqn:replace_empirical_1} it is sufficient to show that
\begin{equation}
  \label{eqn:needed_bounds}
  \begin{split}
&    \limsup_{\epsilon\to 0}\limsup_{L\to\infty}\\
    &\biggl|\int_0^T\frac 1{L^d}\sum_{i\in\mathbb T_L^d}\sum_{k=1}^d \frac{\hat\chi_{i,i+e_k}(\mu^L_t)}{(2\lfloor \epsilon
      L\rfloor\!+\!1)^d} \sum_{|m|\le \lfloor \epsilon L\rfloor}\Bigl(\bigl[L\nabla^{i,i+e_k}G_t(\cdot / L)\bigr]^2
    -\bigl[\partial_kG_t((i\!+\!m)/ L)\bigr]^2\Bigr)\;\!\mathrm dt\biggr|
    \\
    +&\;\!\frac {C_G T}{L^d}\sum_{i\in\mathbb T_L^d}\sum_{k=1}^d \;\!\sum_{\eta\in\Omega_L}
    \biggl(\frac 1T\int_0^T\mu^L_t(\eta)\mathrm dt\biggr)\Bigl|\hat\chi_{i,i+e_k}^{\lfloor \epsilon L\rfloor}(\delta_\eta)
    - \hat\chi_{i,i+e_k}\bigl(\nu_{\eta^{\lfloor \epsilon L\rfloor}(i)}\bigr)\Bigr|
    \\
    +&\;\!\frac {C_G T }{L^d}\sum_{i\in\mathbb T_L^d}\sum_{\eta\in\Omega_L}\biggl(\frac 1T\int_0^T\mu^L_t(\eta)\mathrm dt\biggr)
    \Bigl|\chi\bigl(\eta^{\lfloor \epsilon L\rfloor}(i)\bigr)-\chi\Bigl(\Bigl(\frac{2\epsilon L}{2{\lfloor \epsilon L\rfloor}+1}\Bigr)^d
    \eta^{\lfloor \epsilon L\rfloor}(i)\Bigr)\Bigr|
    \\
    +&\;\!\biggl|\int_0^T\frac 1{L^d}\sum_{i\in\mathbb T_L^d}\sum_{\eta\in\Omega_L}\mu^L_t(\eta) \chi\bigl([\Theta_L(\eta)*\iota_\epsilon](i/L)
    \bigr) |\nabla G_t(i/ L)|^2\;\!\mathrm dt
    \\
    &\;\!\qquad\qquad\qquad-\int_0^T \int_\Lambda \sum_{\eta\in\Omega_L}\mu^L_t(\eta) \chi\bigl([\Theta_L(\eta)*\iota_\epsilon](u)\bigr)
    |\nabla G_t(u)|^2\;\!\mathrm du\;\!\mathrm dt\biggr|=0.
  \end{split}
\end{equation}

By uniform continuity of $(\partial_k G_t)^2$ for each $\delta>0$ there exists an $\epsilon>0$ such that $|u-u'|<\epsilon$ implies
that $|(\partial_k G_t(u))^2-(\partial_k G_t(u'))^2|<\delta$ uniformly in $t\in[0,T]$. Thus, by~\eqref{eqn:C_hat}, the first term
in~\eqref{eqn:needed_bounds} is, for $\epsilon$ small enough, bounded by
\begin{multline*}
  \int_0^T\frac 1{L^d}\sum_{i\in\mathbb T_L^d}\sum_{k=1}^d \frac{\hat\chi_{i,i+e_k}(\mu^L_t)}{(2\lfloor \epsilon L\rfloor\!+\!1)^d}
  \sum_{|m|\le \lfloor \epsilon L\rfloor}\Bigl|\bigl[L\nabla^{i,i+e_k}G_t(\cdot / L)\bigr]^2
  -\bigl[\partial_kG_t((i\!+\!m)/ L)\bigr]^2\Bigr|\;\!\mathrm dt\\
  \le T\delta C_{\hat\chi}.
\end{multline*}
Letting $\delta\to 0$ shows that the first term in~\eqref{eqn:needed_bounds} vanishes.

The second term is controlled by the local equilibrium assumption~\eqref{eqn:local_equilibrium}; the third term vanishes using the
Lipschitz continuity of $\chi$ and the bound on the expected number of particles: The Lipschitz continuity yields that the third
summand in~\eqref{eqn:needed_bounds} is bounded by
\begin{equation*}
  C_G C_{\rm Lip}T\biggl|1-\Bigl(\frac{2\epsilon L}{2{\lfloor \epsilon L\rfloor}+1}\Bigr)^d \biggr|\sum_{\eta\in\Omega_L}
  \biggl(\frac 1T \int_0^T\mu^L_t(\eta)\;\!\mathrm dt\biggr)\;\!\frac 1{L^d}\sum_{i\in\mathbb T_L^d}\eta^{\lfloor \epsilon L\rfloor}(i).
\end{equation*}
By the conservation of particles, the last expression can be bounded by $C_G C_{\rm Lip} C_{\rm tot} T\;\!
\bigl|1-\bigl(\frac{2\epsilon L}{2{\lfloor \epsilon L\rfloor}+1}\bigr)^d \bigr|$, which vanishes as $L\to\infty$.

For the last term in~\eqref{eqn:needed_bounds} recall that $[\Theta_L(\eta)*\iota_\epsilon](u)$ is piecewise constant on
$\{[\frac{2i-1}{2L},\frac{2i+1}{2L})^d\}_{i\in\mathbb T_L^d}$ (cf.~\eqref{eqn:conv_empirical_const}). The proof thus reduces to
  establishing a bound for
  \begin{equation*}
    \int_0^T\sum_{i\in\mathbb T_L^d}\sum_{\eta\in\Omega_L}\mu^L_t(\eta) \chi\bigl([\Theta_L(\eta)*\iota_\epsilon](i/L)\bigr)
    \biggl|\int_{[\frac{2i-1}{2L},\frac{2i+1}{2L})^d}\bigl(|\nabla G_t(i/ L)|^2 - |\nabla G_t(u)|^2\bigr)\;\!\mathrm du \biggr|
      \;\!\mathrm dt,
\end{equation*}
which is easily obtained, as the the last expression is by the Lipschitz continuity,~\eqref{eqn:particle_bound},
and~\eqref{eqn:conv_empirical_const} bounded above by
\begin{equation*}
  C_{\rm Lip}C_{\rm tot} (2\epsilon)^{-d}\int_0^T \sum_{i\in\mathbb T_L^d}\int_{[\frac{2i-1}{2L},\frac{2i+1}{2L})^d} \bigl|\nabla
    G_t(i/ L)|^2 - |\nabla G_t(u)|^2\bigr|\;\!\mathrm du \;\!\mathrm dt,
\end{equation*}
which converges by the uniform continuity of $\nabla G$ to zero for $L\to\infty$.
\end{proof}
Note that the above proof does not depend on the fact that we consider the square gradient of a function $G$. We can replace the
square by the product of two different gradients and immediately obtain the following results.

\begin{lemma}
  \label{lem:replace_empirical}
  Under the assumptions of Theorem~\ref{thm:lower-bound-thm} we have for any $G,H\in C^1([0,T]\times\Lambda;\mathbb R)$ that
  \begin{multline}
    \limsup_{\epsilon\to 0}\limsup_{L\to\infty}\biggl|\int_0^T \frac 1{L^d}\sum_{i\in\mathbb T_L^d}\sum_{k=1}^d
    \hat\chi_{i,i+e_k}(\mu^L_t) \bigl[L\nabla^{i,i+e_k}H_t(\cdot / L)\bigr] \bigl[L\nabla^{i,i+e_k}G_t(\cdot / L)\bigr]\;\!\mathrm
    dt\\ -\int_0^T\int_\Lambda \sum_{\eta\in\Omega_L}\mu^L_t(\eta) \chi\bigl([\Theta_L(\eta)*\iota_\epsilon](u)\bigr) \nabla
    H_t(u)\cdot \nabla G_t(u)\;\!\mathrm du\;\!\mathrm dt\biggr|=0.
  \end{multline}
\end{lemma}

\begin{corollary}
  \label{cor:asymp_upper_bound} 
  Under the assumptions of Theorem~\ref{thm:lower-bound-thm} we have for any $G\in C^{1,2}([0,T]\times\Lambda;\mathbb R)$ that
  \begin{multline}
    \label{eqn:cor_1}
    \qquad\limsup_{\epsilon\to 0}\limsup_{L\to\infty}\biggl|\int_0^T \frac 1{L^d}\sum_{i\in\mathbb T_L^d}\sum_{k=1}^d
    \hat\chi_{i,i+e_k}(\mu^L_t) \bigl[2L\sinh\bigl(\tfrac 12\nabla^{i,i+e_k}G_t(\cdot / L)\bigr)\bigr]^2\;\!\mathrm
    dt\\ -\int_0^T\int_\Lambda \sum_{\eta\in\Omega_L}\mu^L_t(\eta) \chi\bigl([\Theta_L(\eta)*\iota_\epsilon](u)\bigr) |\nabla
    G_t(u)|^2\;\!\mathrm du\;\!\mathrm dt\biggr|=0\qquad
  \end{multline}
  and for any $G,H\in C^{1,2}([0,T]\times\Lambda;\mathbb R)$
  \begin{multline}
    \label{eqn:cor_2}
    \limsup_{\epsilon\to 0}\limsup_{L\to\infty}\\
    \biggl|\int_0^T \frac 1{L^d}\sum_{i\in\mathbb T_L^d}\sum_{k=1}^d
    \hat\chi_{i,i+e_k}(\mu^L_t) \bigl[2L\sinh\bigl(\tfrac 12\nabla^{i,i+e_k}G_t(\cdot /
      L)\bigr)\bigr]\bigl[L\nabla^{i,i+e_k}H_t(\cdot / L)\bigr]\;\!\mathrm dt\\ -\int_0^T\int_\Lambda
    \sum_{\eta\in\Omega_L}\mu^L_t(\eta) \chi\bigl([\Theta_L(\eta)*\iota_\epsilon](u)\bigr) \nabla G_t(u)\cdot \nabla
    H_t(u)\;\!\mathrm du\;\!\mathrm dt\biggr|=0.\quad
  \end{multline}
\end{corollary}

We now turn to the proof of the lower bound in~\eqref{eqn:lower_bound_psi}.
\begin{proposition}
  \label{prop:proof_psi_bound}
  Let the assumptions of Theorem~\ref{thm:lower-bound-thm} hold. Then~\eqref{eqn:lower_bound_psi} is satisfied.
\end{proposition}

\begin{proof}
For any $G\in C^{1,2}([0,T]\times\Lambda;\mathbb R)$ we have
\begin{multline}
  \label{eqn:starting_point}
  \quad \sum_{\eta\in\Omega_L} \tilde G_L(T,\eta) \mu^L_T(\eta) -\sum_{\eta\in\Omega_L} \tilde G_L(0,\eta) \mu^L_0(\eta) -
  \int_0^T \sum_{\eta\in\Omega_L} \partial_t \tilde G_L(t,\eta)\mu^L_t(\eta)\;\!\mathrm dt\\ = \int_0^T \langle \jmath^L_t, \nabla
  \tilde G_L(t,\cdot)\rangle_L\;\!\mathrm dt \le \int_0^T\Psi_L(\mu^L_t, \jmath^L_t)\;\!\mathrm dt+ \int_0^T\Psi^\star_L(\mu^L_t,
  \nabla \tilde G_L(t,\cdot))\;\!\mathrm dt.\quad
\end{multline}
Combined with Proposition~\ref{prop:upper_bound_finite} we obtain that $\frac 1{L^d} \int_0^T\Psi_L(\mu^L_t,\jmath^L_t)\;\!\mathrm
dt$ is bounded below by
\begin{multline}
  \label{eqn:starting_point_2}
  \quad\sum_{\eta\in\Omega_L} \mu^L_T(\eta)\langle\Theta_L(\eta),G_T\rangle -\sum_{\eta\in\Omega_L}
  \mu^L_0(\eta)\langle\Theta_L(\eta),G_0\rangle - \int_0^T \sum_{\eta\in\Omega_L} \mu^L_t(\eta) \langle\Theta_L(\eta),\partial_t
  G_t\rangle\;\!\mathrm dt\\ - \frac 1{2L^d}\int_0^T \sum_{i\in\mathbb T_L^d}\sum_{k=1}^d \hat \chi^0_{i,i+e_k}(\mu^L_t) \bigl[2L
    \sinh\bigl(\tfrac 12 \nabla^{i,i+e_k}G_t(\cdot/L)\bigr)\bigr]^2\;\!\mathrm dt.\quad
\end{multline}

For $\epsilon>0$ and $G$ fixed we define the function $f^{\epsilon,G}\colon\mathcal D([0,T];\mathcal M_+(\Lambda))\to\mathbb R$
which assigns to a path $(\tilde\pi_t)_{t\in[0,T]}$ the value
\begin{multline*}
  f^{\epsilon,G}((\tilde\pi_t)_{t\in[0,T]}) := \langle \tilde\pi_T,G_T\rangle -\langle \tilde\pi_0,G_0\rangle - \int_0^T\langle
  \tilde\pi_t,\partial_t G_t\rangle\;\!\mathrm dt \\
  - \frac 12 \int_0^T\int_\Lambda
  \chi\bigl([\tilde\pi_t*\iota_\epsilon](u)\bigr)|\nabla G_t(u)|^2\;\!\mathrm du\;\!\mathrm dt.
\end{multline*}
By~\eqref{eqn:particle_bound}, we can restrict $f^{\epsilon,G}$ to measures with bounded volume. In this case $f^{\epsilon,G}$ is
continuous and bounded, which follows from dominated convergence using the estimate
$\chi\bigl([\pi_t*\iota_\epsilon](u)\bigr)|\nabla G_t(u)|^2 \le C_G C_{\rm Lip}C_{\rm tot}/(2\epsilon)^d<\infty$.  We can
rewrite~\eqref{eqn:starting_point_2} as
\begin{multline*}
  \qquad\mathbb E_{Q_L}\bigl[f^{\epsilon,G}\bigr]+\frac 12\int_0^T \int_\Lambda \sum_{\eta\in\Omega_L}\mu^L_t(\eta)
  \chi\bigl([\Theta_L(\eta)*\iota_\epsilon](u)\bigr) |\nabla G_t|^2\;\!\mathrm du\;\!\mathrm dt\\ -\frac 1{2L^d}\int_0^T
  \sum_{i\in\mathbb T_L^d}\sum_{k=1}^d \hat \chi_{i,i+e_k}(\mu^L_t) \bigl[2L \sinh\bigl(\tfrac 12
    \nabla^{i,i+e_k}G_t(\cdot/L)\bigr)\bigr]^2\;\!\mathrm dt\qquad
\end{multline*}
and define the remainder
\begin{multline*}
  \qquad R_L^\epsilon:=\frac 12\biggl|\int_0^T \frac 1{L^d}\sum_{i\in\mathbb T_L^d}\sum_{k=1}^d \hat\chi_{i,i+e_k}(\mu_t)
  \bigl[2L\sinh\bigl(\tfrac 12\nabla^{i,i+e_k}G_t(\cdot / L)\bigr)\bigr]^2\;\!\mathrm dt\\ -\int_0^T\int_\Lambda
  \sum_{\eta\in\Omega_L}\mu_t(\eta) \chi\bigl([\Theta_L(\eta)*\iota_\epsilon](u)\bigr) |\nabla G_t|^2\;\!\mathrm du\;\!\mathrm
  dt\biggr|\qquad\qquad
\end{multline*}
to obtain $L^{-d} \int_0^T\Psi_L(\mu_t, j_t)\;\!\mathrm dt\ge \mathbb E_{Q_L}\bigl[f^{\epsilon,G}\bigr] - R_L^\epsilon$.

Since $f^{\epsilon,G}$ is continuous and bounded, the weak convergence $Q_L\to Q^*=\delta_{(\pi_t)_{t\in[0,T]}}$ implies that
$\lim_{L\to\infty}\mathbb E_{Q_L}\bigl[f^{\epsilon,G}\bigr] = \mathbb
E_{Q^*}\bigl[f^{\epsilon,G}\bigr]=f^{\epsilon,G}((\pi_t)_{t\in[0,T]})$.  Furthermore $\limsup_{\epsilon\to
  0}\limsup_{L\to\infty}R_L^\epsilon=0$ by Corollary~\ref{cor:asymp_upper_bound}. Thus $\liminf_{L\to\infty} L^{-d}
\int_0^T\Psi_L(\mu_t, j_t)\;\!\mathrm dt\ge \liminf_{\epsilon\to0}f^{\epsilon,G}((\pi_t)_{t\in[0,T]})$.

For $\pi_t(\mathrm du) = \rho_t(u)\;\!\mathrm du$ the distance
$|f^{\epsilon,G}((\pi_t)_{t\in[0,T]})-f^{0,G}((\pi_t)_{t\in[0,T]})|$ is bounded from above by
\begin{equation}
  \label{eqn:chi_rho_conv}
  \frac {C_G}{2}\int_0^T \int_\Lambda\Bigl|\chi\bigl([\rho_t*\iota_\epsilon](u)\bigr)-\chi(\rho_t(u))\Bigr|\mathrm du \;\!\mathrm
  dt \le \frac {C_GC_{\rm Lip}}{2}\int_0^T \int_\Lambda\bigl|[\rho_t*\iota_\epsilon](u)-\rho_t(u)\bigr|\mathrm du \;\!\mathrm dt,
\end{equation}
which is integrable. The dominated convergence theorem then implies that $f^{\epsilon,G}((\pi_t)_{t\in[0,T]}) \to
f^{0,G}((\pi_t)_{t\in[0,T]})$, which proves $\liminf_{L\to\infty} L^{-d} \int_0^T\Psi_L(\mu_t, j_t)\;\!\mathrm dt\ge
f^{0,G}((\pi_t)_{t\in[0,T]})$.  Taking the supremum over all $G\in C^{1,2}([0,T]\!\times\! \Lambda;\mathbb R)$ finally
yields~\eqref{eqn:lower_bound_psi}.
\end{proof}

\subsubsection{Asymptotic Lower Bound for $\Psi^\star$}

The proofs in this section are very similar to the proofs in Section~\ref{sec:asymp_lower_bound_psi}. We will therefore be brief.
\begin{lemma}
  \label{lem:2nd_remainder}
  Suppose the assumptions of Theorem~\ref{thm:lower-bound-thm} hold. Then
  \begin{multline}
    \label{eqn:2nd_remainder}
%\limsup_{\epsilon\to 0}\limsup_{L\to\infty}\\
%\biggl|\int_0^T \biggl(\frac 1{L^d}\sum_{i\in\mathbb T_L^d}\sum_{k=1}^d \Bigl[\bigl(L\hat\jmath_{i,i+e_k}^V(\mu^L_t)\bigr)( L\nabla^{i,i+e_k} G_t(\cdot / L)) -\frac 12 \hat\chi^V_{i,i+e_k}(\mu^L_t) \bigl[L\nabla^{i,i+e_k}G_t(\cdot / L)\bigr]^2\Bigr]\\
%-\mathbb E_{Q_L}\biggl[\int_\Lambda \phi\bigl([\pi_t*\iota_\epsilon](u)\bigr)\Delta G_t\;\!\mathrm du - \int_\Lambda \chi\bigl([\pi_t*\iota_\epsilon](u)\bigr)\nabla V\cdot \nabla G_t\;\!\mathrm du\\
    %-\frac 12\int_\Lambda \chi\bigl([\pi_t*\iota_\epsilon](u)\bigr)|\nabla G_t|^2\;\!\mathrm du\biggr]\biggr)\mathrm dt\;\!\biggr|=0.
    \limsup_{\epsilon\to 0}\limsup_{L\to\infty}\\ \biggl|\int_0^T \biggl(\frac 1{L^d}\sum_{i\in\mathbb T_L^d}\sum_{k=1}^d
    \Bigl[\bigl(L\hat\jmath_{i,i+e_k}^V(\mu^L_t)\bigr)( L\nabla^{i,i+e_k} G_t(\cdot / L)) -\frac 12 \hat\chi^V_{i,i+e_k}(\mu^L_t)
      \bigl[L\nabla^{i,i+e_k}G_t(\cdot / L)\bigr]^2\Bigr]\\ -\mathbb E_{Q_L}\biggl[\int_\Lambda
      \phi\bigl([\pi_t*\iota_\epsilon](u)\bigr)\Delta G_t\;\!\mathrm du - \int_\Lambda
      \chi\bigl([\pi_t*\iota_\epsilon](u)\bigr)\nabla V\cdot \nabla G_t\;\!\mathrm du \\
      -\frac 12\int_\Lambda
      \chi\bigl([\pi_t*\iota_\epsilon](u)\bigr)|\nabla G_t|^2\;\!\mathrm du\biggr]\biggr)\mathrm dt\;\!  \biggr|=0.
  \end{multline}
\end{lemma}

\begin{proof}
Note that
\begin{equation}
  \label{eqn:micro_current_j^V}
  \hat\jmath^V_{i,i+e_k}(\mu) 
  = \hat\jmath^{\;\! 0}_{i,i+e_k}(\mu)\cosh\bigl(\tfrac 12 \nabla^{i,i+e_k}V(\cdot/L)\bigr) 
  + \hat\chi_{i,i+e_k}^0(\mu)2\sinh\bigl(-\tfrac 12 \nabla^{i,i+e_k}V(\cdot/L)\bigr).
\end{equation}
Using~\eqref{eqn:gradient_j0} and~\eqref{eqn:micro_current_j^V}, a discrete integration by parts (i.e.~a shift of the index)
yields
\begin{equation*}
  \begin{split}
    &\sum_{i\in\mathbb T_L^d}\sum_{k=1}^d \bigl(L\hat\jmath_{i,i+e_k}^V(\mu)\bigr)( L\nabla^{i,i+e_k} G_t(\cdot / L))-\frac 12
    \hat\chi_{i,i+e_k}(\mu) \bigl[L\nabla^{i,i+e_k}G_t(\cdot / L)\bigr]^2\\ 
    =&\sum_{i\in\mathbb T_L^d}\sum_{k=1}^d\hat
    \phi_i(\mu)L^2\Bigl[\cosh\bigl(\tfrac 12 \nabla^{i,i+e_k}V(\cdot/L)\bigr)\nabla^{i,i+e_k} G_t(\cdot / L)\\
   &\qquad\qquad\qquad\qquad\qquad\qquad\qquad\qquad
-\cosh\bigl(\tfrac 12
      \nabla^{i-e_k,i}V(\cdot/L)\bigr)\nabla^{i-e_k,i} G_t(\cdot / L)\Bigr]\\ &\quad +
    \hat\chi_{i,i+e_k}^0(\mu)2L\sinh\bigl(-\tfrac 12 \nabla^{i,i+e_k}V(\cdot/L)\bigr)( L\nabla^{i,i+e_k} G_t(\cdot / L))\\
   &\qquad\qquad\qquad\qquad\qquad\qquad\qquad\qquad -\frac 12
    \hat\chi_{i,i+e_k}(\mu) \bigl[L\nabla^{i,i+e_k}G_t(\cdot / L)\bigr]^2.
  \end{split}
\end{equation*}
Combining this with the expression in~\eqref{eqn:2nd_remainder}, it is sufficient to show that
\begin{multline}
  \label{eqn:1st_part}
  \limsup_{\epsilon\to 0}\limsup_{L\to\infty} \biggl|\int_0^T \frac 1{L^d}\sum_{i\in\mathbb T_L^d}\sum_{k=1}^d\hat
  \phi_i(\mu^L_t)L^2\Bigl[\cosh\bigl(\tfrac 12 \nabla^{i,i+e_k}V(\cdot/L)\bigr)\nabla^{i,i+e_k} G_t(\cdot /
    L)\\ -\cosh\bigl(\tfrac 12 \nabla^{i-e_k,i}V(\cdot/L)\bigr)\nabla^{i-e_k,i} G_t(\cdot / L)\Bigr] -\mathbb
  E_{Q_L}\biggl[\int_\Lambda \phi\bigl([\pi_t*\iota_\epsilon](u)\bigr)\Delta G_t(u)\;\!\mathrm du\biggr]\;\!\mathrm dt \biggr|=0,
\end{multline}
as well as 
\begin{multline}
  \label{eqn:2nd_part}
  \limsup_{\epsilon\to 0}\limsup_{L\to\infty}\biggl|\int_0^T \frac 1{L^d} \sum_{i\in\mathbb
    T_L^d}\sum_{k=1}^d\hat\chi_{i,i+e_k}^0(\mu^L_t)2L\sinh\bigl(\tfrac 12 \nabla^{i,i+e_k}V(\cdot/L)\bigr)( L\nabla^{i,i+e_k}
  G_t(\cdot / L))\\ -\mathbb E_{Q_L}\biggl[\int_\Lambda \chi\bigl([\pi_t*\iota_\epsilon](u)\bigr)\nabla V(u)\cdot \nabla
    G_t(u)\;\!\mathrm du\biggr]\;\!\mathrm dt\biggr|\\
  +\frac 12\biggl|\int_0^T \frac 1{L^d} \sum_{i\in\mathbb T_L^d}\sum_{k=1}^d \hat\chi_{i,i+e_k}(\mu^L_t)
  \bigl[L\nabla^{i,i+e_k}G_t(\cdot / L)\bigr]^2\\
   -\mathbb E_{Q_L}\biggl[\int_\Lambda
    \chi\bigl([\pi_t*\iota_\epsilon](u)\bigr)|\nabla G_t(u)|^2\;\!\mathrm du\biggr]\;\!\mathrm dt\biggr|=0.
\end{multline}
Note that~\eqref{eqn:2nd_part} follows from the above considerations (Lemma~\ref{lem:replace_empirical} and
Corollary~\ref{cor:asymp_upper_bound}), such that we are only left to prove~\eqref{eqn:1st_part}, which can be proven with the
same calculations as above (with $\hat\chi$ replaced by $\hat \phi$ combined with~\eqref{eqn:C_hat} and
using~\eqref{eqn:local_equilibrium_2} instead of~\eqref{eqn:local_equilibrium}).
\end{proof}

\begin{proposition}
  \label{prop:another_ineq}
  Under the assumptions of Theorem~\ref{thm:lower-bound-thm} the inequality~\eqref{eqn:lower_bound_psi_star} holds.
\end{proposition}

\begin{proof}
We only sketch the proof, which is very similar to the one of Proposition~\ref{prop:proof_psi_bound}. For
\begin{multline*}
  f^{\epsilon,G}((\tilde\pi_t)_{t\in[0,T]}) := \int_0^T\int_\Lambda \phi\bigl([\tilde\pi_t*\iota_\epsilon](u)\bigr)\Delta
  G_t\;\!\mathrm du\;\!\mathrm dt \\
  - \int_0^T\int_\Lambda \chi\bigl([\tilde\pi_t*\iota_\epsilon](u)\bigr)\nabla V\cdot \nabla
  G_t\;\!\mathrm du\;\!\mathrm dt -\frac 12 \int_0^T\int_\Lambda \chi\bigl([\tilde\pi_t*\iota_\epsilon](u)\bigr)|\nabla
  G_t|^2\;\!\mathrm du\;\!\mathrm dt.
\end{multline*}
Proposition~\ref{prop:lower_bound_finite} implies that
\begin{multline*}
  \Psi^\star_L\bigl(\mu,F^V(\mu)\bigr) \\
  \ge \sum_{i\in\mathbb T_L^d}\sum_{k=1}^d \Bigl[\bigl(L\hat\jmath_{i,i+e_k}^V(\mu)\bigr)(
    L\nabla^{i,i+e_k} G(\cdot / L)) -\frac 12 \hat\chi^V_{i,i+e_k}(\mu) \bigl[L\nabla^{i,i+e_k}G(\cdot / L)\bigr]^2\Bigr].
\end{multline*}
As in the proof of Proposition~\ref{prop:proof_psi_bound}, one obtains $\frac 1{L^d}\int_0^T \Psi^\star_L(\mu^L_t,F^S(\mu^L_t))
\;\!\mathrm dt \ge \mathbb E_{Q_L}\bigl[f^{\epsilon,G}\bigr] -R_L^\epsilon$, where $R_L^\epsilon$ coincides
with~\eqref{eqn:2nd_remainder} in Lemma~\ref{lem:2nd_remainder}.
%where this time
%\begin{multline}\nonumber
%\qquad\qquad R_L^\epsilon := \biggl|\int_0^T \frac 1{L^d}\sum_{i\in\mathbb T_L^d}\sum_{k=1}^d \Bigl[\bigl(L\hat\jmath_{i,i+e_k}^V(\mu^L_t)\bigr)( L\nabla^{i,i+e_k} G_t(\cdot / L))\\
%-\frac 12 \hat\chi^V_{i,i+e_k}(\mu^L_t) \bigl[L\nabla^{i,i+e_k}G_t(\cdot / L)\bigr]^2\Bigr]\;\!\mathrm dt -\mathbb E_{Q_L}\bigl[f^{\epsilon,G}\bigr]\biggr|.\qquad\qquad
%\end{multline}
%Lemma~\ref{lem:2nd_remainder} 
The latter implies that $\limsup_{\epsilon\to 0}\limsup_{L\to\infty}R_L^\epsilon=0$, such that again by weak convergence with
$\epsilon\to 0$
\begin{equation*}
  \liminf_{L\to\infty}
  \frac 1{L^d}\int_0^T \Psi^\star_L(\mu^L_t,F^S(\mu^L_t)) \;\!\mathrm dt
  \ge  f^{0,G}((\pi_t)_{t\in[0,T]}).
\end{equation*}
Taking the supremum with respect to $G\in C^{1,2}([0,T]\times \Lambda;\mathbb R)$ yields~\eqref{eqn:lower_bound_psi_star}.
\end{proof}

\subsection{Proof of Theorem~\ref{thm:Q_H_limit}}
\label{sec:theorem_2}

\begin{proofof}{Theorem~\ref{thm:Q_H_limit}}
We extend the proof in~\cite{Benois1995a}. We will skip some details, as they are similar to the above calculations. Let $\tilde
H\in C^{1,2}([0,T]\times \Lambda;\mathbb R)$. The log density of $P^{V\!+\!\tilde H}_L$ with respect to $P_L^V$ (where both
measures have the same initial condition $\mu^L_0$) has the explicit representation (cf.~\cite{Benois1995a} and the Appendix
in~\cite{Kaiser2018a})
\begin{multline*}
  \quad \log \frac{dP^{V\!+\!\tilde H}_L}{dP_L^V}((\eta_t)_{t\in[0,T]})
  =\frac {L^d}2 \biggl[\langle \Theta_L(\eta_T),\tilde H_T\rangle  -\langle  \Theta_L(\eta_0),\tilde H_0\rangle
    - \int_0^T\langle \Theta_L(\eta_t),\partial_t \tilde H_t\rangle\;\!\mathrm dt\biggr]\\
  -\int_0^T \sum_{i\in\mathbb T_L^d}\sum_{i': |i-i'|=1} \hat r^V_{\eta_t,\eta_t^{i,i'}}
  L^2\bigl(\mathrm e^{-\frac 12(\tilde H_t(i'/L)-\tilde H_t(i/L))}-1\bigr)\mathrm dt.\quad
\end{multline*}
Using $2(ac + bd) = (a-b)(c-d) + (a+b)(c+d)$ we can represent the expression in the last line as 
\begin{multline*}
  \qquad\qquad\int_0^T\sum_{i\in\mathbb T_L^d}\sum_{k=1}^d \Bigl[ L\bigl(\hat r^V_{\eta_t,\eta_t^{i,i+e_k}}-\hat
    r^V_{\eta_t,\eta_t^{i+e_k,i}}\bigr) \bigl(L \sinh\bigl(-\tfrac 12 \nabla^{i,i+e_k}\tilde H_t(\tfrac \cdot L)\bigr)\bigr) \\ +
    \bigl(\hat r^V_{\eta_t,\eta_t^{i,i+e_k}}+\hat r^V_{\eta_t,\eta_t^{i+e_k,i}}\bigr) L^2 \bigl(\cosh\bigl(-\tfrac 12
    \nabla^{i,i+e_k}\tilde H_t(\tfrac \cdot L)\bigr)-1\bigr) \Bigr]\mathrm dt.\qquad\qquad
\end{multline*}
Taking the expected value of this expression with respect to $P_L^V$, in combined with~\eqref{eqn:density_current}
and~\eqref{eqn:hat_chi}, yields
\begin{multline}
  \qquad\qquad\int_0^T\sum_{i\in\mathbb T_L^d}\sum_{k=1}^d \Bigl[\bigl(L\hat\jmath^V_{i,i+e_k}(\mu^L_t)\bigr)
    \bigl(L\sinh\bigl(-\tfrac 12 \nabla^{i,i+e_k}\tilde H_t(\tfrac \cdot L)\bigr)\bigr) \\ + 2\hat\chi^V_{i,i+e_k}(\mu^L_t) L^2
    \bigl(\cosh\bigl(\tfrac 12 \nabla^{i,i+e_k}\tilde H_t(\tfrac \cdot L)\bigr)-1\bigr) \Bigr]\mathrm dt,\qquad\qquad
\end{multline}
which is asymptotically equivalent to
\begin{equation*}
  \int_0^T\frac 12\sum_{i\in\mathbb T_L^d}\sum_{k=1}^d \Bigl[ -\bigl(L\hat\jmath^V_{i,i+e_k}(\mu^L_t)\bigr)
    \bigl(L\nabla^{i,i+e_k}\tilde H_t(\tfrac \cdot L)\bigr) + \frac 12\hat\chi^V_{i,i+e_k}(\mu^L_t) L^2
    \bigl|\nabla^{i,i+e_k}\tilde H_t(\tfrac \cdot L)\bigr|^2 \Bigr]\mathrm dt.
\end{equation*}
A result similar to Lemma~\ref{lem:2nd_remainder} yields 
\begin{equation*}
  \lim_{L\to\infty}\frac 1{L^d} \mathbb A_L^V\bigl(Q^{V\!+\!\tilde H}_L\bigr) = \lim_{\epsilon\to 0}\lim_{L\to\infty}\frac
  12\mathbb E_{Q_L}\bigl[f^{\epsilon,\tilde H}\bigr] = \frac 12 f^{0,\tilde H}((\pi_t)_{t\in[0,T]}),
\end{equation*}
where the functional $f^{\epsilon,\tilde H}$ is given by
\begin{multline*}
  f^{\epsilon,\tilde H}((\pi_t)_{t\in[0,T]}) := \langle \pi_T,\tilde H_T\rangle -\langle \pi_0,\tilde H_0\rangle - \int_0^T\langle
  \pi_t,\partial_t \tilde H_t\rangle\;\!\mathrm dt \\
  -\int_0^T\int_\Lambda \phi\bigl([\pi_t*\iota_\epsilon](u)\bigr)\Delta \tilde
  H_t\;\!\mathrm du\;\!\mathrm dt \\ + \int_0^T\int_\Lambda \chi\bigl([\pi_t*\iota_\epsilon](u)\bigr)\nabla V\cdot \nabla \tilde
  H_t\;\!\mathrm du\;\!\mathrm dt - \frac 12 \int_0^T\int_\Lambda \chi\bigl([\pi_t*\iota_\epsilon](u)\bigr)|\nabla \tilde
  H_t|^2\;\!\mathrm du\;\!\mathrm dt.
\end{multline*}
Finally, since the hydrodynamic path $(\pi_t)_{t\in[0,T]}$ solves $\dot\rho_t =\Delta\phi(\rho_t) + \nabla\cdot(\chi(\rho_t)\nabla
(V+\tilde H_t))$, we obtain
\begin{equation*}
  f^{0,\tilde H}((\pi_t)_{t\in[0,T]}) = \frac12 \int_0^T \| \tilde H_t \|_{1,\chi(\rho_t)}^2 = \frac12 \int_0^T \| \dot\rho_t -
  \Delta\phi(\rho_t) - \nabla\cdot(\chi(\rho_t)\nabla V) \|_{-1,\chi(\rho_t)}^2.
\end{equation*}
\end{proofof}

\paragraph{Acknowledgements}
We are grateful for stimulating discussions with Federico Cornalba, Max Fathi and Andr\'e Schlichting. Further, we would like to
thank Mark~A.~Peletier for valuable suggestions. MK is supported by a scholarship from the EPSRC Centre for Doctoral Training in
Statistical Applied Mathematics at Bath (SAMBa), under the project EP/L015684/1.  JZ gratefully acknowledges funding by the EPSRC
through project EP/K027743/1, the Leverhulme Trust (RPG-2013-261) and a Royal Society Wolfson Research Merit Award.

%\small
%\setstretch{1}

%\bibliography{jz} %% scaling_limits}{}
%\bibliographystyle{plain}

\def\cprime{$'$} \def\cprime{$'$} \def\cprime{$'$}
  \def\polhk#1{\setbox0=\hbox{#1}{\ooalign{\hidewidth
  \lower1.5ex\hbox{`}\hidewidth\crcr\unhbox0}}} \def\cprime{$'$}
  \def\cprime{$'$}

\end{document}